%% file: neurips_2026.tex
\documentclass{article}

 \usepackage[main, final]{neurips_2026}

\usepackage[utf8]{inputenc} 
\usepackage[T1]{fontenc}    
\usepackage[hypertexnames=false]{hyperref}
\usepackage{url}            
\usepackage{booktabs}       
\usepackage{amsfonts}       
\usepackage{nicefrac}       
\usepackage{microtype}      
\usepackage{xcolor}         
\usepackage{thm-restate}
\usepackage{graphicx}
\usepackage{subcaption}
\usepackage{amsmath}
\usepackage{amssymb}
\usepackage{mathtools}
\usepackage{amsthm}

\usepackage{tikz}
\usetikzlibrary{arrows.meta,positioning}

\allowdisplaybreaks

\usepackage[capitalize,noabbrev]{cleveref}

\input{macros}

\DeclareSymbolFont{extraup}{U}{zavm}{m}{n}
\DeclareMathSymbol{\varheart}{\mathalpha}{extraup}{86}
\DeclareMathSymbol{\vardiamond}{\mathalpha}{extraup}{87}

\hypersetup{
  colorlinks   = true, 
  urlcolor     = blue, 
  linkcolor    = {red!75!black}, 
  citecolor   = {blue!50!black} 
}

\title{Do We Need Asynchronous SGD? On the Near-Optimality of Synchronous Methods}

\author{%
  Grigory Begunov \\
  AXXX, Moscow, Russia \\
  MSU, Moscow, Russia
  \And
  Alexander Tyurin \\
  AXXX, Moscow, Russia \\
  Applied AI Institute, Moscow, Russia
}

\begin{document}

\maketitle

\begin{abstract}
  Modern distributed optimization methods mostly rely on traditional \emph{synchronous} approaches, despite substantial recent progress in \emph{asynchronous} optimization. We revisit Synchronous SGD and its robust variant, called $m$-Synchronous SGD, which can be interpreted as a method with a backup-worker strategy, and \emph{theoretically} show that they are nearly optimal in many heterogeneous computation and communication scenarios, which is somewhat unexpected. We analyze the synchronous methods under random times and adversarial partial participation of workers, and prove that their time complexities are optimal in many practical regimes, up to logarithmic factors. While synchronous methods are not universal solutions and there exist tasks where asynchronous methods may be necessary, we show theoretically that they are sufficient for many modern heterogeneous computation and communication scenarios.
\end{abstract}

\section{Introduction}
\label{sec:introduction}
We consider the stochastic optimization problem of minimizing a function:
\begin{align*}
  \min\limits_{x \in \R^d} f(x),
\end{align*}
where $f\,:\, \R^d \to \R$ is $L$--smooth function (Assumption \ref{ass:lipschitz_constant}) and bounded (Assumption~\ref{ass:lower_bound}). There are $n$ workers who can only calculate stochastic gradients $\nabla f(x; \xi)$ at any point $x \in \R^d$ with $\xi \sim \mathcal{D}$, where $\mathcal{D}$ is some distribution, such that $\nabla f(x; \xi)$ is an unbiased estimate of $\nabla f(x)$ and has $\sigma^2$--bounded variance (Assumption \ref{ass:stochastic_variance_bounded}). Our goal in the nonconvex setting is to find a (potentially random) vector $\bar{x} \in \R^d$ such that $\mathbb{E}[\norm{\nabla f(\bar{x})}^2] \leq \varepsilon$, i.e., an $\varepsilon$--stationary point.
This is a standard optimization problem in machine learning, federated learning, training deep neural networks, and large language models \citep{konevcny2016federated,mcmahan2016federated}. 

We focus on the setting in which multiple workers (e.g., GPUs, servers) work together to find an $\varepsilon$--stationary point, 
and we consider realistic scenarios where the computation times of stochastic gradients are not fixed and may be different for each worker, random due to natural computation fluctuations and spontaneous outages, and potentially varying over time. In Section~\ref{sec:comm}, we extend the results to scenarios where communication also takes time. In this work, we assume vector aggregation times and system-dependent delays are negligible compared with the computation and communication times of stochastic gradients.
We consider the standard assumptions:
\begin{assumption}
    \label{ass:lipschitz_constant}
    Function $f$ is $L$--smooth:
    $\norm{\nabla f(x) - \nabla f(y)} \leq L \norm{x - y}$ for all $x, y \in \R^d.$
 \end{assumption}
 \begin{assumption}
    \label{ass:lower_bound}
    There exist $f^* \in \R$ such that $f(x) \geq f^*$ for all $x \in \R^d$. 
    We define $\Delta \eqdef f(x^0) - f^*,$ where $x^0$ is the starting point of optimization methods.
 \end{assumption}
\begin{assumption}
    \label{ass:stochastic_variance_bounded}
    The stochastic gradients $\nabla f(x; \xi)$ are unbiased and have $\sigma^2$--bounded variance: $\ExpSub{\xi}{\nabla f(x;\xi)} = \nabla f(x)$ and ${\mathbb{E}}_{\xi}[\|\nabla f(x;\xi) - \nabla f(x)\|^2] \leq \sigma^2$ for all $x \in \R^d,$ where $\sigma > 0$ is some constant.
\end{assumption} 
\begin{figure*}[t]
\centering
\begin{minipage}[t]{0.49\textwidth}
\begin{figure}[H]
\centering
\resizebox{0.85\linewidth}{!}{%
\begin{tikzpicture}[x=0.16cm,y=0.9cm, font=\small]

  \draw[-{Stealth[length=2.2mm]}] (0,-0.1) -- (36,-0.1) node[right] {time};

  \def\tOneStart{0}
  \def\tOneEnd{20}

  \def\tTwoStart{20}
  \def\tTwoEnd{40}

  \def\tOneWOne{8}
  \def\tOneWTwo{11}
  \def\tOneWThree{20}

  \def\tTwoWOne{28}
  \def\tTwoWTwo{32}
  \def\tTwoWThree{40}

  \def\BubbleLabel{bubble}
  \def\SyncLabel{sync}

  \tikzset{
    comp/.style={draw, rounded corners=1.5pt, minimum height=6mm, fill=black!10},
    bubble/.style={draw, rounded corners=1.5pt, minimum height=6mm, fill=white, dashed},
    sync/.style={draw, thick},
    msg/.style={-Latex, thick},
    upd/.style={circle, fill=black, draw=black, minimum size=1.2mm, inner sep=0pt}
  }

  \node[anchor=east] at (0,3.2) {Worker 1};
  \node[anchor=east] at (0,2.2) {Worker 2};
  \node[anchor=east] at (0,1.2) {Worker 3};

  \draw[comp]   (\tOneStart,3) rectangle (\tOneWOne,3.35) node[midway] {$\tau_{1,1}$};
  \draw[bubble] (\tOneWOne,3)  rectangle (\tOneEnd,3.35) node[midway] {\BubbleLabel};

  \draw[comp]   (\tOneStart,2) rectangle (\tOneWTwo,2.35) node[midway] {$\tau_{1,2}$};
  \draw[bubble] (\tOneWTwo,2)  rectangle (\tOneEnd,2.35) node[midway] {\BubbleLabel};

  \draw[comp]   (\tOneStart,1) rectangle (\tOneWThree,1.35) node[midway] {$\tau_{1,3}$};


  \draw[sync] (\tOneEnd,0.6) -- (\tOneEnd,3.3);

  \draw[comp]   (\tTwoStart,3) rectangle (\tTwoWOne,3.35) node[midway] {$\tau_{2,1}$};
  \draw[bubble] (\tTwoWOne,3)  rectangle (\tTwoEnd,3.35) node[midway] {\BubbleLabel};

  \draw[comp]   (\tTwoStart,2) rectangle (\tTwoWTwo,2.35) node[midway] {$\tau_{2,2}$};
  \draw[bubble] (\tTwoWTwo,2)  rectangle (\tTwoEnd,2.35) node[midway] {\BubbleLabel};

  \draw[comp]   (\tTwoStart,1) rectangle (\tTwoWThree,1.35) node[midway] {$\tau_{2,3}$};

  \draw[sync] (\tTwoEnd,0.6) -- (\tTwoEnd,3.3);

  \draw[thick] (\tOneStart,0.35) -- (\tTwoEnd,0.35);

  \foreach \t/\y in {\tOneWThree/1.18}{
    \draw[msg] (\t,\y) -- (\t,0.52);
    \node[upd] (u1) at (\t,0.35) {};
    \node[anchor=north, xshift=2pt, yshift=4pt] at (u1.south) {$x^1$};
  }

  \foreach \t/\y in {\tTwoWThree/1.18}{
    \draw[msg] (\t,\y) -- (\t,0.52);
    \node[upd] (u2) at (\t,0.35) {};
    \node[anchor=north, xshift=2pt, yshift=4pt] at (u2.south) {$x^2$};
  }

\end{tikzpicture}%
}
\caption{Wall-clock time visualization of Synchronous SGD (Algorithm~\ref{alg:alg_orig} or \ref{alg:alg_server_m_star}).}
\label{fig:three_workers_bubbles}
\end{figure}
\end{minipage}\hfill
\begin{minipage}[t]{0.49\textwidth}
\begin{figure}[H]
\centering
\resizebox{0.85\linewidth}{!}{%
\begin{tikzpicture}[x=0.16cm,y=0.9cm, font=\small]

  \draw[-{Stealth[length=2.2mm]}] (0,-0.1) -- (35,-0.1) node[right] {time};

  \def\tEnd{40}

  \def\dWOneA{8}   \def\dWOneB{8}
  \def\dWTwoA{11}  \def\dWTwoB{12}
  \def\dWThreeA{20}\def\dWThreeB{20}

  \def\AsyncLabel{async update}
  \def\ServerLabel{Server}

  \tikzset{
    comp/.style={draw, rounded corners=1.5pt, minimum height=6mm, fill=black!10},
    serverline/.style={draw, thick},
    msg/.style={-Latex, thick},
    upd/.style={fill=black, draw=black},
    tick/.style={draw=black!60, dashed, line width=0.35pt},
  }

  \node[anchor=east] at (0,3.2) {Worker 1};
  \node[anchor=east] at (0,2.2) {Worker 2};
  \node[anchor=east] at (0,1.2) {Worker 3};

  \draw[serverline] (0,0.35) -- (\tEnd,0.35);

  \draw[comp] (0,3) rectangle (\dWOneA,3.35) node[midway] {$\tau_{1,1}$};
  \draw[comp] (\dWOneA,3) rectangle ({\dWOneA+\dWOneB},3.35) node[midway] {$\tau_{2,1}$};
  \draw[comp] ({\dWOneA+\dWOneB},3) rectangle ({\dWOneA+2*\dWOneB},3.35) node[midway] {$\tau_{3,1}$};
  \draw[comp] ({\dWOneA+2*\dWOneB},3) rectangle ({\dWOneA+3*\dWOneB},3.35) node[midway] {$\tau_{4,1}$};
  \draw[comp] ({\dWOneA+3*\dWOneB},3) rectangle (\tEnd,3.35) node[midway] {$\tau_{5,1}$};

  \foreach \t/\lab in {8/$x^1$,16/$x^3$,24/$x^6$,32/$x^7$,40/$x^9$}{
    \draw[msg] (\t,3.18) -- (\t,0.52);
    \node[upd, circle, inner sep=1.2pt] (u\t) at (\t,0.35) {};
    \node[anchor=north, xshift=2pt, yshift=4pt] at (u\t.south) {\lab};
  }

  \draw[comp] (0,2) rectangle (\dWTwoA,2.35) node[midway] {$\tau_{1,2}$};
  \draw[comp] (\dWTwoA,2) rectangle ({\dWTwoA+\dWTwoB},2.35) node[midway] {$\tau_{2,2}$};
  \draw[comp] ({\dWTwoA+\dWTwoB},2) rectangle ({\dWTwoA+\dWTwoB+\dWTwoA},2.35) node[midway] {$\tau_{3,2}$};
  \draw[comp] ({\dWTwoA+\dWTwoB+\dWTwoA},2) rectangle (\tEnd,2.35) node[midway] {$\tau_{4,2}$};

  \foreach \t/\lab in {11/$x^2$,23/$x^5$,34/$x^8$}{
    \draw[msg] (\t,2.18) -- (\t,0.52);
    \node[upd, circle, inner sep=1.2pt] (u\t) at (\t,0.35) {};
    \node[anchor=north, xshift=2pt, yshift=4pt] at (u\t.south) {\lab};
  }

  \draw[comp] (0,1) rectangle (\dWThreeA,1.35) node[midway] {$\tau_{1,3}$};
  \draw[comp] (\dWThreeA,1) rectangle (\tEnd,1.35) node[midway] {$\tau_{2,3}$};

  \foreach \t/\lab in {20/$x^4$,40/$\,$}{
  \draw[msg] (\t,1.18) -- (\t,0.52);
  \node[upd, circle, inner sep=1.2pt] (u\t) at (\t,0.35) {};
  \node[anchor=north, xshift=2pt, yshift=4pt] at (u\t.south) {\lab};
  }


\end{tikzpicture}
}
\caption{Wall-clock time visualization of Asynchronous SGD (Algorithm~\ref{alg:asgd}).}
\label{fig:three_workers_async}
\end{figure}
\end{minipage}
\textbf{\small Surprisingly, this paper demonstrates that ``bubbles'' are not a fundamental obstacle, and that Synchronous SGD and $m$-Synchronous SGD can remain nearly optimal despite their presence.}
\end{figure*}

\label{sec:known_methods}
\textbf{Known methods in asynchronous optimization.} The most naive and simplest method in this setting is Synchronous SGD (Algorithm~\ref{alg:alg_orig})
where all workers calculate \emph{one} stochastic gradient each, aggregate them, and perform stochastic gradient descent.
The method has a clear disadvantage: in every iteration, it waits for the slowest worker, leading to ``bubbles'' in computation and idle workers (see Figure~\ref{fig:three_workers_bubbles}). The same problem applies to the synchronous ADAM and Muon methods\footnote{In this paper, we focus on the SGD update since it is theoretically optimal \citep{arjevani2022lower,tyurin2023optimal}; however, all our conclusions potentially apply to other updates.} \citep{kingma2014adam,jordan2024muon}.

In order to improve computational efficiency and avoid ``bubbles,'' several asynchronous methods have been proposed. One of the most celebrated is Asynchronous SGD \citep{tsitsiklis1986distributed,recht2011hogwild} (Algorithm~\ref{alg:asgd}).
The main idea is that, instead of waiting for all workers, Algorithm~\ref{alg:asgd} immediately updates the iterate once a stochastic gradient is received from any worker (see Figure~\ref{fig:three_workers_async}). The theoretical properties of it have been analyzed in \citep{lian2015asynchronous,feyzmahdavian2016asynchronous,stich2020error,sra2016adadelay,cohen2021asynchronous,koloskova2022sharper,mishchenko2022asynchronous,maranjyan2025ringmaster}.

Another method designed for asynchronous settings is Rennala SGD \citep{tyurin2023optimal}. The main idea is slightly different compared to Asynchronous SGD: instead of immediately updating the current iterate, Rennala SGD employs an asynchronous batch aggregation strategy to reduce ``bubbles'' and idle time. Additionally, there are many other asynchronous methods with state-of-the-art theoretical properties \citep{tyurin2025birch}.

\textbf{Time complexities.} Many synchronous and asynchronous methods have been developed. \emph{How can we theoretically compare them?} \citet{mishchenko2022asynchronous} proposed assuming that, for all $i \in [n],$ the $i$\textsuperscript{th} worker requires $\tau_i$ seconds to calculate a stochastic gradient. Without loss of generality (w.l.o.g.), we assume that $\tau_1 \leq \dots \leq \tau_n.$ 
Using this assumption, we can compare methods.
One can prove that the \emph{time complexity} of Synchronous SGD (Algorithm~\ref{alg:alg_orig}) is
\begin{align}
\label{eq:ELIDwwYnjnAnCj}
\textstyle \cO\left(\tau_n \max\left\{\frac{L \Delta}{\varepsilon}, \frac{\sigma^2 L \Delta}{n \varepsilon^2}\right\}\right)
\end{align}
seconds,
where the result follows from the fact that Synchronous SGD converges after $\max\left\{\nicefrac{L \Delta}{\varepsilon}, \nicefrac{\sigma^2 L \Delta}{n \varepsilon^2}\right\}$ iterations \citep{lan2020first}, and each iteration requires $\tau_n$ seconds, since the method waits for the slowest worker. At the same time, it was shown \citep{tyurin2023optimal,maranjyan2025ringmaster,tyurin2025birch} that the time complexity of Rennala SGD, Asynchronous SGD (Ringmaster ASGD), and other asynchronous methods is
\begin{equation}
\begin{aligned}
    \label{eq:optimal_const}
    &\textstyle T_{\textnormal{optimal}} \eqdef \Theta\left(\min \limits_{m \in [n]} \left[\left(\frac{1}{m}\sum\limits_{i=1}^m \frac{1}{\tau_i}\right)^{-1}\left(\max\left\{\frac{L \Delta}{\varepsilon}, \frac{\sigma^2 L \Delta}{m \varepsilon^2}\right\}\right)\right]\right),
\end{aligned}
\end{equation}
which is also \emph{optimal}. The time complexities \eqref{eq:ELIDwwYnjnAnCj} and \eqref{eq:optimal_const} formalize the intuition that Synchronous SGD is suboptimal and does not fully utilize the workers' time. In particular, \eqref{eq:ELIDwwYnjnAnCj} is never better than \eqref{eq:optimal_const}, but can be arbitrarily worse. For instance, take $\tau_n \to \infty,$ i.e., the last worker is extremely slow. In this case, $\eqref{eq:ELIDwwYnjnAnCj} \to \infty,$ while \eqref{eq:optimal_const} remains stable due to its harmonic dependence on $\{\tau_i\}.$
\begin{figure}[t]
\centering
\begin{minipage}[t]{0.47\linewidth}
\begin{algorithm}[H]
  \caption{Synchronous SGD}
  \label{alg:alg_orig}
  \begin{algorithmic}[1]
  \STATE \textbf{Input:} point $x^0 \in \R^d$, stepsize $\gamma > 0$
  \FOR{$k = 0, \dots, K - 1$}
  \STATE Workers calculate stoch. gradients at $x^k$
  \STATE $g^k = \frac{1}{n}\sum_{i=1}^{n} \nabla f(x^k; \xi_i)$ 
  \STATE Update $x^{k+1} = x^k - \gamma g^k$
  \ENDFOR
  \end{algorithmic}
\end{algorithm}
\end{minipage}
\hfill
\begin{minipage}[t]{0.52\linewidth}
\begin{algorithm}[H]
    \caption{Asynchronous SGD}
    \label{alg:asgd}
    \begin{algorithmic}[1]
        \STATE \textbf{Input:} point $x^0 \in \R^d$, stepsizes $\gamma_k \geq 0$
        \STATE Workers start computing stoch. gradients at $x^0$
        \FOR{$k = 0, \dots, K - 1$}
            \STATE $\nabla f(x^{k-\delta^k}; \xi^{k-\delta^k}_{i})$ arrives from worker $i$
            \STATE Update $x^{k+1} = x^{k} - \gamma_k \nabla f(x^{k-\delta^k}; \xi^{k-\delta^k}_{i})$
            \STATE Worker $i$ begins calculating $\nabla f(x^{k+1}; \xi^{k+1}_{i})$
        \ENDFOR
    \end{algorithmic}
\end{algorithm}
\end{minipage}
\vspace{-0.6cm}
\end{figure}

\textbf{Is Synchronous SGD fundamentally limited?} Indeed, Synchronous SGD (Algorithm~\ref{alg:alg_orig}) is not robust to large computation times of slow workers. This leads to our first research question: instead of designing and using asynchronous methods, would it be better to simply ignore slow workers and use Synchronous SGD with the fastest ones during optimization? To this end, consider $m$-Synchronous SGD (Algorithm~\ref{alg:alg_server_m_star}) with the only difference being that, at each iteration, it aggregates \emph{one stochastic gradient} from each of the first $m$ workers to finish, rather than from all $n$ workers:
\begingroup
\setlength{\textfloatsep}{3pt}
\setlength{\intextsep}{3pt}
\setlength{\floatsep}{0pt}
\begin{algorithm}[H]
  \caption{$m$-Synchronous SGD}
  \label{alg:alg_server_m_star}
  \begin{algorithmic}[1]
  \STATE \textbf{Input:} point $x^0 \in \R^d$, stepsize $\gamma$, \# of workers $m \in [n]$
  \FOR{$k = 0, 1, \dots, K - 1$}
  \STATE Each idle worker starts computing a stochastic gradient at $x^k$
  \STATE Aggregate\footnotemark the first $m$ workers at $x^k$: $\textstyle g^k = \frac{1}{m}\sum_{i \in S_k} \nabla f(x^k; \xi_i^k),\,\,S_k \subseteq [n], \,\, \abs{S_k} = m$
  \STATE Update $x^{k+1} = x^k - \gamma g^k$
  \STATE Discard any gradients for iteration $k$ that arrive after the update
  \ENDFOR
  \end{algorithmic}
\end{algorithm}
\endgroup
\footnotetext{Not necessarily from worker $1$ to $m$. It can be any of the first $m$ workers from the set $[n]$ who finish computing at $x^k$.}
Notice that $m$-Synchronous SGD is still synchronous and asks for only one stochastic gradient from each of the $m$ participating workers; it reduces to Synchronous SGD when $m = n$ and can be viewed as a robust version of it that ignores $n - m$ stragglers. Another important interpretation of this method is as a backup-worker strategy, which is widely used in practice instead of Asynchronous SGD \citep{chen2016revisiting}.

A related question is the following: while $\eqref{eq:ELIDwwYnjnAnCj} \to \infty$ as $\tau_n \to \infty$, in practice $\tau_n < \infty$, and it may be the case that the resulting gap is not too large. Even if $\tau_n \gg \tau_1,$ it might be possible that $\eqref{eq:ELIDwwYnjnAnCj}$ is optimal or nearly optimal\footnote{We say that a method is nearly optimal if it is optimal up to a constant or logarithmic factor.}. Moreover, even if the computation times are random, their variance may not be large enough to make Synchronous SGD (Algorithm~\ref{alg:alg_orig}) asymptotically suboptimal.

Synchronous methods are still the de facto and most widely used, there is empirical evidence that it is actually fast and yields better accuracy \citep{chen2016revisiting,megatron-lm,rajbhandari2020zero,grattafiori2024llama}. Can we provide a \emph{theoretical justification} for their time performance?

\subsection{Contributions}
We provide theoretical evidence that synchronous methods are nearly as fast as modern asynchronous methods in many heterogeneous computational scenarios. 

$\spadesuit$ We first show that, under Assumption~\ref{ass:fixed_computation_model}, $m$-Synchronous SGD (Algorithm~\ref{alg:alg_server_m_star}) is nearly optimal up to a $\log (n + 1)$ factor for arbitrary heterogeneous worker computation times, a somewhat unexpected result (see Section~\ref{sec:sync_optimal}). This also holds for $m$-Synchronous SGD with minibatching in the presence of communication times (see Section~\ref{sec:comm}). In particular, this is the first theoretical evidence that a synchronous method with a backup-worker strategy can be nearly as fast as asynchronous methods, a fact that was previously known only empirically \citep{chen2016revisiting}.

$\clubsuit$ We then extend this result to Assumption~\ref{ass:general_heter_random_computation_model}, where worker computation times are random, and show theoretically that $m$-Synchronous SGD (Algorithm~\ref{alg:alg_server_m_star}) remains nearly optimal in many practical settings (see also Section~\ref{sec:comm} for the case with communication times). To establish optimality, we derive a new lower bound (Theorem~\ref{thm:random_lower_bound}), which we believe is of independent interest (see Section~\ref{sec:random}).

$\vardiamond$ Even the standard Synchronous SGD method \emph{with all workers participating} is nearly optimal in some important regimes: (i) when computation times are random with equal means, for example, when they follow the exponential distribution (see Section~\ref{sec:random} and Corollary~\ref{cor:random}); and (ii) when the means are not equal but follow a power law $\tau_m = \tau_1 m^{\alpha} + \delta_m$, with $0 \le \alpha \le 1$ and bounded ${\delta_m}$. In this case, Synchronous SGD is still nearly optimal, for instance, for the truncated normal distribution, even though $\tau_n \gg \tau_1$ (see Section~\ref{sec:optimal_m} and Proposition~\ref{prop:tau}).

$\varheart$ We also consider scenarios in which computation trends are non-stationary and vary over time. One important result is that $m$-Synchronous SGD (Algorithm~\ref{alg:alg_server_m_star}) is optimal under arbitrary, and even adversarial, worker participation (Assumption~\ref{ass:general_heter_random_computation_model}) (see Section~\ref{sec:univ}).

We acknowledge in Section~\ref{sec:where_needed} that synchronous methods are not a silver bullet, and that there are several scenarios where asynchronous methods are still needed. Nevertheless, somewhat counterintuitively, we prove that Synchronous SGD and $m$-Synchronous SGD are nearly optimal in a wide range of heterogeneous computation settings, which may be sufficient in practice.

\section{Synchronous Methods Are Nearly Optimal under Fixed Computation}
\label{sec:sync_optimal}
In order to present our main result, we first recall the classical SGD result for the required number of iterations to converge:
\begin{theorem}[Iteration Complexity; \citep{lan2020first}]
    \label{thm:sync_sgd}
    Under Assumptions \ref{ass:lipschitz_constant}, \ref{ass:lower_bound}, and \ref{ass:stochastic_variance_bounded}, let $\gamma = \min\left\{ \nicefrac{1}{2 L}, \nicefrac{\varepsilon m}{4L\sigma^2} \right\}$ in Algorithm~\ref{alg:alg_server_m_star}. Then
    $\frac{1}{K}\sum_{k=0}^{K - 1} {\mathbb{E}}[\|\nabla f (x^k)\|^2] \le \varepsilon$ after
       $\textstyle K \eqdef \ceil{16 \max\left\{\nicefrac{L \Delta}{\varepsilon}, \nicefrac{\sigma^2 L \Delta}{m \varepsilon^2}\right\}}$
    iterations, where $m \in [n]$ is an arbitrary number.
\end{theorem}
Theorem~\ref{thm:sync_sgd} holds since Algorithm~\ref{alg:alg_server_m_star} is SGD with batch size $m.$ At the beginning, we consider the fixed computation model, discussed in Section~\ref{sec:known_methods}:
\begin{assumption}
\label{ass:fixed_computation_model}
For all $i \in [n],$ the $i$\textsuperscript{th} worker requires $\tau_i$ seconds to calculate a stochastic gradient. W.l.o.g., we assume that $\tau_1 \leq \dots \leq \tau_n.$
\end{assumption}
Recall the optimal time complexity $T_{\textnormal{optimal}}$ from \eqref{eq:optimal_const}. In the following theorem, we prove our first result: $m$-Synchronous SGD, with a proper choice of $m$, achieves the optimal time complexity up to a logarithmic factor.
\begin{restatable}[Time Complexity under Fixed Computation Model]{theorem}{TIMECOMPLEXITY}
    \label{thm:sync_sgd_time}
    Under the setting of Theorem~\ref{thm:sync_sgd}, additionally assume that Assumption~\ref{ass:fixed_computation_model} holds. Then, the time complexity of Algorithm~\ref{alg:alg_server_m_star} to find an $\varepsilon$--stationary point is
    \begin{equation}
    \begin{aligned}
        \label{eq:VDoAilJBJwVbJEFTc}
        \textstyle T_{\textnormal{sync}} \eqdef \frac{16 L \Delta}{\varepsilon} \min\limits_{m \in [n]} \left[\tau_m\max\left\{1, \frac{\sigma^2}{m \varepsilon}\right\}\right]
    \end{aligned}
    \end{equation}
    if $m$ is a minimizer of $\min\limits_{m \in [n]} \left[\tau_m \max\left\{1, \frac{\sigma^2}{m \varepsilon}\right\}\right].$ More importantly, the result is nearly optimal:
    \begin{align}
        \label{eq:EYaDmxXrTww}
        T_{\textnormal{sync}} = \cO\left(T_{\textnormal{optimal}} \times \log (n + 1)\right).
    \end{align}
\end{restatable}
To the best of our knowledge, this is the first result showing that a synchronous method is nearly as fast as the optimal asynchronous algorithms for stochastic optimization in the presence of heterogeneous computation times (Assumption \ref{ass:fixed_computation_model}). This result appears to have been overlooked in previous works \citep{cohen2021asynchronous,koloskova2022sharper,mishchenko2022asynchronous, tyurin2023optimal, maranjyan2025ringmaster}. 
It turns out that it is possible to design a nearly optimal \emph{synchronous} method under Assumption \ref{ass:fixed_computation_model}. Theorem~\ref{thm:best} shows a similar result with communication times.

\textbf{Proof sketch.} The first term in \eqref{eq:VDoAilJBJwVbJEFTc} is straightforward and follows from Theorem~\ref{thm:sync_sgd} and the fact that the method waits for the fastest $m$ workers. Then, we choose $m \in [n]$ that minimizes $\tau_m\max\left\{1, \nicefrac{\sigma^2}{m \varepsilon}\right\}.$ Under Assumption~\ref{ass:fixed_computation_model}, the fastest $m$ workers are those with indices $\{1,\dots,m\}$.

It is left to prove \eqref{eq:EYaDmxXrTww}. For simplicity, assume that $\nicefrac{\sigma^2}{\varepsilon} \geq n,$ then $m = \arg\min_{m \in [n]} \tau_m\max\left\{1, \nicefrac{\sigma^2}{m \varepsilon}\right\} = \arg\min_{m \in [n]} \nicefrac{\tau_m}{m}.$ Thus, $\nicefrac{\tau_m}{m} \leq \nicefrac{\tau_i}{i}$ and $\tau_i \geq \nicefrac{\tau_m i}{m}$ for all $i \in [n].$ Substituting the last inequality to \eqref{eq:optimal_const}, if $\nicefrac{\sigma^2}{\varepsilon} \geq n,$ then
\begin{align*}
    &\textstyle T_{\textnormal{optimal}} 
    = \Theta\left(\left(\frac{1}{n}\sum\limits_{i=1}^n \frac{1}{\tau_i}\right)^{-1}\frac{\sigma^2 L \Delta}{n \varepsilon^2}\right) \textstyle \geq \Omega\left(\frac{L \Delta}{\varepsilon} \left(\sum\limits_{i=1}^n \frac{1}{i}\right)^{-1} \frac{\tau_m \sigma^2}{m \varepsilon}\right) = \Omega\left(\frac{1}{\log(n + 1)}T_{\textnormal{sync}}\right)
\end{align*}
since $\sum_{i=1}^n \frac{1}{i} = \Theta(\log(n + 1)),$ and we get \eqref{eq:EYaDmxXrTww}. The full proof is provided in Section~\ref{sec:proof}. Moreover, the logarithmic term is unavoidable in the case where $\tau_i = i$ for all $i \in [n]$.

\textbf{Why is this result surprising?} Theorem~\ref{thm:sync_sgd_time} shows that we can achieve a nearly optimal rate using standard synchronous SGD with the fastest $m$ workers, \emph{which still compute one stochastic gradient per iteration}. Therefore, the ``bubble'' problem (Figure~\ref{fig:three_workers_bubbles}) is still present, since the fastest worker with computation time $\tau_1$ has to wait for slower workers $2$, $3$, and so on, up to worker $m$. 
Nevertheless, this strategy is nearly optimal. Intuitively, it is always possible to choose $m$ in Algorithm~\ref{alg:alg_server_m_star} such that the size of the ``bubbles'' and the idle time are not too large, and it is as fast as more advanced asynchronous methods. It is also interesting to note that even Synchronous SGD (Algorithm~\ref{alg:alg_orig}) can be nearly optimal despite $\tau_n \gg \tau_1$ (see Section~\ref{sec:optimal_m}).
\section{Random Computation Model}
\label{sec:random}
The obtained result holds not only for the fixed computation model (Assumption~\ref{ass:fixed_computation_model}). One way to generalize Assumption~\ref{ass:fixed_computation_model} is to assume that the computation times are random:
\begin{assumption}
\label{ass:general_heter_random_computation_model}
For all $i \in [n],$ the $i$\textsuperscript{th} worker requires $\bar{\tau}_i$ seconds to calculate a stochastic gradient, where $\bar{\tau}_i$ is a $(\tau_i, R)$--sub-exponential random variable, i.e., $\Exp{\exp(\abs{\bar{\tau}_i - \tau_i} / R)} \leq 2$, $R > 0,$ $\Exp{\bar{\tau}_i} = \tau_i,$ and $\bar{\tau}_i \geq 0$ a.s.. We assume that the sampled times $\{\bar{\tau}_i\}$ and the stochastic random variables $\{\xi_i\}$ are statistically independent. Moreover, $\tau_1 \leq \dots \leq \tau_n$ (w.l.o.g.).
\end{assumption}
This assumption includes a wide range of distributions, including truncated normal, any bounded distributions, all sub-Gaussian distributions, as well as Gamma and exponential distributions \citep{vershynin2018high}. The latter distribution is commonly used to model time delays in real computational systems (e.g., \citep{mitliagkas2016asynchrony, dutta2018slow}). 
\begin{restatable}[Time Complexity under Random Computation Model]{theorem}{TIMECOMPLEXITYRANDOM}
  \label{thm:sync_sgd_time_random}
  Under the setting of Theorem~\ref{thm:sync_sgd}, additionally assume that Assumption~\ref{ass:general_heter_random_computation_model} holds. Then, the time complexity of Algorithm~\ref{alg:alg_server_m_star} to find an $\varepsilon$--stationary point is at most
  \begin{align}
    \label{eq:jUeog}
    \textstyle T_{\textnormal{rand}} \eqdef \sum\limits_{k=0}^{K - 1} \max_{i \in [m]} \bar{\tau}_{k,i}
  \end{align}
  seconds, where $\bar{\tau}_{k,i}$ is a $(\tau_i, R)$--sub-exponential random variable (independent copy). Moreover, we have
  \begin{align}
    \label{eq:ccJgxrZmhVOoCPhU}
    \textstyle \Exp{T_{\textnormal{rand}}} = \cO\left(\frac{L \Delta}{\varepsilon}\left(\tau_m + R \log (n)\right)\max\left\{1, \frac{\sigma^2}{m \varepsilon}\right\}\right)
  \end{align}
  for all $m \in [n]$ in Algorithm~\ref{alg:alg_server_m_star}.
\end{restatable}
We obtain exactly the same result as in Theorem~\ref{thm:sync_sgd_time}, with the only difference being an additional ``noise'' term $R \log(n)$, and the result now holds on average.
After establishing the following lower bound, we show that the $R$ term does not dominate and that the result is nearly optimal with $m = \arg\min_{m \in [n]} \left[\tau_m \max\left\{1, \nicefrac{\sigma^2}{m \varepsilon}\right\}\right]$ in a number of relevant examples, both in theory and in practice.

\textit{Remark.} In this section, we implicitly assume that the workers are not allowed to stop calculations because modern frameworks do not allow this, and the only option is to ``kill'' the running computation process on the GPU and start the code again, which would lead to even larger delays. We refer to \citep{maranjyan2025mindflayer}, where the authors consider scenarios in which it is possible.

\textbf{Lower bound.} For Assumption~\ref{ass:general_heter_random_computation_model}, to the best of our knowledge, there is no lower bound in the literature. Thus, we prove a new lower bound to understand the optimality gap of \eqref{eq:ccJgxrZmhVOoCPhU}.
\begin{restatable}[Lower Bound under Random Computation Model]{theorem}{LOWERBOUND}
  \label{thm:random_lower_bound}
  For all $\Delta, L, \varepsilon, R, \sigma > 0,$ and time delays satisfying Assumption~\ref{ass:general_heter_random_computation_model} such that\footnote{The inequality is satisfied if, for instance, $\varepsilon$ is small enough, $R$ is small enough, or $\sigma$ is large enough.} $\tau_1 \max\left\{1, \nicefrac{\sigma^2}{n \varepsilon}\right\} \geq \bar{c} R \max\left\{1, \nicefrac{\sigma}{\sqrt{\varepsilon}}\right\} \log \left(n\right),$ where $\bar{c}$ is a universal constant, there exists a smooth function $f$ with $L$--Lipschitz gradients, and $f(0) - f^* \leq \Delta,$ and an oracle that returns unbiased stochastic gradients with $\sigma^2$--bounded gradient variance such that the expected time of any zero-respecting (see Definition~\ref{def:zerorespecting}) algorithm that does not stop computations is
  $\Omega\left(T_{\textnormal{optimal}}\right)$ seconds (from \eqref{eq:optimal_const}) to get an $\varepsilon$-stationary point.
\end{restatable}
\begin{corollary}
    \label{cor:random}
    Under the setting of Theorem~\ref{thm:sync_sgd_time_random}, $m$-Synchronous SGD (Algorithm~\ref{alg:alg_server_m_star}) is nearly optimal, on average, at least in the regimes where $\tau_1 \max\left\{1, \nicefrac{\sigma^2}{n \varepsilon}\right\} \geq \Omega(R \max\left\{1, \nicefrac{\sigma}{\sqrt{\varepsilon}}\right\} \log \left(n\right))$ and\footnotemark $\tau_1 \geq \tilde{\Omega}(R)$:
    \begin{align*}
        \Exp{T_{\textnormal{rand}}} = \cO\left(T_{\textnormal{optimal}} \times \log(n + 1)\right)
    \end{align*}
    if $m$ is a minimizer of $\min\limits_{m \in [n]} \left[\tau_m \max\left\{1, \frac{\sigma^2}{m \varepsilon}\right\}\right].$
\end{corollary}
\footnotetext{the inequalities hold, for example, when $\varepsilon$ is small and $R$ is small compared to $\tau_1$, e.g., when $\bar{\tau}_i \sim \mathcal{N}(\mu_i, \varsigma^2)$ truncated to $\bar{\tau}_i \ge 0$, or when ${\bar{\tau}_i}$ follow $\textnormal{Exp}(\lambda)$. See also Section~\ref{sec:more_examples}.}
Hence, even Synchronous SGD (Algorithm~\ref{alg:alg_orig}) is nearly optimal in the listed regimes if the computation times are random and the means $\{\tau_i\}$ are the same.

\textbf{Theoretical examples.} The standard examples are the truncated normal and exponential distributions. For instance, if $\bar{\tau}_i$ is from the truncated normal distribution, that is, $\bar{\tau}_i \sim \mathcal{N}(\mu_i, \varsigma^2)$ conditional on $\bar{\tau}_i \geq 0$ for any $\mu_i, \varsigma \geq 0$, then it satisfies Assumption~\ref{ass:general_heter_random_computation_model} with $R = \cO(\varsigma)$ \citep{barreto2025optimal}. Moreover, for the truncated normal distribution, $\tau_i = \Omega(\max\{\mu_i, \varsigma\}).$ Therefore, $R = \cO(\tau_i)$ for all $i \in [n],$ $R$ does not dominate in \eqref{eq:ccJgxrZmhVOoCPhU}, and the result \eqref{eq:ccJgxrZmhVOoCPhU} \emph{is nearly optimal} due to Corollary~\ref{cor:random} when, for instance, $\varepsilon$ is small.

Consider another example where $\bar{\tau}_i$ is from the exponential distribution $\textnormal{Exp}(\lambda)$ for all workers. In this case, $R = \Theta(1 / \lambda)$ and $\tau_i = \Exp{\bar{\tau}_i} = 1 / \lambda$ for all $i \in [n]$ \citep{vershynin2018high}. Therefore, $\eqref{eq:ccJgxrZmhVOoCPhU} = \tilde{\cO}\left(\nicefrac{\tau_1 L \Delta}{\varepsilon} \max\left\{1, \nicefrac{\sigma^2}{n \varepsilon}\right\}\right)$ for $m = n,$ and even though the computation times are random, 
\emph{the time complexity of Synchronous SGD (Algorithm~\ref{alg:alg_orig}) is nearly optimal for all $\lambda > 0,$ when, for instance, $\varepsilon$ is small.} See more examples in Section~\ref{sec:more_examples}.

\textbf{Practical simulations.} In practice, with NanoGPT \citep{nanogpt}, we also observe that $R \log n$ is negligibly small compared to the mean. See details in Section~\ref{sec:nanogpt}.

\section{On the Optimal Selection of Active Workers}
\label{sec:optimal_m}
In order to obtain nearly optimal time complexities for $m$-Synchronous SGD (Algorithm~\ref{alg:alg_server_m_star}) in Sections~\ref{sec:sync_optimal} and \ref{sec:random}, we have to minimize
\begin{align}
    \label{eq:mapping_g}
    \textstyle g(m) \eqdef \tau_m \max\left\{1, \frac{\sigma^2}{m \varepsilon}\right\}
\end{align}
for all $m \in [n],$ where $m$ is the number of active workers and $\{\tau_m\}$ is a non-decreasing sequence.
\begin{proposition}
\label{prop:g}
Consider the mapping $g(m)$ from \eqref{eq:mapping_g}. Then, $g(m)$ is minimized with $m \leq \min\left\{\ceil{\nicefrac{\sigma^2}{\varepsilon}}, n\right\}.$ If $\nicefrac{\sigma^2}{\varepsilon} \leq 1,$ then the optimal $m = 1.$ Otherwise, if $\nicefrac{\sigma^2}{\varepsilon} > 1,$ then $\frac{\sigma^2 h(m) }{\varepsilon} \leq g(m) \leq \frac{2 \sigma^2 h(m)}{\varepsilon}$ for all $m \leq \min\left\{\ceil{\nicefrac{\sigma^2}{\varepsilon}}, n\right\},$ where 
\begin{align}
    \label{eq:mQHhz}
    \textstyle h(m) \eqdef \frac{\tau_m}{m}.
\end{align}
\end{proposition}
Proposition~\ref{prop:g} says that it is necessary and sufficient to find a minimizer of $h(m)$ for all $m \leq \min\left\{\ceil{\nicefrac{\sigma^2}{\varepsilon}}, n\right\}$ to obtain the minimal value of $g(m)$ for all $m \in [n]$ (up to a factor of $2$). 
This proposition provides a simple choice of $m$ and identifies regimes where even Synchronous SGD is nearly optimal.
For instance, suppose that $\tau_{m+1} \leq \tau_m \left(1 + \nicefrac{1}{m}\right),$ meaning that $\tau_m$ is a non-decreasing sequence that increases by at most a factor of $1 + \nicefrac{1}{m}.$
One can easily show that the sequences $\tau_m = \tau_1 \times m$ (linear increase) and $\tau_m = \tau_1 \times m^{\alpha}$ (sublinear increase) satisfy this condition for any $\alpha \in [0, 1].$ In this case, since $h(m)$ is non-increasing, we can take $m = \min\left\{\ceil{\nicefrac{\sigma^2}{\varepsilon}}, n\right\}$ as a minimizer of $h(m),$ and use this $m$ in Algorithms~\ref{alg:alg_server_m_star}. 
We can generalize the previous result: 
\begin{proposition}
    \label{prop:tau}
    Assume that 
    \begin{align}
    \label{eq:ZBNvgXux}
    \tau_m = \tau_1 m^{\alpha} + \delta_m
    \end{align}
    for any $\alpha \in [0, 1]$ and $\delta_m$ such that $0 \leq \delta_m \leq \delta.$ Then, the choice $m = \min\left\{\ceil{\nicefrac{\sigma^2}{\varepsilon}}, n\right\}$ is optimal in \eqref{eq:mapping_g} when $\min\left\{\ceil{\nicefrac{\sigma^2}{\varepsilon}}, n\right\} \geq \left(\delta / \tau_1\right)^{1 / \alpha}.$ Thus, for $\tau_m$ satisfying \eqref{eq:ZBNvgXux} and large  enough $\nicefrac{\sigma^2}{\varepsilon},$ even Synchronous SGD (Algorithm~\ref{alg:alg_orig}) is nearly optimal in Corollary~\ref{cor:random} since $m = n.$
\end{proposition}
Thus, if $\{\tau_m\}$ follows the power law \eqref{eq:ZBNvgXux} with $0 \leq \alpha \leq 1$ and bounded $\{\delta_m\}$, then it is sufficient to take
$m = \min\left\{\ceil{\nicefrac{\sigma^2}{\varepsilon}},\, n\right\}$ when either $n$ or $\nicefrac{\sigma^2}{\varepsilon}$ is large, \emph{meaning that even Synchronous SGD (Algorithm~\ref{alg:alg_orig}) is nearly optimal in this case when $\nicefrac{\sigma^2}{\varepsilon} \geq n,$ despite the fact that $\tau_n \gg \tau_1.$}

Notice that the regime when $\tau_m \approx \tau_1 m^{\alpha}$ and $\alpha > 1$ is not very interesting since \eqref{eq:optimal_const} reduces to $\Theta\left(\tau_1 \max\left\{\nicefrac{L \Delta}{\varepsilon}, \nicefrac{\sigma^2 L \Delta}{m \varepsilon^2}\right\}\right),$ which is the time complexity of the fastest worker, and multiple workers are therefore not needed since $\{\tau_i\}$ grows too fast. 
\section{Universal Computation Model}
\label{sec:univ}
Here, we consider another way to generalize Assumption~\ref{ass:fixed_computation_model}:
\begin{assumption}[\citet{tyurin2024tighttimecomplexitiesparallel}]
    \label{ass:speed}
    For each worker $i \in [n]$, we relate a non-negative Riemann integrable (continuous almost everywhere) \emph{computation power} $v_i \,:\, \R_{+} \rightarrow \R_{+}$.
    The number of stochastic gradients that worker $i$ can compute between times $t_0$ and $t_1$ is
    \begin{align}
    \label{eq:rSIiSfVcmivSKsfzoSA}
    N_i(t_0,t_1) \eqdef \flr{\int_{t_0}^{t_1} v_i(\tau) d \tau}. 
    \end{align}
\end{assumption}
For instance, assume that the computation times of workers are fixed: $v_i(\tau) = v_i \in \R_+.$ In this case, $N_i(t_0,t_1) = \flr{v_i \times (t_1 - t_0)},$ and if worker $i$ starts calculating at time $t_0,$ then it will compute one gradient after $t_0 + \nicefrac{1}{v_i}$ seconds because $N_i(t_0, t_0 + \nicefrac{1}{v_i}) = 1,$ two gradients after $t_0 + \nicefrac{2}{v_i}$ seconds, and so forth. In this example, Assumption~\ref{ass:speed} reduces to Assumption~\ref{ass:fixed_computation_model} with $\tau_i \equiv \nicefrac{1}{v_i}.$ 

However, Assumption~\ref{ass:speed} allows us to use virtually any computation power, capturing virtually all possible computational scenarios. For instance, it is possible that $v_i(t) > 0$ for all $t \leq \bar{t}$ and $v_i(t) = 0$ for all $t > \bar{t}$, meaning that worker $i$ is not available after time $\bar{t}$. Moreover, $v_i$ can have arbitrary fluctuations (see Figures~\ref{fig:power_1} and \ref{fig:power_2}), and, unlike Assumption~\ref{ass:general_heter_random_computation_model}, it supports non-stationary trends. See more examples in Sections~\ref{sec:theoretical_examples}, \ref{sec:strag}, and \ref{sec:numerical_examples}. \citet{tyurin2024tighttimecomplexitiesparallel} proved the following lower bound.
\begin{theorem}[Informal Lower Bound under Universal Computation Model; \citet{tyurin2024tighttimecomplexitiesparallel}]
  \label{thm:univ_lower_bound}
  Let Assumptions~\ref{ass:lipschitz_constant}, \ref{ass:lower_bound}, \ref{ass:stochastic_variance_bounded}, and \ref{ass:speed} hold. With high probability, it is impossible to find an $\varepsilon$--stationary point faster than $\nicefrac{1}{2} \times \underline{t}_{\underline{K}}$ seconds, where the sequence $\{\underline{t}_k\}$ is defined as 
  \begin{align}
    \textstyle \underline{t}_k = \min\left\{t \, : \, \sum\limits_{i=1}^n N_i(\underline{t}_{k-1},t) \geq c_2 \ceil{\frac{\sigma^2}{\varepsilon}}\right\},
  \label{eq:lower_bound_seq}
  \end{align}
  where $\underline{t}_{0} = 0,$ $\underline{K} \eqdef \ceil{c_1 \times \frac{L \Delta}{\varepsilon}},$ and $c_1$ and $c_2$ are universal constants.
\end{theorem}
This lower bound $\Omega(\underline{t}_{\underline{K}})$ is achieved by the Rennala SGD and Ringmaster ASGD asynchronous methods (up to constant factors) \citep{tyurin2023optimal,maranjyan2025ringmaster}. The result is non-explicit, and, in general, to find $\underline{t}_{\underline{K}},$ one has to first find $\underline{t}_{1},$ $\underline{t}_{2},$ and so forth, either theoretically or numerically.
Under Assumption~\ref{ass:speed}, we can prove the following theorem for $m$-Synchronous SGD:
\begin{restatable}[Time Complexity under Universal Computation Model]{theorem}{TIMECOMPLEXITYUNIVERSALBASE}
    \label{thm:sync_sgd_time_univ}
    Under the setting of Theorem~\ref{thm:sync_sgd}, additionally assume that Assumption~\ref{ass:speed} holds. Then, for all $m \in [n],$ the time complexity of $m$-Synchronous SGD (Algorithm~\ref{alg:alg_server_m_star}) is $\bar{t}_{\bar{K}}$ seconds, where the sequence $\{\bar{t}_k\}$ is defined as 
    \begin{equation}
    \begin{aligned}
        \textstyle \bar{t}_{k+1} = \min \left\{t \geq 0 \,:\, \max\limits_{S \subseteq [n], \abs{S} = m} \min\limits_{i \in S} N_i(\bar{t}_{k},t) = 2\right\},
    \end{aligned}
    \end{equation}
    $\bar{t}_0 = 0,$ and $\bar{K} \eqdef \ceil{16 \max\left\{\nicefrac{L \Delta}{\varepsilon}, \nicefrac{\sigma^2 L \Delta}{m \varepsilon^2}\right\}}.$
\end{restatable}
Theorem~\ref{thm:sync_sgd_time_univ} is also non-explicit. 
In general, Theorem~\ref{thm:sync_sgd_time_univ} does not match the lower bound in Theorem~\ref{thm:univ_lower_bound} (see Section~\ref{sec:couter_example}); however, there are several theoretical and practical examples where it does. We consider examples with constant powers and numerically generated powers in Section~\ref{sec:more_examples_univ}.
\subsection{Example: Optimization with stragglers}
\label{sec:strag}
One important example is when all workers have the same computational power, but at most $pn$ of them can be idle at any moment (e.g., $p = 0.1$ and at most 10\% are stragglers; see possible visualization in Figure~\ref{fig:power_1}). 
\begin{assumption}
\label{ass:partial}
Under Assumption~\ref{ass:speed}, for some fixed $p \in [0, 1),$ assume that $\forall t \geq 0,$ there exists a subset $S_t$ of at least size $(1 - p) n$ such that $v_i(t) = v \in \R_+$ $\forall i \in S_t,$ and $v_i(t) \leq v$ $\forall i \not\in S_t.$ 
\end{assumption}
Notice that for all $t \geq 0,$ we do not make any assumptions about the subset $S_t$; \emph{the subsets can be fixed, random, or even adversarial to the algorithm.} Clearly, in this case, a lower bound is
\begin{align}
    \label{eq:rrhiLqThpCcViIxf}
    \textstyle \Omega\left(\frac{1}{v} \max\left\{\frac{L \Delta}{\varepsilon}, \frac{\sigma^2 L \Delta}{n \varepsilon^2}\right\}\right),
\end{align}
in the best possible case when all workers participate. We now show that the time complexity of $m$-Synchronous SGD (Algorithm~\ref{alg:alg_server_m_star}) is optimal and matches \eqref{eq:rrhiLqThpCcViIxf} when $p < 0.4$:
\begin{restatable}[Time Complexity under Partial Participation]{theorem}{TIMECOMPLEXITYPARTIAL}
  \label{thm:partial}
  Under the setting of Theorem~\ref{thm:sync_sgd}, additionally assume that Assumption~\ref{ass:partial} holds. Then, the time complexity of Algorithm~\ref{alg:alg_server_m_star} is optimal for all $p < 0.4.$ For all $p < 0.4$ and any $m \in \{\flr{n / 5}, \dots, \flr{(1 - 2 p) n}\},$ the time complexity is
  \begin{align}
    \label{eq:vNzZqndZIaxqZyB}
    \textstyle \cO\left(\frac{1}{v} \max\left\{\frac{L \Delta}{\varepsilon}, \frac{\sigma^2 L \Delta}{n \varepsilon^2}\right\}\right)
  \end{align}
\end{restatable}
Comparing~\eqref{eq:rrhiLqThpCcViIxf} and~\eqref{eq:vNzZqndZIaxqZyB}, one can see that the complexity of $m$-Synchronous SGD is tight up to a constant factor when $p < 0.4$, which is a practical regime in the training of machine learning models, where it is very unlikely that more than 40\% of workers (GPUs) fail at the same moment in time. Algorithm~\ref{alg:alg_server_m_star} does not need the sets of participating workers; it is sufficient to fix any $m \in \{\flr{n / 5}, \dots, \flr{(1 - 2 p) n}\}.$ In practice, of course, one should take $m = \flr{(1 - 2p) n}$ to maximize the number of active workers. Notice that it is not necessarily workers $1$ through $m$ that participate; the set of active workers may vary.
Thus, even under the general partial participation assumption, it is sufficient to use $m$-Synchronous SGD to obtain the optimal theoretical time complexity.
\section{When Asynchronicity Might Be Needed}
\label{sec:where_needed}
\label{sec:couter_exmaple_main}
In the previous sections, we explained and proved that $m$-Synchronous SGD and Synchronous SGD are almost as fast as asynchronous methods in many scenarios. At the same time, we should acknowledge that synchronous methods are not a silver bullet and that there are several scenarios and corner cases where asynchronous methods might be needed. \\
\textbf{Rapid changes in computation times.} In general, it is possible that the upper bound $\bar{t}_{\bar{K}}$ from Theorem~\ref{thm:sync_sgd_time_univ} is asymptotically larger than the lower bound $\underline{t}_{\underline{K}}$ from Theorem~\ref{thm:univ_lower_bound}. 
Indeed, consider an example when the computation times of workers are the same and equal $\tau$ at the beginning. In this case, we choose $m = n$ in Algorithm~\ref{alg:alg_server_m_star}. However, suddenly, the first worker becomes infinitely fast after a few iterations: the new computation time is $\tilde{\tau}_1 \approx 0$, while the other computation times remain the same, $\tilde{\tau}_i = \tau$ for all $i > 1$. In this case, the time complexity of $m$-Synchronous SGD does not change, since it depends only on $\tau$ because $m = n,$ and the first worker calculates only one gradient, even though it is extremely fast. At the same time, the asynchronous methods will immediately adapt to the sudden speed of the first worker and leverage it. This is an example when the asynchronicity helps. See the formal proof in Section~\ref{sec:couter_example}. However, this scenario is rare and might be arguably impractical.
\footnotetext{In the proof of the lower bound, $c_1$ and $c_2$ are universal but small constants. For a fair comparison of $\bar{t}_{\bar{K}}$ and $\underline{t}_{\underline{K}},$ we take $c_1 = 16$ and $c_2 = 1,$ which can be obtained for Rennala SGD.} \\
\textbf{Large variance.} Recall Theorem~\ref{thm:sync_sgd_time_random} from Section~\ref{sec:random}, where we prove the time complexity \eqref{eq:ccJgxrZmhVOoCPhU} for $m$-Synchronous SGD. One potential limitation of this result is that $R$ can be large. Informally, the $R$ term comes from the fact that it is possible that $\Exp{\max_{i \in [m]} \bar{\tau}_{k,i}} \approx R \gg \tau_i$ due to the $\max_{i \in [m]}$ operation; at the same time, the asynchronous Rennala SGD method might be more robust to noise since, unlike $m$-Synchronous SGD, the method waits for several stochastic gradients, and thus the total waiting time averages the noise, potentially leading to a better asymptotic dependence on $R.$ \\
\textbf{Heterogeneous setup.} Under Assumptions~\ref{ass:fixed_computation_model} and \ref{ass:speed}, the heterogeneous setup with asynchronous methods is also well-explored \citep{tyurin2023optimal,tyurin2024tighttimecomplexitiesparallel,maranjyan2025ringleader}. In particular, consider stochastic optimization problem of minimizing a function:
$\textstyle \min_{x \in \R^d} \left\{f(x) \eqdef \nicefrac{1}{n} \sum_{i=1}^{n}\ f_i(x) \right\}$
where $f_i : \R^d \to \R$ and $f$ is a smooth function (Assumption~\ref{ass:lipschitz_constant}). Unlike the previous homogeneous setup, here we assume that only worker $i$ can compute stochastic gradients $\nabla f_i(x; \xi)$ at any point $x \in \R^d$ with $\xi \sim \mathcal{D}_i$, where $\mathcal{D}_i$ is some distribution such that $\nabla f_i(x; \xi)$ is an unbiased estimate of $\nabla f_i(x)$ and $\sigma^2$-variance bounded. Under Assumption~\ref{ass:fixed_computation_model}, it was shown that the optimal time complexity is
\begin{align}
\label{eq:SDnfeIwQUdMd}
\textstyle \Theta\left(\tau_n \frac{L \Delta}{\varepsilon} + \left(\frac{1}{n} \sum\limits_{i=1}^{n} \tau_i\right)\frac{\sigma^2 L \Delta}{n \varepsilon^2}\right),
\end{align}
achieved by Malenia SGD \citep{tyurin2023optimal}, an async method similar to Rennala SGD. 

We want to emphasize that classical Asynchronous SGD (Algorithm~\ref{alg:asgd}) cannot converge in general without the participation of all workers, since it is always possible to ``put all information about $f$'' on a worker that does not participate \citep{mishchenko2022asynchronous,tyurin2023optimal}. For the same reason, we cannot implement Algorithm~\ref{alg:alg_server_m_star} with $m < n$, since in the heterogeneous setting the participation of all workers is required because, in the worst case, potentially all information about $f$ can be adversarially placed in the workers that are ignored. Thus, it might be the case that the asynchronicity in Malenia SGD is still needed to achieve the optimal time complexity \eqref{eq:SDnfeIwQUdMd} and improve \eqref{eq:ELIDwwYnjnAnCj}.
\newcommand{\experimentone}{final_results/final__nn_1000_dim_1000_noise_0_01_delay_sqrt_time5000.pdf}
\begin{figure}[t]
    \centering
    \begin{subfigure}[t]{0.49\linewidth}
        \centering
        \includegraphics[width=0.8\linewidth]{\experimentone}
        \caption{Quadratic optimization with $n=10^3$ and $\tau_i = \sqrt{i}.$}
        \label{fig:plot1_main}
    \end{subfigure}
    \hfill
    \begin{subfigure}[t]{0.49\linewidth}
        \centering
        \begin{tikzpicture}
            \node[inner sep=0] (img) {\includegraphics[width=0.8\linewidth]
            {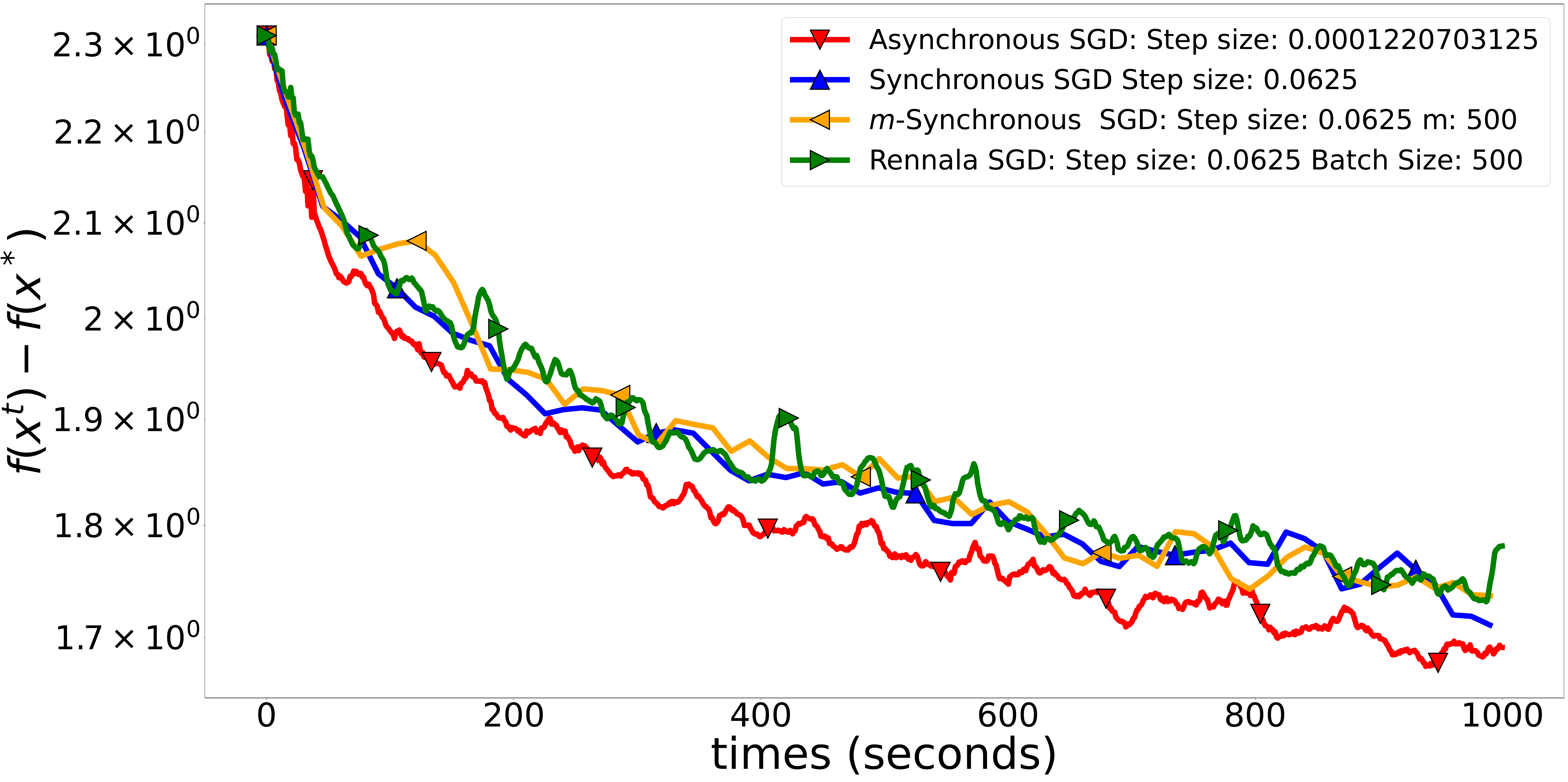}};

            \begin{scope}[overlay]
                \draw[<->, thick]
                    ([xshift=1.1cm,yshift=0.65cm] img.south)
                    --
                    node[below, font=\tiny, inner sep=4.2pt] {$\approx$ 1.45x}
                    ([xshift=2.5cm,yshift=0.65cm] img.south);
            \end{scope}
        \end{tikzpicture}
        \caption{CIFAR10 with NN and $\bar{\tau}_i \sim \mathrm{Unif}(0.5,1.5).$}
        \label{fig:plot2_main}
    \end{subfigure}
\end{figure}

However, the gap between $\tau_n$ and $\frac{1}{n} \sum_{i=1}^{n} \tau_i$ is not as bad as one might think. While it can be arbitrarily large (e.g., $\tau_1 = \dots = \tau_{n-1} \approx 0$ and $\tau_n \gg 0$), one can show that $\frac{1}{n} \sum_{i=1}^{n} \tau_i \geq \frac{1}{2} \tau_{\flr{n / 2}}$, and $\tau_n / \left(\frac{1}{n} \sum_{i=1}^{n} \tau_i\right) = \cO(\tau_n / \tau_{\flr{n / 2}}).$ In many scenarios where the ${\tau_i}$ grow reasonably fast, the gap can be a constant factor. For instance, if $\tau_m = \tau_1 m^{\alpha}$ for all $m \geq 1$ and some $0 \leq \alpha \leq 4$, then $\tau_n / \left(\frac{1}{n} \sum_{i=1}^{n} \tau_i\right) = \cO\left(1\right)$, meaning that even if $\tau_m$ grows as fast as $\tau_1 m^4$, the gap between Malenia SGD and Synchronous SGD is only a constant factor. 
\section{Experimental Highlights}


In Section~\ref{sec:exp_appendix}, we conducted numerical experiments on quadratic optimization tasks and small-scale machine learning problems in controlled heterogeneous computation environments, and on NanoGPT training in a real environment. All our experiments concur with our theoretical results, and we observe that either Synchronous SGD or $m$-Synchronous SGD performs well under heterogeneous computation times and converges comparably to the optimal Asynchronous SGD and Rennala SGD methods. In particular, in Figure~\ref{fig:plot1_main}, we run the methods on a quadratic optimization task with $n = 1000$, assuming that worker $i$ requires $\tau_i = \sqrt{i}$ seconds to compute one stochastic gradient, and ignoring other factors, including communication times. This experiment matches the discussion in Sections~\ref{sec:introduction} and \ref{sec:sync_optimal}, where we see that Synchronous SGD is slow due to the stragglers with large $\tau_i$; however, $m$-Synchronous SGD with $m = 10$, effectively ignoring stragglers, has the same convergence speed as the optimal methods despite computing one stochastic gradient per worker in every iteration. 
In Figure~\ref{fig:plot2_main}, we evaluate on CIFAR-10 \citep{krizhevsky2009learning} under a setup with random computation times and equal means. The results show that Synchronous SGD converges as fast as Rennala SGD and is at most a factor of $1.45\times$ slower than Asynchronous SGD, as theoretically predicted in Section~\ref{sec:random}. In Section~\ref{sec:why}, we explain why a constant-factor gap of $1$–$2\times$ is not an issue in real applications and environments.

\section{Conclusion}
\label{sec:why}
In this paper, we revisit the classical Synchronous SGD method and its robust variant, $m$-Synchronous SGD, which can be interpreted as a method with a backup-worker strategy, and compare them with modern asynchronous optimization methods. To summarize, we show that synchronous methods are nearly optimal (i) under the fixed computation model (Assumption~\ref{ass:fixed_computation_model}); (ii) under the random computation model (Assumption~\ref{ass:general_heter_random_computation_model}), at least when $\sigma^2 / \varepsilon$ is large and the times are truncated normal or exponential (see also Section~\ref{sec:more_examples}); (iii) under partial participation (Assumption~\ref{ass:partial}); and (iv) in several numerical examples (Section~\ref{sec:numerical_examples}).
Our work provides theoretical evidence for the popularity of synchronous optimization methods in modern training.



One might argue that asynchronous methods can still be better by a constant or a $\log n$ factor under Assumptions~\ref{ass:fixed_computation_model}, \ref{ass:general_heter_random_computation_model}, \ref{ass:partial}, and in Figure~\ref{fig:plot2_main}. We agree with this. However, in our NanoGPT experiment in a real homogeneous environment, with unstable computation times and custom implementations, we observe that the convergence rates are similar (see Section~\ref{sec:nanogpt_another}). Moreover, in practice, engineers implement various optimization and caching tricks to speed up synchronization \citep{sergeev2018horovod}, which can potentially mitigate the constant and $\log n$ factor losses relative to asynchronous methods. On top of all this, Asynchronous SGD is not all-reduce or communication-friendly \citep{tyurin2025birch}, unlike synchronous methods.

\section*{Acknowledgements}
The work was supported by the grant for research centers in the field of AI provided by the Ministry of Economic Development of the Russian Federation in accordance with the agreement 000000C313925P4F0002 and the agreement №139-10-2025-033.

\bibliography{example_paper}
\bibliographystyle{apalike}


\appendix
\newpage
\tableofcontents

\newpage

\section{Notations}
\begin{table}[H]
\centering
\caption{List of notations used throughout the paper.}
\begin{tabular}{cl}
\toprule
\textbf{Notation} & \textbf{Meaning} \\
\midrule
$\R_+$ & Denotes $[0, \infty).$ \\ 
$[n]$ & Denotes a finite set $\{1, \dots, n\}$. \\
$\norm{x}$ & Euclidean norm of a vector $x$. \\
$g = \cO(f)$ & There exists $C > 0$ such that $g(z) \le C \, f(z)$ for all $z \in \cZ$. \\
$g = \Omega(f)$ & There exists $C > 0$ such that $g(z) \ge C \, f(z)$ for all $z \in \cZ$. \\
$g = \Theta(f)$ & There exist $C_1, C_2 > 0$ such that $C_1 f(z) \le g(z) \le C_2 f(z)$ for all $z \in \cZ$. \\
$\tilde{\cO},\tilde{\Omega},$ and $\tilde{\Theta}$ & The same as $\cO$, $\Omega,$ and $\Theta,$ but up to logarithmic factors. \\
$\Exp{\cdot}$ & Full expectation. \\ 
\bottomrule
\end{tabular}
\end{table}

\section{Analysis with Communication Times}
\label{sec:comm}
In the main part, we analyze the methods assuming that communication times are negligible. We now extend the results from Sections~\ref{sec:sync_optimal} and \ref{sec:random} to scenarios where communication takes times. The main conclusions remain unchanged; however, the resulting time complexity becomes more involved, so for clarity, we defer the full results and analysis to the appendix.
\begin{assumption}
\label{ass:general_heter_random_computation_model_comm}
For all $i \in [n],$ the $i$\textsuperscript{th} worker requires $\bar{\tau}_i$ seconds to calculate a stochastic gradient and $\bar{\kappa}_i$ seconds to send/receive a vector to/from the server, where $\bar{\tau}_i$ is a $(\tau_i, R)$--sub-exponential random variable and $\bar{\kappa}_i$ is a $(\kappa_i, R)$--sub-exponential random variable. We assume that the sampled times $\{\bar{\tau}_i\},$ $\{\bar{\kappa}_i\}$, and the stochastic random variables $\{\xi_i\}$ are statistically independent. 
\end{assumption}
Instead of Algorithm~\ref{alg:alg_server_m_star}, we analyze its minor modification with mini-batching, which helps to reduce the number of synchronizations and improve communication efficiency.
\begin{algorithm}[H]
  \caption{$m$-Synchronous SGD with mini-batching}
  \label{alg:alg_server_m_star_mini}
  \begin{algorithmic}[1]
  \STATE \textbf{Input:} point $x^0 \in \R^d$, stepsize $\gamma$, \# of workers $m \in [n],$ batch size $B$
  \FOR{$k = 0, 1, \dots, K - 1$}
  \STATE Each idle worker starts computing stochastic gradients at $x^k$
  \STATE Aggregate from the first $m$ workers in this iteration: $$\textstyle g^k = \frac{1}{m}\sum\limits_{i \in S_k} \frac{1}{B} \sum_{j=1}^{B} \nabla f(x^k; \xi_{ij}^k),\,\,S_k \subseteq [n], \,\, \abs{S_k} = m$$
  \STATE Update $x^{k+1} = x^k - \gamma g^k$
  \STATE Start a non-blocking broadcast of $x^{k+1}$ to all workers
  \STATE Discard any gradients for iteration $k$ that arrive after the update
  \ENDFOR
  \end{algorithmic}
\end{algorithm}
Notice that Algorithm~\ref{alg:alg_server_m_star_mini} reduces to Algorithm~\ref{alg:alg_server_m_star} when $B = 1.$ We now prove the following iteration complexity result, an auxiliary result:
\begin{theorem}[Iteration Complexity; \citep{lan2020first}]
    \label{thm:sync_sgd_comm}
    Under Assumptions \ref{ass:lipschitz_constant}, \ref{ass:lower_bound}, and \ref{ass:stochastic_variance_bounded}, let $\gamma = \min\left\{ \frac{1}{2 L}, \frac{\varepsilon m B}{4L\sigma^2} \right\}$ in Algorithm~\ref{alg:alg_server_m_star_mini}. Then
    $\frac{1}{K}\sum \limits_{k=0}^{K-1} \Exp{\sqnorm{\nabla f (x^k)}} \le \varepsilon$ after
    \begin{align}
       \textstyle K \eqdef \ceil{16 \max\left\{\frac{L \Delta}{\varepsilon}, \frac{\sigma^2 L \Delta}{B m \varepsilon^2}\right\}}
    \end{align}
    iterations, where $m \in [n]$ and $B \geq 1$ are arbitrary numbers.
\end{theorem}
We now extend Theorem~\ref{thm:sync_sgd_time_random}:
\begin{restatable}[Time Complexity]{theorem}{TIMECOMPLEXITYRANDOMCOMM}
  \label{thm:sync_sgd_time_random_comm}
  Under the setting of Theorem~\ref{thm:sync_sgd_comm}, additionally assume that Assumption~\ref{ass:general_heter_random_computation_model_comm} holds. Then, the time complexity of Algorithm~\ref{alg:alg_server_m_star_mini} to find an $\varepsilon$--stationary point is at most
  \begin{align}
    \label{eq:hBmRQHDDbzMHfVB}
    \textstyle T_{\textnormal{rand}} \eqdef \min\limits_{S \subseteq [n], \abs{S} = m} \left[\sum\limits_{k=0}^{K - 1} \left(\max\limits_{i \in S} \sum_{j=1}^{B} \bar{\tau}_{k,ij} + \max\limits_{i \in S} \bar{\kappa}_{k,i1} + \max\limits_{i \in S} \bar{\kappa}_{k,i2}\right)\right]
  \end{align}
  seconds, where $\bar{\tau}_{k,ij}$ are $(\tau_i, R)$--sub-exponential random variables (independent copies) and $\bar{\kappa}_{k,i1}$ and $\bar{\kappa}_{k,i2}$ are $(\kappa_i, R)$--sub-exponential random variables (independent copies). Moreover, we have
  \begin{equation}
    \label{eq:EyRpBUbIxIyvCmPPZJ}
  \begin{aligned}
    \textstyle \Exp{T_{\textnormal{rand}}} = 
    \cO\Bigg(&\frac{L \Delta}{\varepsilon}\Bigg(\min\limits_{S \subseteq [n], \abs{S} = m} \left[\max\limits_{i \in S} B \tau_{i} + \max\limits_{i \in S} \kappa_{i}\right] \\
    &\quad + R \max\left\{\log (2 n) , \sqrt{B \log (2 n)}\right\}\Bigg)\max\left\{1, \frac{\sigma^2}{B m \varepsilon}\right\}\Bigg),
  \end{aligned}
  \end{equation}
  for all $m \in [n]$ and $B \geq 1$ in Algorithm~\ref{alg:alg_server_m_star_mini}. Finally, taking $B = \max\left\{1, \ceil{\frac{\sigma^2}{m \varepsilon}}\right\},$ 
  \begin{equation}
  \label{eq:EUHvPBGIsBUtfMo}
  \begin{aligned}
    \Exp{T_{\textnormal{rand}}} 
    &= \cO\Bigg(T_{\textnormal{det}}(m) + \frac{L \Delta}{\varepsilon} R \max\left\{\log (2 n) , \sqrt{\max\left\{1, \frac{\sigma^2}{m \varepsilon}\right\} \log (2 n)}\right\}\Bigg)
  \end{aligned}
  \end{equation}
  for all $m \in [n],$ where
  \begin{align*}
    T_{\textnormal{det}}(m) \eqdef \frac{L \Delta}{\varepsilon} \min\limits_{S \subseteq [n], \abs{S} = m} \left[\max \left\{\max\limits_{i \in S} \kappa_{i}, \max\limits_{i \in S} \tau_{i} \max\left\{1, \frac{\sigma^2}{m \varepsilon}\right\}\right\}\right].
  \end{align*}
\end{restatable}
Similarly to Theorem~\ref{thm:sync_sgd_time_random}, we get the deterministic term plus an additional ``noise'' term proportional to $R.$ In the deterministic term, we have a free parameter $m \in [n].$ In order to find an optimal $m,$ we can use Algorithm~\ref{alg:worker-min} and prove the following result:
\begin{algorithm}[t]
\caption{Minimizing the Objective $T_{\textnormal{det}}(m)$}
\label{alg:worker-min}
\begin{algorithmic}[1]
\REQUIRE Workers $i\in[n]$ with parameters $(\kappa_i,\tau_i)$, ratio $\sigma^2 / \varepsilon$
\ENSURE Optimal $m$
\STATE $\mathrm{OPT} = +\infty$
\STATE $S^\star = \emptyset$
\STATE Let $K_1 < K_2 < \cdots < K_L$ be the distinct values among $\{\kappa_i : i\in[n]\}$
\FOR{$\ell = 1,\dots,L$}
    \STATE $P_\ell = \{i\in[n] : \kappa_i \leq K_\ell\}$
    \STATE Sort $P_\ell$ by increasing $\tau_i$: $\tau_{j_1} \leq \tau_{j_2} \leq \cdots \leq \tau_{j_{|P_\ell|}}.$
    \FOR{$m = 1,\dots, |P_\ell|$}
        \STATE $\beta_m = \max\left\{1,\frac{\sigma^2}{m \varepsilon}\right\}$
        \STATE $V = \max\left\{\max\limits_{p \in [m]} \kappa_{j_p}, \beta_m \tau_{j_m}\right\}$
        \IF{$V < \mathrm{OPT}$}
            \STATE $\mathrm{OPT} = V$
            \STATE $S^\star = \{j_1,\dots,j_m\}$
        \ENDIF
    \ENDFOR
\ENDFOR
\RETURN $m = \abs{S^\star}$
\end{algorithmic}
\end{algorithm}
\begin{restatable}[Time Complexity with a Good Choice of $m$]{theorem}{THEOREMBEST}
  \label{thm:best}
  Consider the setting of Theorem~\ref{thm:sync_sgd_time_random_comm}. Using Algorithm~\ref{alg:worker-min}, we can find $m^* \in [n]$ that minimizes $T_{\textnormal{det}}(m).$ Moreover,
    $\min\limits_{m \in [n]} T_{\textnormal{det}}(m) = \cO\left(\log(n + 1) T_{\textnormal{optimal}}\right),$
  and 
  \begin{align*}
    \Exp{T_{\textnormal{rand}}} 
    &= \cO\Bigg(\log(n + 1) T_{\textnormal{optimal}} + \frac{L \Delta}{\varepsilon} R \max\left\{\log (2 n) , \sqrt{\max\left\{1, \frac{\sigma^2}{m^* \varepsilon}\right\} \log (2 n)}\right\}\Bigg)
  \end{align*}
  for $m = m^*$ and $B = \max\left\{1, \ceil{\frac{\sigma^2}{m \varepsilon}}\right\},$ where 
  \begin{align*}
  T_{\textnormal{optimal}}
  =
  \tilde{\Theta}\left(\frac{L \Delta}{\varepsilon}
  \min_{m \in [n]}
  \max\left\{
    \max\{\kappa_{\pi_m}, \tau_{\pi_m}\},
    \frac{\sigma^2}{\varepsilon}
    \left(\sum_{i=1}^{m}\frac{1}{\tau_{\pi_i}}\right)^{-1}
  \right\}\right)
\end{align*}
is the optimal time complexity in the centralized setup with deterministic computation and communication times \citep{tyurin2024shadowheart,tyurin2024optimalgraph}
and $\pi$ is a permutation that sorts $\{\max\{\kappa_i,\tau_i\}\}_{i \in [n]}:$ $\max\{\kappa_{\pi_1},\tau_{\pi_1}\} \leq \dots \leq \max\{\kappa_{\pi_n},\tau_{\pi_n}\}.$
\end{restatable}
Theorem~\ref{thm:best} shows that the time complexity of Algorithm~\ref{alg:alg_server_m_star_mini} is nearly optimal up to a logarithmic factor when the times are non-random (e.g., $R = 0$ in Assumption~\ref{ass:general_heter_random_computation_model_comm}), a result similar to Theorem~\ref{thm:sync_sgd_time}. In more general setting, when $R > 0,$ Algorithm~\ref{alg:alg_server_m_star_mini} achieves $\cO\left(\log(n + 1) T_{\textnormal{optimal}}\right)$ when $\min_{m \in [n]} T_{\textnormal{det}}(m)$ dominates $\frac{L \Delta}{\varepsilon} R \max\left\{\log (2 n) , \sqrt{\max\left\{1, \frac{\sigma^2}{m^* \varepsilon}\right\} \log (2 n)}\right\},$ which is true when $\varepsilon$ is small, $\{\kappa_i\}$ are large, $\sigma$ is large, or $R$ is small, a result similar to Corollary~\ref{cor:random}.

\newpage
\section{Universal Computation Model: More Examples}
\label{sec:more_examples_univ}
\begin{figure*}[h]
\centering
\begin{minipage}[t]{0.49\textwidth}
    \centering
    \includegraphics[width=0.65\textwidth]{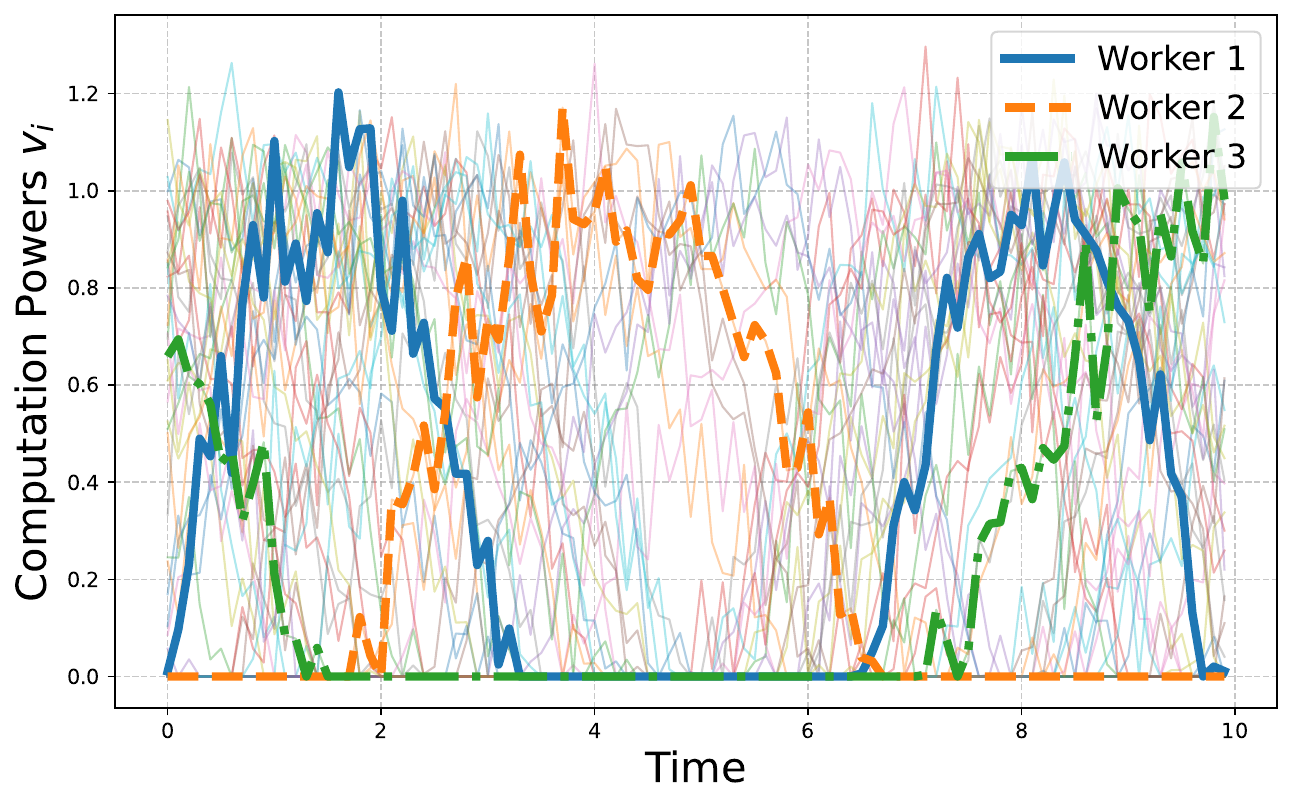}
    \captionof{figure}{Computation powers for $50$ workers, where every worker chaotically stops participating in the computation process (its computation power is zero for some time intervals). We generate worker powers as $v_i(t_k)=\max\{\sin(a_i t_k+s_i)+\varepsilon_{ik},\,0\},\varepsilon_{ik}\sim\mathcal{N}(0,0.1^2), s_i\sim\mathrm{Unif}(0,2\pi), a_i\sim\mathrm{Unif}(0.5,1),$ on a uniform grid $t_k=0.1k$, and define $v_i(t)$ by linear interpolation.}
    \label{fig:power_1}
\end{minipage}\hfill
\begin{minipage}[t]{0.49\textwidth}
    \centering
    \includegraphics[width=0.65\textwidth]{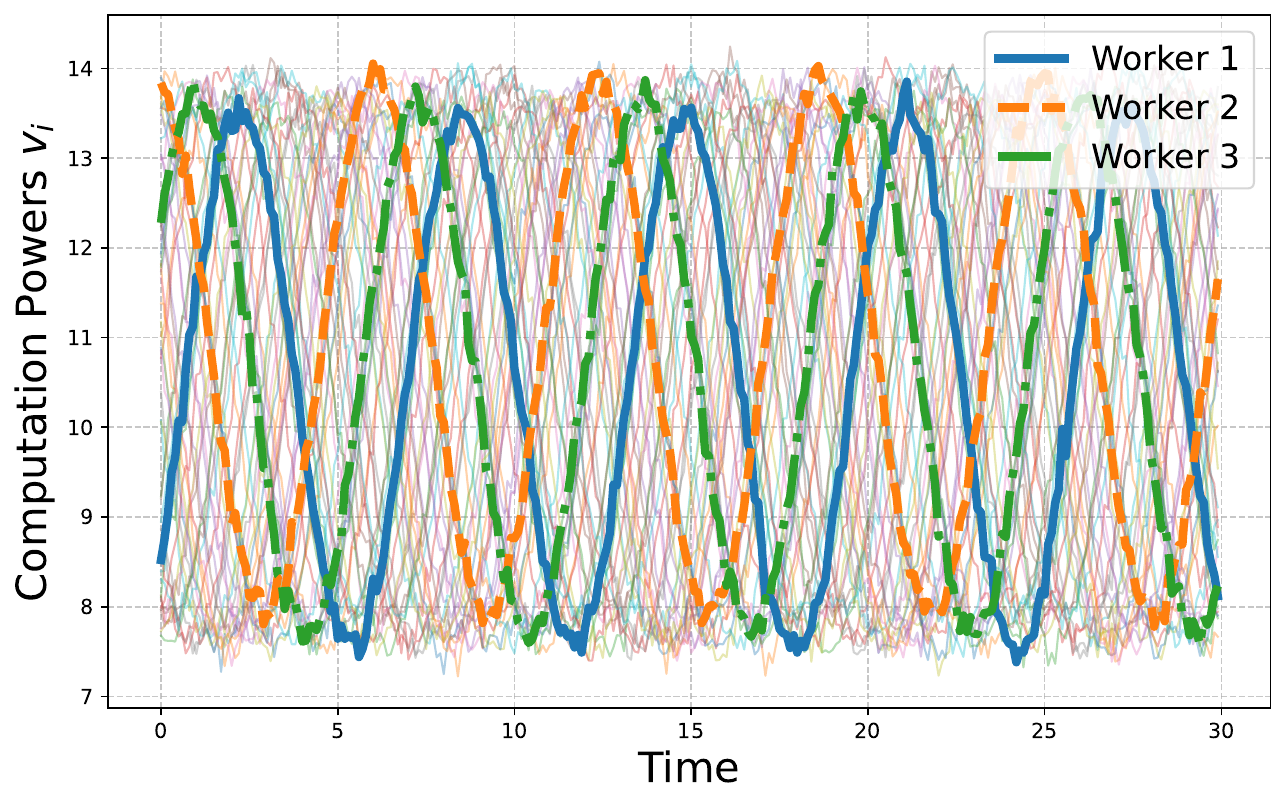}
    \captionof{figure}{Computation powers for $50$ workers, where every worker has a periodic computation power. We generate worker powers $v_i(t_k)=\max\{s_i+3\sin(t_k+\phi_i)+\varepsilon_{ik},\,0.1\}, \varepsilon_{ik}\sim\mathcal{N}(0,0.1^2),\  s_i\sim\mathrm{Unif}(10.5,11.0), \phi_i\sim\mathrm{Unif}(0,2\pi),$ on a uniform grid $t_k=0.1k,$ and define $v_i(t)$ by linear interpolation.}
    \label{fig:power_2}
\end{minipage}
\end{figure*}
In this section, we consider some examples of universal computation models from Section~\ref{sec:univ}.
\subsection{Example: different constant powers}
\label{sec:theoretical_examples}
In the case of constant powers, when $v_i(t) = v_i \in \R_+$ for all $t \geq 0,$ Theorem~\ref{thm:sync_sgd_time_univ} reduces to Theorem~\ref{thm:sync_sgd_time}. Indeed, $N_i(t_0,t_1) = \flr{(t_1 - t_0) v_i}$ and 
    $\max_{S \subseteq [n], \abs{S} = m} \min_{i \in S} N_i(\bar{t}_{k},t) = \max_{S \subseteq [n], \abs{S} = m} \min_{i \in S} \flr{(t - \bar{t}_{k}) v_i}.$
W.l.o.g., assume that $v_1 \geq \dots \geq v_n,$ then
    $\max_{S \subseteq [n], \abs{S} = m} \min_{i \in S} N_i(\bar{t}_{k},\bar{t}_{k+1}) = 2$ 
for $\bar{t}_{k+1} = \bar{t}_k + \frac{2}{v_m} = \frac{2 (k + 1)}{v_m}.$ Due to Theorem~\ref{thm:sync_sgd_time_univ}, the method converges after $\bar{t}_{\bar{K}} = \cO\left(\nicefrac{1}{v_m} \max\left\{\nicefrac{L \Delta}{\varepsilon}, \nicefrac{\sigma^2 L \Delta}{m \varepsilon^2}\right\}\right) $
seconds, and we get the same result as in Theorem~\ref{thm:sync_sgd_time} up to a constant factor if we take $m = \arg\min_{m \in [n]} \left[1 / v_m \cdot \max\left\{1, \nicefrac{\sigma^2}{m \varepsilon}\right\}\right]$ and recall that $\tau_m \equiv 1 / v_m$ in the scenario with the fixed powers. Similarly, one can show that the lower bound in Theorem~\ref{thm:univ_lower_bound} is $\underline{t}_{\underline{K}} = \eqref{eq:optimal_const}$
which corresponds to the optimal time complexity~\eqref{eq:optimal_const}. Here, as expected, $\bar{t}_{\bar{K}}$ is tight up to logarithmic factors for an appropriate choice of $m$.

\subsection{Numerical examples}
\label{sec:numerical_examples}
Since $\bar{t}_{\bar{K}}$ and $\underline{t}_{\underline{K}}$ are non-explicit, we generate several practical scenarios and \emph{numerically} evaluate the gap $\bar{t}_{\bar{K}} / \underline{t}_{\underline{K}}$ between the lower bound and the time complexity of $m$-Synchronous SGD (Algorithm~\ref{alg:alg_server_m_star}).

We start with Figure~\ref{fig:power_1}, where every worker participates in the optimization process chaotically. We numerically evaluate\footnotemark $\bar{t}_{\bar{K}}$ and $\underline{t}_{\underline{K}}$ from Theorems~\ref{thm:sync_sgd_time_univ} and~\ref{thm:univ_lower_bound} in two noise regimes. We fix $n = 50$ and $L \Delta / \varepsilon = 1$ and show that $\bar{t}_{\bar{K}} / \underline{t}_{\underline{K}} \leq 1.52$ for $\sigma^2 / \varepsilon = 100$ and $m = 15,$ and that $\bar{t}_{\bar{K}} / \underline{t}_{\underline{K}} \leq 1.85$ for $\sigma^2 / \varepsilon = 1000$ and $m = 14.$ Thus, the difference is only a factor of $1.85.$ In Figure~\ref{fig:power_2}, we consider a slightly different scenario with periodic computation powers. We fix $n = 50$ and $L \Delta / \varepsilon = 1$ and show that $\bar{t}_{\bar{K}} / \underline{t}_{\underline{K}} \leq 1.11$ for $\sigma^2 / \varepsilon = 100$ and $m = 49,$ and that $\bar{t}_{\bar{K}} / \underline{t}_{\underline{K}} \leq 1.37$ for $\sigma^2 / \varepsilon = 1000$ and $m = 50.$

Notice that the computation powers in Figures~\ref{fig:power_1} and \ref{fig:power_2} are very chaotic, and intuitively, the time complexity of $m$-Synchronous SGD should be much worse than the time complexity of asynchronous methods. Nevertheless, with a proper choice of $m$ in Algorithm~\ref{alg:alg_server_m_star}, it is possible to obtain a time complexity that is worse by at most a factor of $2.$ See our discussion in Section~\ref{sec:why} why a minor constant factor is not a problem in practice.

\newpage
\section{Proof of Theorem~\ref{thm:sync_sgd_time}}
\label{sec:proof}

\TIMECOMPLEXITY*

\begin{proof}
    The first term in the theorem follows from the fact that every iteration takes at most $\tau_m$ seconds; thus, the total time is at most $\tau_m \times 16 \max\left\{\frac{L \Delta}{\varepsilon}, \frac{\sigma^2 L \Delta}{m \varepsilon^2}\right\}$ seconds due to Theorem~\ref{thm:sync_sgd}. 

    Let $\frac{\sigma^2}{\varepsilon} \leq 1.$ Then
    \begin{align*}
        \min_{m \in [n]} \left[\tau_m\max\left\{\frac{L \Delta}{\varepsilon}, \frac{\sigma^2 L \Delta}{m \varepsilon^2}\right\}\right] = \tau_1 \frac{L \Delta}{\varepsilon}
    \end{align*}
    and 
    \begin{align*}
        T_{\textnormal{optimal}} \eqdef \Theta\left(\min \limits_{m \in [n]} \left[\left(\frac{1}{m}\sum\limits_{i=1}^m \frac{1}{\tau_i}\right)^{-1}\left(\max\left\{\frac{L \Delta}{\varepsilon}, \frac{\sigma^2 L \Delta}{m \varepsilon^2}\right\}\right)\right]\right) = \Theta\left(\tau_1 \frac{L \Delta}{\varepsilon}\right).
    \end{align*}
    since both terms are minimized with $m = 1.$ Thus, the theorem is correct. Further we assume that $\frac{\sigma^2}{\varepsilon} > 1.$
    
    Let us define $m^*$ as the smallest minimum of $\min\limits_{m \in [n]} \left[\tau_m \max\left\{1, \frac{\sigma^2}{m \varepsilon}\right\}\right].$ Then
    \begin{align}
        \label{eq:fcpaNqu}
        \tau_m \max\left\{1, \frac{\sigma^2}{m \varepsilon}\right\} \geq \tau_{m^*}\max \left\{1, \frac{\sigma^2}{m^* \varepsilon}\right\}
    \end{align} for all $m \in [n].$ 

    The optimal time complexity is 
    \begin{align}
        \label{eq:AZbAlB}
        T_{\textnormal{optimal}} \eqdef \Theta\left(\min \limits_{m \in [n]} \left[\left(\frac{1}{m}\sum\limits_{i=1}^m \frac{1}{\tau_i}\right)^{-1}\left(\max\left\{\frac{L \Delta}{\varepsilon}, \frac{\sigma^2 L \Delta}{m \varepsilon^2}\right\}\right)\right]\right).
    \end{align}
    Let us define $\bar{m}$ as the smallest minimizer in \eqref{eq:AZbAlB}. Then 
    \begin{align}
        \label{eq:UIIDsdAeyvdJqkL}
        T_{\textnormal{optimal}} 
        &= \Theta\left(\left(\frac{1}{\bar{m}}\sum\limits_{i=1}^{\bar{m}} \frac{1}{\tau_i}\right)^{-1}\left(\max\left\{\frac{L \Delta}{\varepsilon}, \frac{\sigma^2 L \Delta}{\bar{m} \varepsilon^2}\right\}\right)\right).
    \end{align}
    and $\bar{m} \leq \ceil{\frac{\sigma^2}{\varepsilon}} \leq \frac{2 \sigma^2}{\varepsilon}$ since $\left(\frac{1}{m}\sum\limits_{i=1}^{m} \frac{1}{\tau_i}\right)^{-1}$ is a non-decreasing sequence as a function of $m$ and $\frac{\sigma^2}{\varepsilon} > 1.$ Moreover, $\bar{m} \geq \max\{m^* - 1, 1\}.$ It is true if $m^* = 1.$ Assuming $m^* > 1,$ let us consider two cases to prove $\bar{m} \geq \max\{m^* - 1, 1\}:$ \\
    i) if $1 \leq \frac{\sigma^2}{m^* \varepsilon},$ then
    \begin{align*}
        g(m) \eqdef \left(\frac{1}{m}\sum\limits_{i=1}^m \frac{1}{\tau_i}\right)^{-1}\left(\max\left\{\frac{L \Delta}{\varepsilon}, \frac{\sigma^2 L \Delta}{m \varepsilon^2}\right\}\right) = \left(\sum\limits_{i=1}^m \frac{1}{\tau_i}\right)^{-1}\left(\frac{\sigma^2 L \Delta}{\varepsilon^2}\right) \\
    \end{align*}
    for all $m \leq m^*,$ and $g(m)$ can only potentially decrease if we take $m > m^* \geq \max\{m^* - 1, 1\}.$ \\
    ii) if $1 > \frac{\sigma^2}{m^* \varepsilon},$ then $1 < \frac{\sigma^2}{(m^* - 1) \varepsilon}$ because $m^*$ is the smallest minimum of $\min\limits_{m \in [n]} \left[\tau_m \max\left\{1, \frac{\sigma^2}{m \varepsilon}\right\}\right].$ Similarly to the previous case, since $1 < \frac{\sigma^2}{(m^* - 1) \varepsilon},$ we can conclude that $\bar{m} \geq m^* - 1 \geq \max\{m^* - 1, 1\}.$
    

    Due to \eqref{eq:fcpaNqu}, for all $m \in [\bar{m}],$
    \begin{align*}
        \frac{2 \tau_m \sigma^2}{m \varepsilon} \overset{m \leq \bar{m} \leq \frac{2 \sigma^2}{\varepsilon}}{\geq} \tau_m \max\left\{1, \frac{\sigma^2}{m \varepsilon}\right\} \geq \tau_{m^*}\max \left\{1, \frac{\sigma^2}{m^* \varepsilon}\right\} \geq \frac{\tau_{m^*} \sigma^2}{m^* \varepsilon}
    \end{align*}
    and 
    \begin{align*}
        \tau_m \geq \frac{\tau_{m^*} m}{2 m^*}.
    \end{align*}
    Substituting this bound to \eqref{eq:UIIDsdAeyvdJqkL}, we get
    \begin{align*}
        T_{\textnormal{optimal}} 
        &= \Omega\left(\left(\frac{m^*}{\tau_{m^*} \bar{m}}\sum\limits_{i=1}^{\bar{m}} \frac{1}{i}\right)^{-1}\left(\max\left\{\frac{L \Delta}{\varepsilon}, \frac{\sigma^2 L \Delta}{\bar{m} \varepsilon^2}\right\}\right)\right).
    \end{align*}
    Using the well-known inequality $\sum_{i=1}^{\bar{m}} \frac{1}{i} \leq 1 + \log \bar{m} \leq 2 \log(n + 1),$
    \begin{align*}
        T_{\textnormal{optimal}} 
        &= \Omega\left(\frac{\tau_{m^*} \bar{m}}{m^* \log(n + 1)}\left(\max\left\{\frac{L \Delta}{\varepsilon}, \frac{\sigma^2 L \Delta}{\bar{m} \varepsilon^2}\right\}\right)\right) = \Omega\left(\frac{1}{\log(n + 1)} \times \tau_{m^*}\left(\max\left\{\frac{\bar{m} L \Delta}{m^* \varepsilon}, \frac{\sigma^2 L \Delta}{m^* \varepsilon^2}\right\}\right)\right).
    \end{align*}
    It is left to recall that $\bar{m} \geq \max\{m^* - 1, 1\} \geq \frac{m^*}{2}$ to get
    \begin{align*}
        T_{\textnormal{optimal}} = \Omega\left(\frac{1}{\log(n + 1)} \times \tau_{m^*}\left(\max\left\{\frac{L \Delta}{\varepsilon}, \frac{\sigma^2 L \Delta}{m^* \varepsilon^2}\right\}\right)\right) = \Omega\left(\frac{1}{\log(n + 1)} \times \min_{m \in [n]} \left[\tau_m\max\left\{\frac{L \Delta}{\varepsilon}, \frac{\sigma^2 L \Delta}{m \varepsilon^2}\right\}\right]\right)
    \end{align*}
    and the result of the theorem.
\end{proof}

\section{Proof of Theorem~\ref{thm:sync_sgd_time_random}}

\TIMECOMPLEXITYRANDOM*

\begin{proof}
    We run Algorithm~\ref{alg:alg_server_m_star} at most 
    \begin{align*}
        K \eqdef \ceil{16 \max\left\{\frac{L \Delta}{\varepsilon}, \frac{\sigma^2 L \Delta}{m \varepsilon^2}\right\}}
    \end{align*}
    iterations. Let $t_k$ be the time when the $k$\textsuperscript{th} iteration is finished. Workers $1, \dots, m$ finish computing the first stochastic gradients after at most
    \begin{align*}
        \max_{i \in [m]} \bar{\tau}_{0,i} \geq t_1
    \end{align*}
    seconds. Using mathematical induction, assume that $\sum_{k=0}^{K - 1} \max_{i \in [m]} \bar{\tau}_{k,i} \geq t_{K}.$ Iteration $K + 1$ starts at time $t_{K}.$ Let $\bar{k}_i \in \{0, \dots, K\}$ be the last iteration when worker $i$ finishes before the end of this iteration. Then worker $i$ starts calculating the first stochastic gradient after time $t_k,$ after at most
    \begin{align*}
        t_{\bar{k}_i} + \sum_{k=\bar{k}_i}^{K - 1} \bar{\tau}_{k,i}
    \end{align*}
    seconds. Therefore, workers $1, \dots, m$ finish their first stochastic gradient after time $t_k,$ after at most
    \begin{align*}
        &\max_{i \in [m]} \left(t_{\bar{k}_i} + \sum_{k=\bar{k}_i}^{K - 1} \bar{\tau}_{k,i} + \bar{\tau}_{K,i}\right) \leq \max_{i \in [m]} \left(\sum_{k=0}^{\bar{k}_i - 1} \max_{i \in [m]} \bar{\tau}_{k,i} + \sum_{k=\bar{k}_i}^{K} \bar{\tau}_{k,i}\right) \\
        &\leq \max_{i \in [m]} \left(\sum_{k=0}^{\bar{k}_i - 1} \max_{i \in [m]} \bar{\tau}_{k,i} + \sum_{k=\bar{k}_i}^{K} \max_{i \in [m]} \bar{\tau}_{k,i}\right) = \sum_{k=0}^{K} \max_{i \in [m]} \bar{\tau}_{k,i}
    \end{align*}
    seconds, meaning that $t_{K+1} \leq \sum_{k=0}^{K} \max\limits_{i \in [m]} \bar{\tau}_{k,i}.$

    It is left to prove the second part of the theorem. We have
    \begin{align*}
        \Exp{\max_{i \in [m]} \left(\bar{\tau}_{k,i} - \tau_i\right)} = \frac{1}{\lambda} \Exp{\log \left(\exp\left(\lambda \max_{i \in [m]} \left(\bar{\tau}_{k,i} - \tau_i\right)\right)\right)}
    \end{align*}
    for all $\lambda > 0.$ Since $\log x$ is concave, 
    \begin{align*}
        \Exp{\max_{i \in [m]} \left(\bar{\tau}_{k,i} - \tau_i\right)} 
        &\leq \frac{1}{\lambda} \log \left(\Exp{\exp\left(\lambda \max_{i \in [m]} \left(\bar{\tau}_{k,i} - \tau_i\right)\right)}\right)
        = \frac{1}{\lambda} \log \left(\Exp{\max_{i \in [m]} \exp\left(\lambda \left(\bar{\tau}_{k,i} - \tau_i\right)\right)}\right) \\
        &\leq \frac{1}{\lambda} \log \left(\Exp{\sum_{i=1}^{m} \exp\left(\lambda \left(\bar{\tau}_{k,i} - \tau_i\right)\right)}\right)
        \leq \frac{1}{\lambda} \log \left(\sum_{i=1}^{m} \Exp{\exp\left(\lambda \abs{\bar{\tau}_{k,i} - \tau_i}\right)}\right).
    \end{align*}
    Taking $\lambda = \frac{1}{R},$ we get
    \begin{align*}
        \Exp{\max_{i \in [m]} \left(\bar{\tau}_{k,i} - \tau_i\right)} 
        &\leq R \log \left(\sum_{i=1}^{m} \Exp{\exp\left(\abs{\bar{\tau}_{k,i} - \tau_i} / R\right)}\right) \leq R \log (2 m)
    \end{align*}
    since $\bar{\tau}_{k,i}$ is a $(\tau_i, R)$--sub-exponential random variable.
    It is left to notice that 
    \begin{align*}
        \max_{i \in [m]} \bar{\tau}_{k,i} \leq \max_{i \in [m]} \left(\bar{\tau}_{k,i} - \tau_i\right) + \max_{i \in [m]} \tau_i.
    \end{align*}
\end{proof}

\section{Proof of Theorem~\ref{thm:sync_sgd_time_random_comm}}

\TIMECOMPLEXITYRANDOMCOMM*

\begin{proof}
    We fix any $S \subseteq [n]$ such that $\abs{S} = m.$ Then, the proof of \eqref{eq:hBmRQHDDbzMHfVB} is the same in the proof of Theorem~\ref{thm:sync_sgd_time_random}. There are only two differences: i) every worker calculates $B$ stochastic gradients instead of one stochastic gradient; ii) worker $i$ sends aggregated stochastic gradients, which takes $\bar \kappa_{k,i1}$ seconds and receives a new iterate from the server, which takes $\bar \kappa_{k,i2}$ seconds. Then, it is left to use that $S \subseteq [n]$ is an arbitrary subset such that $\abs{S} = m.$

    We now prove \eqref{eq:EyRpBUbIxIyvCmPPZJ}, which also almost repeats the proof of Theorem~\ref{thm:sync_sgd_time_random}. Notice that 
    \begin{align*}
      T_{\textnormal{rand}} 
      &\leq \sum\limits_{k=0}^{K - 1} \left(\max\limits_{i \in S} B \tau_{i} + 2 \max\limits_{i \in S} \kappa_{i}\right) \\
      &\quad + \sum\limits_{k=0}^{K - 1} \left(\max\limits_{i \in S} \sum_{j=1}^{B} (\bar{\tau}_{k,ij} - \tau_{i}) + \max\limits_{i \in S} (\bar{\kappa}_{k,i1} - \kappa_{i}) + \max\limits_{i \in S} (\bar{\kappa}_{k,i2} - \kappa_{i})\right)
    \end{align*}
    for all $S \subseteq [n]$ such that $\abs{S} = m.$
    Similarly to the proof of Theorem~\ref{thm:sync_sgd_time_random}, one can easily show that $\Exp{\max\limits_{i \in S} (\bar{\kappa}_{k,i1} - \kappa_{i}) + \max\limits_{i \in S} (\bar{\kappa}_{k,i2} - \kappa_{i})} \leq 2 R \log(2 n).$ Also,
    \begin{align*}
      \Exp{\max\limits_{i \in S} \sum_{j=1}^{B} (\bar{\tau}_{k,ij} - \tau_{i})} \leq \frac{1}{\lambda} \log \left(\sum_{i=1}^{m} \Exp{\exp\left(\lambda \sum_{j=1}^{B} (\bar{\tau}_{k,ij} - \tau_{i})\right)}\right).
    \end{align*}
    Since the random variables are i.i.d.,
    \begin{align*}
      \Exp{\max\limits_{i \in S} \sum_{j=1}^{B} (\bar{\tau}_{k,ij} - \tau_{i})} \leq \frac{1}{\lambda} \log \left(\sum_{i=1}^{m} \left(\Exp{\exp\left(\lambda (\bar{\tau}_{k,i1} - \tau_{i})\right)}\right)^B\right).
    \end{align*}
    Since $\bar{\tau}_{k,i1}$ is a sub-exponential random variable, there exist a universal constant $c_1 \geq 1$ such that 
    \begin{align*}
        &\Exp{\max\limits_{i \in S} \sum_{j=1}^{B} (\bar{\tau}_{k,ij} - \tau_{i})} \leq \frac{1}{\lambda} \log \left(m \exp\left(c_1^2 \lambda^2 R^2 B\right)\right) = \frac{\log m}{\lambda} + c_1^2 \lambda R^2 B
    \end{align*}
    for all $\lambda \leq \frac{1}{c_1 R}$ \citep{vershynin2018high}. Taking $\lambda = \min\left\{\frac{1}{c_1 R}, \frac{\sqrt{\log (2 m)}}{c_1 R \sqrt{B}}\right\},$
    \begin{align*}
        &\Exp{\max\limits_{i \in S} \sum_{j=1}^{B} (\bar{\tau}_{k,ij} - \tau_{i})} = \cO\left(R \max\left\{\log (2 m) , \sqrt{B \log (2 m)}\right\}\right).
    \end{align*}
    Thus,
    \begin{align*}
      \Exp{T_{\textnormal{rand}}} = 
    \cO\left(\frac{L \Delta}{\varepsilon}\left(\left[\max\limits_{i \in S} B \tau_{i} + \max\limits_{i \in S} \kappa_{i}\right] + R \max\left\{\log (2 n) , \sqrt{B \log (2 n)}\right\}\right)\max\left\{1, \frac{\sigma^2}{B m \varepsilon}\right\}\right)
    \end{align*}
    for all $S \subseteq [n]$ such that $\abs{S} = m.$ It is left to take the min over all $S \subseteq [n]$ such that $\abs{S} = m.$
\end{proof}

\section{Proof of Theorem~\ref{thm:best}}

\THEOREMBEST*

\begin{proof}
Define
\begin{align*}
  \widehat{T}_{\textnormal{det}}(m)
  \eqdef
  \min_{\substack{S \subseteq [n], |S| = m}}
  \max\left\{
    \max_{i \in S} \kappa_i,
    \beta_m \max_{i \in S} \tau_i
  \right\},
  \qquad
  \beta_m \eqdef \max\left\{1,\frac{\sigma^2}{m\varepsilon}\right\}.
\end{align*}
For each $\ell \in [L]$,
\begin{align*}
  P_\ell
  =
  \{i \in [n] : \kappa_i \le K_\ell\},
\end{align*}
which also defined in Algorithm~\ref{alg:worker-min}. After sorting $P_\ell$ by increasing $\tau_i$, write
\begin{align*}
  \tau_{j_1} \le \tau_{j_2} \le \cdots \le \tau_{j_{|P_\ell|}}.
\end{align*}
For every $m \le |P_\ell|$, define
  $S_{\ell,m}
  \eqdef
  \{j_1,\dots,j_m\}.$
The value computed by the algorithm in every inner loop of Algorithm~\ref{alg:worker-min} is
\begin{align*}
  V_{\ell,m}
  &\eqdef
  \max\left\{
    \max_{p \in [m]} \kappa_{j_p},
    \beta_m \tau_{j_m}
  \right\} =
  \max\left\{
    \max_{i \in S_{\ell,m}} \kappa_i,
    \beta_m \max_{i \in S_{\ell,m}} \tau_i
  \right\}.
\end{align*}
Therefore $V_{\ell,m}$ is exactly the objective value on the set
$S_{\ell,m}$. Hence
$\min_{m \in [n]} \widehat{T}_{\textnormal{det}}(m)
  \le
  \mathrm{OPT}.$

It remains to prove the opposite inequality. Fix any $m \in [n]$ and any set
$S \subseteq [n]$ with $|S|=m$. Let
  $K_\ell
  =
  \max_{i \in S} \kappa_i.$
for some $\ell.$ Since $\kappa_i \le K_\ell$ for all
$i \in S$, we have
  $S \subseteq P_\ell.$
Therefore $|P_\ell| \ge m$. Since $j_1,\dots,j_m$ are the $m$ workers in
$P_\ell$ with the smallest values of $\tau_i$, we have
\begin{align*}
  \tau_{j_m}
  \le
  \max_{i \in S} \tau_i.
\end{align*}
Moreover, because $S_{\ell,m} \subseteq P_\ell$, we have
\begin{align*}
  \max_{p \in [m]} \kappa_{j_p}
  \le
  K_\ell
  =
  \max_{i \in S} \kappa_i.
\end{align*}
Thus,
\begin{align*}
  V_{\ell,m}
  &=
  \max\left\{
    \max_{p \in [m]} \kappa_{j_p},
    \beta_m \tau_{j_m}
  \right\} \le
  \max\left\{
    \max_{i \in S} \kappa_i,
    \beta_m \max_{i \in S} \tau_i
  \right\}.
\end{align*}
Since Algorithm~\ref{alg:worker-min} minimizes over all such pairs
$(\ell,m)$, we obtain
\begin{align*}
  \mathrm{OPT}
  \le
  \max\left\{
    \max_{i \in S} \kappa_i,
    \beta_m \max_{i \in S} \tau_i
  \right\}.
\end{align*}
Taking the minimum over all $S \subseteq [n]$ with $|S|=m$, and then over all
$m \in [n]$, gives
  $\mathrm{OPT}
  \le
  \min_{m \in [n]} \widehat{T}_{\textnormal{det}}(m).$
Combining the two inequalities, we conclude that
  $\mathrm{OPT}
  =
  \min_{m \in [n]} \widehat{T}_{\textnormal{det}}(m).$
It left to note that $T_{\textnormal{det}}(m) = \frac{L \Delta}{\varepsilon} \widehat{T}_{\textnormal{det}}(m)$ for all $m \in [n].$

Now we prove the second part of the theorem. Recall that 
\begin{align*}
  T_{\textnormal{optimal}}
  =
  \frac{L \Delta}{\varepsilon}
  \min_{m \in [n]}
  \max\left\{
    \max\{\kappa_{\pi_m}, \tau_{\pi_m}\},
    \frac{\sigma^2}{\varepsilon}
    \left(\sum_{i=1}^{m}\frac{1}{\tau_{\pi_i}}\right)^{-1}
  \right\},
\end{align*}
up to constant factors, where $\pi$ sorts $\{\max\{\kappa_i,\tau_i\}\}_{i \in [n]}$.
Let $\bar m$ be an optimal choice and define
\begin{align*}
  a_{\bar m}
  \eqdef
  \max\{\kappa_{\pi_{\bar m}}, \tau_{\pi_{\bar m}}\}, 
  A_{\bar m}
  \eqdef
  \sum_{i=1}^{\bar m}\frac{1}{\tau_{\pi_i}}, 
  \bar S
  \eqdef
  \{\pi_1,\dots,\pi_{\bar m}\}.
\end{align*}
Then
\begin{align*}
  T_{\textnormal{optimal}}
  =
  \frac{L \Delta}{\varepsilon}
  \max\left\{
    a_{\bar m},
    \frac{\sigma^2}{\varepsilon} A_{\bar m}^{-1}
  \right\}.
\end{align*}

Now reorder the elements of $\bar S$ so that
  $\tau_{q_1} \le \tau_{q_2} \le \dots \le \tau_{q_{\bar m}},
  \{q_1,\dots,q_{\bar m}\} = \bar S.$
For each $j \in [\bar m]$, define $S_j \eqdef \{q_1,\dots,q_j\}$. Then
\begin{align*}
  \max_{i \in S_j} \tau_i = \tau_{q_j},
  \qquad
  \max_{i \in S_j} \kappa_i \le a_{\bar m}.
\end{align*}
Choose
  $j^* \in \argmin_{j \in [\bar m]} \frac{\tau_{q_j}}{j}.$
Note that
\begin{align*}
  \frac{\tau_{q_{j^*}}}{j^*}
  \le
  \frac{H_{\bar m}}{A_{\bar m}},
  \qquad
  H_{\bar m} \eqdef \sum_{j=1}^{\bar m} \frac{1}{j}.
\end{align*}
Indeed,
\begin{align*}
  A_{\bar m}
  =
  \sum_{j=1}^{\bar m} \frac{1}{\tau_{q_j}}
  =
  \sum_{j=1}^{\bar m} \frac{j/\tau_{q_j}}{j}
  &\le
  \left(\max_{j \in [\bar m]} \frac{j}{\tau_{q_j}}\right)
  \sum_{j=1}^{\bar m} \frac{1}{j} = \frac{H_{\bar m}}{\min_{j \in [\bar m]} \tau_{q_j}/j}.
\end{align*}
Using the feasible choice $S_{j^*}$ in the definition of $T_{\textnormal{det}}$, we obtain
\begin{align*}
  \min_{m \in [n]} T_{\textnormal{det}}(m)
  &\le
  T_{\textnormal{det}}(j^*) \\
  &\le
  \frac{L \Delta}{\varepsilon}
  \max\left\{
    \max_{i \in S_{j^*}} \kappa_i,
    \max_{i \in S_{j^*}} \tau_i
    \max\left\{1, \frac{\sigma^2}{j^* \varepsilon}\right\}
  \right\} \\
  &=
  \frac{L \Delta}{\varepsilon}
  \max\left\{
    \max_{i \in S_{j^*}} \kappa_i,
    \tau_{q_{j^*}},
    \frac{\sigma^2 \tau_{q_{j^*}}}{j^* \varepsilon}
  \right\}.
\end{align*}
Since $S_{j^*} \subseteq \bar S$, we have
\begin{align*}
  \max_{i \in S_{j^*}} \kappa_i \le a_{\bar m},
  \qquad
  \tau_{q_{j^*}} \le a_{\bar m}.
\end{align*}
Therefore,
\begin{align*}
  \min_{m \in [n]} T_{\textnormal{det}}(m)
  &\le
  \frac{L \Delta}{\varepsilon}
  \max\left\{
    a_{\bar m},
    \frac{\sigma^2 \tau_{q_{j^*}}}{j^* \varepsilon}
  \right\} 
  \le
  \frac{L \Delta}{\varepsilon}
  \max\left\{
    a_{\bar m},
    \frac{\sigma^2}{\varepsilon}
    \frac{H_{\bar m}}{A_{\bar m}}
  \right\} 
  \le
  H_{\bar m}
  \frac{L \Delta}{\varepsilon}
  \max\left\{
    a_{\bar m},
    \frac{\sigma^2}{\varepsilon} A_{\bar m}^{-1}
  \right\}.
\end{align*}
By optimality of $\bar m$, this yields
  $\min_{m \in [n]} T_{\textnormal{det}}(m)
  \leq
  H_{\bar m} \, T_{\textnormal{optimal}}.$
Finally, since $H_{\bar m} \le 2 \log(n + 1),$ we get
\begin{align*}
  \min_{m \in [n]} T_{\textnormal{det}}(m) \leq 2 \log(n + 1) T_{\textnormal{optimal}}.
\end{align*}
\end{proof}

\section{Proof of Theorem~\ref{thm:random_lower_bound}}
\LOWERBOUND*
\begin{proof}
    We are slightly concise in our descriptions, since the proof largely repeats the ideas from \citep{arjevani2020second,carmon2021lower,arjevani2022lower,tyurin2023optimal,tyurin2024optimalgraph,tyurin2024shadowheart}. We analyze zero-respecting algorithms:
    \begin{definition}[Zero-respecting algorithm]
    \label{def:zerorespecting}
    At any moment of time or iteration, given the previously computed stochastic gradients 
    $\{\nabla f(\bar{x}^k)\}_{0 \leq k \leq \bar{k}},$ the algorithm can only produce a query point $x \in \mathbb{R}^d$ such that
    $$
    \textnormal{supp}(x) \subseteq \bigcup_{k=0}^{\bar{k}} \textnormal{supp}\bigl(\nabla f(\bar{x}^k)\bigr),
    $$
    where $\textnormal{supp}(v) := \{i \in [d] \mid v_i \neq 0\}$ denotes the support of a vector $v$.
    \end{definition}
    Basically, the proofs reduce to the analysis of the concentration of random variables. In particular, \citet{tyurin2023optimal} showed that the time required to find an $\varepsilon$--stationary point is lower bounded by
    \begin{align}
        \label{eq:LFBNKdEHXwdaNVHFdgvZ}
        \sum_{k=1}^{T} \min_{i \in [n]} \tau_i \eta_{ki}
    \end{align}
    if the computation times satisfy Assumption~\ref{ass:fixed_computation_model}, where $T = \Theta\left(\frac{L \Delta}{\varepsilon}\right)$ and $\eta_{ki},$ conditioned on $\{\eta_{k'i}\}_{i \in [n], k' < k},$ satisfies the geometric distribution with parameter $p = \Theta\left(\min\{\frac{\varepsilon}{\sigma^2}, 1\}\right)$ 
    for all $k \geq 1$ and $i \in [n].$ 

    Then, 
    \begin{align}
        \label{eq:pzrMcFiRrsWDAjcsN}
        \Exp{\sum_{k=1}^{T} \min_{i \in [n]} \tau_i \eta_{ki}} \geq \Omega\left(\min \limits_{m \in [n]} \left[\left(\frac{1}{m}\sum\limits_{i=1}^m \frac{1}{\tau_i}\right)^{-1}\left(\max\left\{\frac{L \Delta}{\varepsilon}, \frac{\sigma^2 L \Delta}{m \varepsilon^2}\right\}\right)\right]\right),
    \end{align}
    which follows from Theorem 6.4 \citep{tyurin2023optimal} since $\Exp{\mu} \geq t \times \Prob{\mu \geq t}$ for all non-negative $\mu$ random variables.
    
    However, under Assumption~\ref{ass:general_heter_random_computation_model}, we have to analyze 
    \begin{align}
        \label{eq:YGmxDqrrWlkrGVbG}
        T_{\textnormal{lower}} \eqdef \sum_{k=1}^{T} \min_{i \in [n]} \sum_{j=1}^{\eta_{ki}} \bar{\tau}_{kij},
    \end{align}
    where $\{\bar{\tau}_{kij}\}$ are i.i.d. random variables that satisfy Assumption~\ref{ass:general_heter_random_computation_model}. The lower bound \eqref{eq:YGmxDqrrWlkrGVbG} reduces to \eqref{eq:LFBNKdEHXwdaNVHFdgvZ}, when $\bar{\tau}_{kij} = \tau_i$ almost surely. Similarly to \citep{tyurin2023optimal}, the lower bound follows from the fact that workers have to ``uncover'' $T$ coordinates of the ``worst-case'' function \citep{carmon2020lower} sequentially, and the time to discover one coordinate is lower bounded by $\min_{i \in [n]} \sum_{j=1}^{\eta_{ki}} \bar{\tau}_{kij}$ since each worker calculate stochastic gradients in parallel and wait for the moment when one of them draw a ``lucky'' stochastic gradient. See details in \citep{arjevani2022lower,tyurin2023optimal}.
    
    It is left to bound $\Exp{T_{\textnormal{lower}}}.$ Notice that
    \begin{align}
        \label{eq:uxwqaGghZSqiICiKvI}
        \Exp{T_{\textnormal{lower}}} \geq \Exp{\sum_{k=1}^{T} \min_{i \in [n]} \tau_{i} \eta_{ki}} + \Exp{\sum_{k=1}^{T} \min_{i \in [n]} \sum_{j=1}^{\eta_{ki}} (\bar{\tau}_{kij} - \tau_{i})} 
        = \Exp{\sum_{k=1}^{T} \min_{i \in [n]} \tau_{i} \eta_{ki}} - \Exp{\sum_{k=1}^{T} \max_{i \in [n]} \sum_{j=1}^{\eta_{ki}} \left(\tau_{i} - \bar{\tau}_{kij}\right)}.
    \end{align}
    Using \eqref{eq:pzrMcFiRrsWDAjcsN}, we can bound the first term. Consider one of the elements from the second term. Using simple algebra,
    \begin{align*}
        &\Exp{\max_{i \in [n]} \sum_{j=1}^{\eta_{ki}} (\tau_{i} - \bar{\tau}_{kij})} = \frac{1}{\lambda}\Exp{\lambda \max_{i \in [n]} \sum_{j=1}^{\eta_{ki}} (\tau_{i} - \bar{\tau}_{kij})} \\
        &=\frac{1}{\lambda} \log \exp \left(\Exp{\lambda \max_{i \in [n]} \sum_{j=1}^{\eta_{ki}} (\tau_{i} - \bar{\tau}_{kij})}\right) \leq \frac{1}{\lambda} \log \Exp{\exp \left(\lambda \max_{i \in [n]} \sum_{j=1}^{\eta_{ki}} (\tau_{i} - \bar{\tau}_{kij})\right)},
    \end{align*}
    for all $\lambda > 0,$ where we use Jensen's inequality. Next, 
    \begin{align}
        \label{eq:TLfjg}
        &\Exp{\max_{i \in [n]} \sum_{j=1}^{\eta_{ki}} (\tau_{i} - \bar{\tau}_{kij})} \leq \frac{1}{\lambda} \log \left(\Exp{\sum_{i=1}^{n} \exp \left(\lambda \sum_{j=1}^{\eta_{ki}} (\tau_{i} - \bar{\tau}_{kij})\right)}\right) = \frac{1}{\lambda} \log \left(\sum_{i=1}^{n} \Exp{\exp \left(\lambda \sum_{j=1}^{\eta_{ki}} (\tau_{i} - \bar{\tau}_{kij})\right)}\right).
    \end{align}
    Next, 
    \begin{align*}
        \Exp{\exp \left(\lambda \sum_{j=1}^{\eta_{ki}} (\tau_{i} - \bar{\tau}_{kij})\right)} = \Exp{\ExpCond{\prod_{j=1}^{\eta_{ki}}\exp \left(\lambda (\tau_{i} - \bar{\tau}_{kij})\right)}{\eta_{ki}}}.
    \end{align*}
    By construction, $\eta_{ki}$ is the smallest index in a sequence of random Bernoulli variables that is constructed independently of ${\bar{\tau}_{kij}}.$ Moreover, ${\bar{\tau}_{kij}}$ are i.i.d. due to Assumption~\ref{ass:general_heter_random_computation_model}. Thus,
    \begin{align*}
        \Exp{\exp \left(\lambda \sum_{j=1}^{\eta_{ki}} (\tau_{i} - \bar{\tau}_{kij})\right)} = \Exp{\prod_{j=1}^{\eta_{ki}}\Exp{\exp \left(\lambda (\tau_{i} - \bar{\tau}_{kij})\right)}}.
    \end{align*}
    Since $\bar{\tau}_{kij}$ is sub-exponential and $\Exp{\bar{\tau}_{kij}} = \tau_{i}$, there exist a universal constant $c_1 \geq 1$ such that 
    \begin{align*}
        \Exp{\exp \left(\lambda \sum_{j=1}^{\eta_{ki}} (\tau_{i} - \bar{\tau}_{kij})\right)} \leq \Exp{\prod_{j=1}^{\eta_{ki}}e^{c_1^2 \lambda^2 R^2}} = \Exp{e^{c_1^2 \lambda^2 R^2 \eta_{ki}}}
    \end{align*}
    for all $\lambda \leq \frac{1}{c_1 R}$ \citep{vershynin2018high}. Now, since $\eta_{ki} \sim \textnormal{Geom}(p),$ we get
    \begin{align*}
        \Exp{\exp \left(\lambda \sum_{j=1}^{\eta_{ki}} (\tau_{i} - \bar{\tau}_{kij})\right)} \leq \Exp{e^{c_1^2 \lambda^2 R^2 \eta_{ki}}} = \frac{p e^{c_1^2 \lambda^2 R^2}}{1 - (1 - p) e^{c_1^2 \lambda^2 R^2}}
    \end{align*}
    for all $\lambda > 0$ such that $\lambda \leq \frac{1}{c_1 R}$ and $c_1^2 \lambda^2 R^2 < - \log(1 - p),$ where we take the moment-generating function (MGF) of $\eta_{ki},$ which is only defined for $c_1^2 \lambda^2 R^2 < - \log(1 - p).$ Taking $\lambda = \frac{\sqrt{p}}{2 c_1 R},$ we ensure that $\lambda \leq \frac{1}{c_1 R}$ and $c_1^2 \lambda^2 R^2 < - \log(1 - p).$ For this choice of $\lambda,$
    \begin{align*}
        \Exp{\exp \left(\lambda \sum_{j=1}^{\eta_{ki}} (\tau_{i} - \bar{\tau}_{kij})\right)} \leq 4
    \end{align*}
    because $1 - (1 - p) e^{c_1^2 \lambda^2 R^2} = 1 - (1 - p) e^{\frac{p}{4}} \geq 1 - (1 - p) (1 + \frac{p}{2}) \geq \frac{p}{2}$ for all $p \in (0, 1].$ Substituting this inequality and our choice of $\lambda$ to \eqref{eq:TLfjg},
    \begin{align*}
        &\Exp{\max_{i \in [n]} \sum_{j=1}^{\eta_{ki}} (\tau_{i} - \bar{\tau}_{kij})} \leq \frac{2 c_1 R}{\sqrt{p}} \log \left(4 n\right). 
    \end{align*}
    Using \eqref{eq:pzrMcFiRrsWDAjcsN}, \eqref{eq:uxwqaGghZSqiICiKvI}, $T = \Theta\left(\frac{L \Delta}{\varepsilon}\right),$ $p = \Theta\left(\min\left\{1, \frac{\varepsilon}{\sigma^2}\right\}\right),$ and the last inequality,
    \begin{align}
        \label{eq:GLKpSufk}
        \Exp{T_{\textnormal{lower}}} \geq \Omega\left(\min \limits_{m \in [n]} \left[\left(\frac{1}{m}\sum\limits_{i=1}^m \frac{1}{\tau_i}\right)^{-1}\left(\max\left\{\frac{L \Delta}{\varepsilon}, \frac{\sigma^2 L \Delta}{m \varepsilon^2}\right\}\right)\right]\right) - \Theta\left(\frac{R L \Delta}{\varepsilon}  \max\left\{1, \frac{\sigma}{\sqrt{\varepsilon}}\right\} \log \left(4 n\right)\right).
    \end{align}
    Since 
    \begin{align*}
        \min \limits_{m \in [n]} \left[\left(\frac{1}{m}\sum\limits_{i=1}^m \frac{1}{\tau_i}\right)^{-1}\left(\max\left\{\frac{L \Delta}{\varepsilon}, \frac{\sigma^2 L \Delta}{m \varepsilon^2}\right\}\right)\right] \geq \tau_1 \max\left\{\frac{L \Delta}{\varepsilon}, \frac{\sigma^2 L \Delta}{n \varepsilon^2}\right\}
    \end{align*}
    and we assume that
    \begin{align*}
        \tau_1 \max\left\{1, \frac{\sigma^2}{n \varepsilon}\right\} \geq \bar{c} \times R \max\left\{1, \frac{\sigma}{\sqrt{\varepsilon}}\right\} \log \left(n\right),
    \end{align*}
    where $\bar{c}$ is universal constant, then the term with $\Omega$ dominates the term with $\Theta$ in \eqref{eq:GLKpSufk} and we get 
    \begin{align*}
        \Exp{T_{\textnormal{lower}}} \geq \Omega\left(\min \limits_{m \in [n]} \left[\left(\frac{1}{m}\sum\limits_{i=1}^m \frac{1}{\tau_i}\right)^{-1}\left(\max\left\{\frac{L \Delta}{\varepsilon}, \frac{\sigma^2 L \Delta}{m \varepsilon^2}\right\}\right)\right]\right).
    \end{align*}
\end{proof}

\subsection{More theoretical examples}
\label{sec:more_examples}
Here we present additional examples for Section~\ref{sec:random}.

\paragraph{Chi-square distributions.}
Suppose that $\bar{\tau}_i \sim \chi^2_{k_i}$ (chi-square distributions) with $k_i \in \N$ for all $i \in [n]$. Then $\Exp{\bar{\tau}_i} = \tau_i = k_i$, and the chi-square distribution is $(\tau_i, R_i)$--sub-exponential with $R_i = \cO(\sqrt{k_i})$. As long as $\min_{i \in [n]} k_i \geq \max_{i \in [n]} \sqrt{k_i}$, we have $R \equiv \max_{i \in [n]} R_i = \cO(\tau_i)$ for all $i\in[n]$. Then, due to Corollary~\ref{cor:random}, $m$-Synchronous SGD is nearly optimal if, for instance, $\varepsilon$ is small.

\paragraph{Gamma distribution.}
Suppose that $\bar{\tau}_i \sim \text{Gamma}(a_i,b_i)$ with $a_i,b_i > 0$. In this case, $\Exp{\bar{\tau}_i} = \tau_i = a_i b_i$, and $\bar{\tau}_i$ is $(\tau_i, R_i)$--sub-exponential random variable with $R_i = \cO(\max\{\sqrt{a_i}, 1\}b_i)$ \citep{boucheron2003concentration}. Then, due to Corollary~\ref{cor:random}, $m$-Synchronous SGD is nearly optimal when $\min_{i \in [n]} {a_i b_i} \geq \max_{i \in [n]}\max\{\sqrt{a_i}, 1\}b_i$ and $\varepsilon$ is small, where the former true if, for instance, $\{a_i\}$ are large.

\paragraph{Shifted exponential distribution.}
Assume that $\bar{\eta}_i$ is from the exponential distribution $\textnormal{Exp}(\lambda_i)$ for all $i \in [n]$. We define the shifted exponential random variable as $\bar{\tau}_i = \mu_i + \bar{\eta}_i$ for some $\mu_i \geq 0.$ In this case, we can take $R = \Theta(\max_{i \in [n]} (1 / \lambda_i))$ and $\tau_i = \mu_i + 1 / \lambda_i$ for all $i \in [n].$ Thus, as long as $\max_{i \in [n]} (1 / \lambda_i) \leq \min_{i \in [n]} (\mu_i + 1 / \lambda_i),$ $m$-Synchronous SGD is nearly optimal when, for instance, $\varepsilon$ is small.


\section{Proof of Proposition \ref{prop:g}}
\begin{proof}
    Consider 
    \begin{align*}
      g(m) \eqdef \tau_m \max\left\{1, \frac{\sigma^2}{m \varepsilon}\right\}
    \end{align*}
    for all $m \in [n].$ 
    For all $m \geq \ceil{\frac{\sigma^2}{\varepsilon}},$ 
    $g(m) = \tau_m,$ which is non-decreasing. Thus, it is sufficient to restrict the function to $m \leq \min\left\{\ceil{\nicefrac{\sigma^2}{\varepsilon}}, n\right\}$ to find its minimizer. 
    
    For all $m \leq \min\left\{\ceil{\nicefrac{\sigma^2}{\varepsilon}}, n\right\},$ if $\frac{\sigma^2}{\varepsilon} \leq 1,$ then $m = 1.$ Let $\frac{\sigma^2}{\varepsilon} > 1.$ If $m \leq \min\left\{\ceil{\nicefrac{\sigma^2}{\varepsilon}}, n\right\},$ then clearly $m \leq \ceil{\nicefrac{\sigma^2}{\varepsilon}}.$ Notice that
    \begin{align*}
      \tau_m \frac{\sigma^2}{m \varepsilon} \leq \tau_m \max\left\{1, \frac{\sigma^2}{m \varepsilon}\right\} \leq 2 \tau_m \frac{\sigma^2}{m \varepsilon}
    \end{align*}
    since $m \leq \ceil{\nicefrac{\sigma^2}{\varepsilon}} \leq \nicefrac{2 \sigma^2}{\varepsilon}.$
\end{proof}

\section{Proof of Theorems~\ref{thm:sync_sgd_time_univ}}
\TIMECOMPLEXITYUNIVERSALBASE*
\begin{proof}
    Algorithm~\ref{alg:alg_server_m_star} starts at time $\bar{t}_0 = t_0 = 0.$ Let us define $t_k$ as the first time moment when Algorithm~\ref{alg:alg_server_m_star} finishes the $k$\textsubscript{th} iteration. Then, if 
    \begin{align*}
        \bar{t}_1 = \min \left\{t \geq 0 \,:\, \max_{S \subseteq [n], \abs{S} = m} \min_{i \in S} N_i(\bar{t}_{0},t) = 2\right\},
    \end{align*}
    then $\bar{t}_1 \geq t_1$ because it is sufficient for the algorithm to wait for the fastest $m$ workers. Using mathematical induction, assume that $\bar{t}_k \geq t_k.$ The $k + 1$\textsuperscript{th} iteration starts at time $t_k,$ and the algorithm waits for the first $m$ workers to finish their calculations. In the worst case, worker $i$ has to finish calculating the previous stochastic gradient (requested before time $t_k$ and not used), and start calculating the new stochastic gradient after time $t_k.$ Thus, worker $i$ will finish calculating the first stochastic gradient started after $t_k$ after $\hat{t}$ seconds such that $N_i(\bar{t}_{k},\hat{t}) \geq 2,$ and therefore 
    \begin{align*}
        \bar{t}_{k+1} = \min \left\{t \geq 0 \,:\, \max_{S \subseteq [n], \abs{S} = m} \min_{i \in S} N_i(\bar{t}_{k},t) = 2\right\}
    \end{align*}
    is an upper bound for $t_{k+1}.$
\end{proof}

\section{Proof of Theorem~\ref{thm:partial}}
\label{sec:example}
\TIMECOMPLEXITYPARTIAL*
\begin{proof}
In Assumption~\ref{ass:partial}, we assume that, for all $t \geq 0,$ there exists a subset $S_t$ of at least size $(1 - p)n$ such that $v_i(t) = v \in \R_+$ for all $i \in S_t,$ $v_i(t) \leq v$ for all $i \not\in S_t.$ Using Theorem~\ref{thm:sync_sgd_time_univ}, we now derive an explicit formula for the time complexity of $m$-Synchronous SGD (Algorithm~\ref{alg:alg_server_m_star}). 

Notice that 
\begin{align*}
    \sum_{i=1}^{n} \int_{\bar{t}_k}^{t} v_i(t) dt = \int_{\bar{t}_k}^{t} \sum_{i=1}^{n} v_i(t) dt \geq (1 - p) n v (t - t_k)
\end{align*}
for all $t \geq \bar{t}_k.$ At the same time, 
\begin{align*}
    \int_{\bar{t}_k}^{t} v_i(t) dt \leq (t - \bar{t}_k) v
\end{align*}
for all $t \geq \bar{t}_k$ and $i \in [n].$ For any fixed $t > \bar{t}_k,$ let $\bar{n}$ be the number of workers such that
\begin{align*}
    \int_{\bar{t}_k}^{t} v_i(t) dt < \frac{1}{2} \times (t - \bar{t}_k) v.
\end{align*}
Then 
\begin{align*}
    (n - \bar{n}) v (t - \bar{t}_k) + \frac{1}{2} \bar{n} v (t - \bar{t}_k) \geq (1 - p) n v (t - \bar{t}_k) \Leftrightarrow \bar{n} \leq 2 p n.
\end{align*}
Therefore, for all $t > \bar{t}_k,$ at least $n - 2 p n = (1 - 2 p) n$ workers satisfy the inequality
\begin{align*}
    \int_{\bar{t}_k}^{t} v_i(t) dt \geq \frac{1}{2} (t - \bar{t}_k) v.
\end{align*}
Thus, for all $t > \bar{t}_k,$ $p < 0.4,$ and $m \in \{\flr{n / 5}, \dots, \flr{(1 - 2 p) n}\},$ there exists a set $S$ of size $m$ such that 
\begin{align*}
    \min_{i \in S} N_i(\bar{t}_{k},t) \geq \frac{1}{2} (t - \bar{t}_k) v,
\end{align*}
meaning that $\bar{t}_{k+1} \leq \bar{t}_{k} + \frac{4}{v}.$ Finally,
\begin{align*}
    \bar{t}_{\bar{K}} = \cO\left(\frac{1}{v} \max\left\{\frac{L \Delta}{\varepsilon}, \frac{\sigma^2 L \Delta}{m \varepsilon^2}\right\}\right).
\end{align*}
It is left to notice that $m = \Theta(n)$ for all $m \in \{\flr{n / 5}, \dots, \flr{(1 - 2 p) n}\}$ and $p < 0.4.$
\end{proof}


\section{Proof of Proposition~\ref{prop:tau}}
\label{sec:derive_m}
\begin{proof}
We consider $\tau_m = \tau_1 m^{\alpha} + \delta_m$ under the conditions $0 \leq \alpha \leq 1$ and $0 \leq \delta_m \leq \delta.$ 
In this case, we have to find $m$ that minimizes
\begin{align*}
    h(m) \eqdef \frac{\tau_m}{m}
\end{align*}
for all $m \leq \min\left\{\ceil{\nicefrac{\sigma^2}{\varepsilon}}, n\right\}$ due to Proposition~\ref{prop:g}. We have
\begin{align*}
    h(m) = \frac{\tau_1 m^{\alpha} + \delta_m}{m} = \Theta\left(\frac{\tau_1 m^{\alpha}}{m}\right)
\end{align*}
for all $m \geq \left(\frac{\delta}{\tau_1}\right)^{1 / \alpha}.$ Thus, $m = \min\left\{\ceil{\nicefrac{\sigma^2}{\varepsilon}}, n\right\}$ is a minimizer on the interval $\left[\left(\frac{\delta}{\tau_1}\right)^{1 / \alpha}, \min\left\{\ceil{\nicefrac{\sigma^2}{\varepsilon}}, n\right\}\right]$ and 
\begin{align*}
    h(m) = \Theta\left(\frac{\tau_1}{m^{1 - \alpha}}\right)
\end{align*}
for $m = \min\left\{\ceil{\nicefrac{\sigma^2}{\varepsilon}}, n\right\}.$

For all $m \leq \left(\frac{\delta}{\tau_1}\right)^{1 / \alpha},$
\begin{align*}
    h(m) \geq \frac{\tau_1}{m^{1 - \alpha}} \geq \frac{\tau_1^{1 / \alpha}}{\delta^{(1 - \alpha) / \alpha}}.
\end{align*}

Combining both cases, we can conclude that $m = \min\left\{\ceil{\nicefrac{\sigma^2}{\varepsilon}}, n\right\}$ minimizes $h$ when $m \geq \left(\frac{\delta}{\tau_1}\right)^{1 / \alpha}$ (either $n$ or $\nicefrac{\sigma^2}{\varepsilon}$ is large).
\end{proof}

\section{A Theoretical Example Where Asynchronicity Is Needed}
\label{sec:couter_example}
In this section, we formalize the example from Section~\ref{sec:couter_exmaple_main}. In order to formalize it, we can define the computation powers as
\begin{align*}
    v_i(t) = v
\end{align*}
for all $i > 1$ and $t \geq 0,$ and
\begin{align*}
    v_1(t) = \begin{cases}
        v, & t \leq \bar{t} \eqdef \frac{1}{v}, \\
        \infty, & t > \bar{t}
    \end{cases} 
\end{align*}
for the first worker. All workers except the first one have constant computational power. The first worker has the constant computation power for the first $\bar{t}$ seconds, and then becomes infinitely fast. 

In this example, knowing only the behavior of the first worker for at most the first $\bar{t}$ seconds, the best choice of $m$ in Algorithm~\ref{alg:alg_server_m_star} is $m = n,$ and the time complexity is
\begin{align}
    \label{eq:WQiKkEz}
    \Theta\left(\frac{1}{v} \max\left\{\frac{L \Delta}{\varepsilon}, \frac{\sigma^2 L \Delta}{n \varepsilon^2}\right\}\right)
\end{align}
because the algorithm always waits for all the workers, while the first one quickly computes one stochastic gradient and then is idle during each iteration.

At the same time, according to Theorem~\ref{thm:univ_lower_bound}, the lower bound and the time complexity of the optimal asynchronous methods can be much faster. Indeed, notice that
\begin{align*}
    \underline{t}_1 
    &= \min\left\{t \, : \, \sum\limits_{i=1}^n \flr{\int_{0}^{t} v_i(\tau) d \tau} \geq c_2 \ceil{\frac{\sigma^2}{\varepsilon}}\right\} \\
    &= \min\left\{t \, : \, (n - 1) \flr{t v} + \mathbf{1}[t \leq \bar{t}] \times \flr{t v} + \mathbf{1}[t > \bar{t}] \times \infty \geq c_2 \ceil{\frac{\sigma^2}{\varepsilon}}\right\},
\end{align*}
where we assume that $0 \times \infty = 0.$ 
Therefore, $\underline{t}_1 \leq 2 \bar{t},$ and 
\begin{align*}
    \underline{t}_2 
    &= \min\left\{t \, : \, \sum\limits_{i=1}^n N_i(\underline{t}_{1},t) \geq c_2 \ceil{\frac{\sigma^2}{\varepsilon}}\right\} \leq \min\left\{t \, : \, \sum\limits_{i=1}^n N_i(2 \bar{t},t) \geq c_2 \ceil{\frac{\sigma^2}{\varepsilon}}\right\} \leq 2 \bar{t} + \delta
\end{align*}
for all $\delta > 0$ because
$\sum\limits_{i=1}^n N_i(2 \bar{t},2 \bar{t} + \delta) \geq N_1(2 \bar{t},2 \bar{t} + \delta) = \infty.$
Using the same reasoning, we can conclude that
\begin{align*}
    \underline{t}_{\underline{K}} \leq \Theta\left(\bar{t} + \delta \frac{L \Delta}{\varepsilon}\right).
\end{align*}
Taking $\delta = \frac{\varepsilon \bar{t}}{L \Delta},$
\begin{align*}
    \underline{t}_{\underline{K}} \leq \Theta\left(\bar{t}\right) = \Theta\left(\frac{1}{v}\right),
\end{align*}
which can arbitrarily smaller than \eqref{eq:WQiKkEz}.

\section{NanoGPT Benchmark}
\label{sec:nanogpt}
In this section, we numerically evaluate the computation times for the NanoGPT model to understand the difference between the noise factor $R$ and the mean computation time $\tau_1$ in Assumption~\ref{ass:general_heter_random_computation_model}. We use the NanoGPT model \citep{nanogpt} with \texttt{vocab\_size=50304}, \texttt{block\_size=512}, \texttt{n\_layer=6}, \texttt{n\_head=6}, and \texttt{n\_embd=384}. We use \texttt{Tesla V100-SXM3-32GB}. We first run 10 warmup steps and then 200 backpropagation steps to estimate the mean and the variance. Empirically, we observe that the mean is $\approx 72.2$ ms, while the parameter $R$ is\footnote{we numerically found the smallest $R$ such that $\frac{1}{200} \sum_{j=1}^{200} \Exp{\exp(\abs{\bar{\tau}_j - \tilde{\tau}} / R)} = 2,$ where $\bar{\tau}_j$ are the recorder times and $\tilde{\tau}$ is the estimated mean.} $\approx 0.6.$ Thus, $R \log n \approx 0.6 \log(1,000,000) = 8.2 \ll 71.2$ even if the number of devices is $n = 1,000,000.$ In Figure~\ref{fig:qqplot}, we also present a Q–Q plot, where the computation times are plotted against a normal distribution; we see that the computation times are visually close to normal.
\begin{figure}[h]
    \centering
    \includegraphics[width=\linewidth]{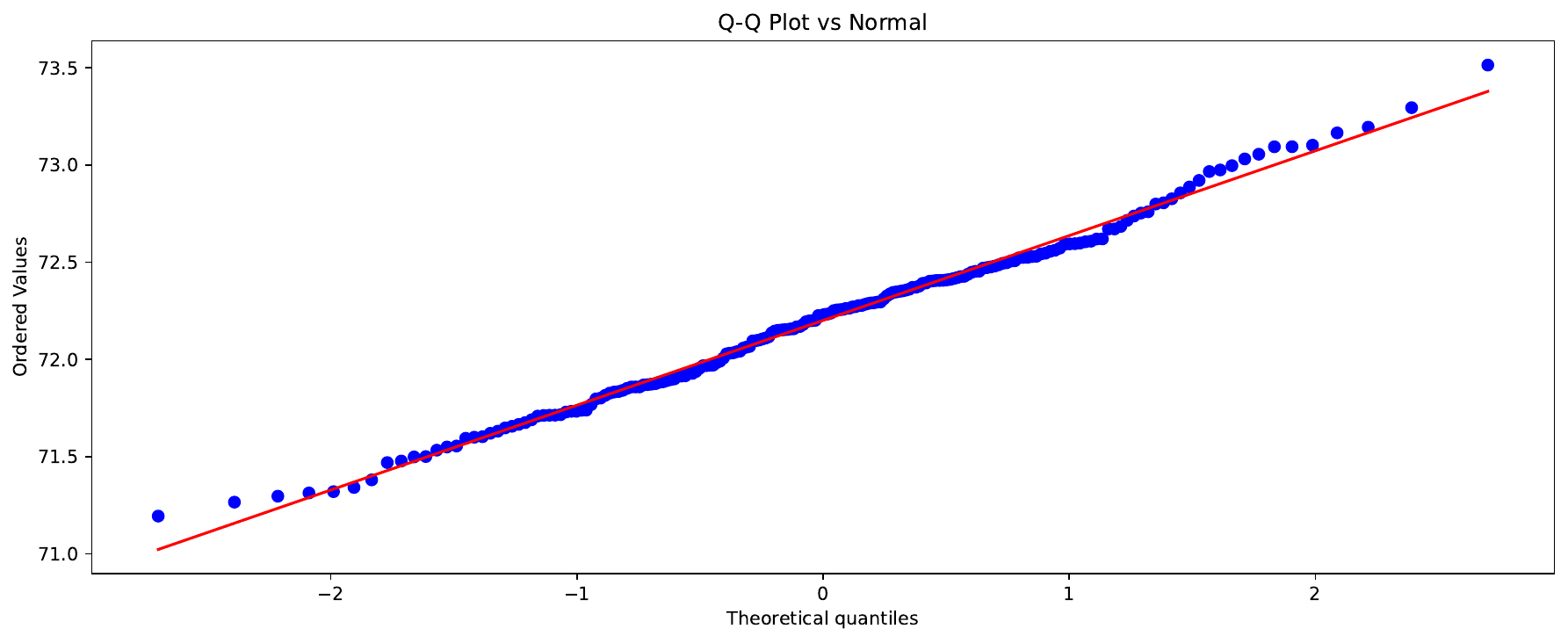}
    \caption{Q-Q Plot of the data from Section~\ref{sec:nanogpt} versus a normal distribution.}
    \label{fig:qqplot}
\end{figure}


\clearpage
\newpage
\section{Extra Experiments and Details}
\label{sec:exp_appendix}
The code was prepared in Python~3 and executed on a machine with 52~CPUs (Intel(R) Xeon(R) Gold~6278C @~2.60,GHz).

We conduct three types of experiments: (i) we consider a synthetic quadratic optimization task and controlled computation times, where we verify our theoretical results and compare the methods; (ii) then, we verify different computation times on small-scale real machine learning problems; (iii) finally, we implement a custom version in PyTorch and compare Syncrounous SGD and Asyncrounous SGD on the NanoGPT training task.

We first compare the methods on synthetic quadratic optimization problems in which the noise in the computation times of stochastic gradients can be precisely controlled. Consider the problem
\begin{align*}
f(x) = \frac{1}{2} x^\top \mathbf{A} x - b^\top x
\end{align*}

for all $x \in \mathbb{R}^d$ and take $d = 1000$,
\begin{align*}
\mathbf{A} = \frac{1}{4}
\begin{pmatrix}
2 & -1 &        & 0 \\
-1 & \ddots & \ddots &   \\
   & \ddots & \ddots & -1 \\
0  &        & -1 & 2
\end{pmatrix}
\in \mathbb{R}^{d \times d}
\quad\text{and,}\quad
b = \frac{1}{4}
\begin{bmatrix}
-1 \\ 0 \\ \vdots \\ 0
\end{bmatrix}
\in \mathbb{R}^d.
\end{align*}

We define function $\operatorname{prog}(x) := {\max\left\{i \geq 0\,|\, x_i \neq 0\right\}} \quad (x_0 \equiv 0)$ 
and consider the following stochastic gradients:
\begin{equation}
\label{eq:zTeQKwqxLsZHYPWCC}
\left[ \nabla f(x; \xi) \right]_j
:= \nabla_j f(x) \left( 1 + \mathbf{1}\!\left[ j > \operatorname{prog}(x) \right] \left( \frac{\xi}{p} - 1 \right) \right)
\quad \forall x \in \mathbb{R}^d,
\end{equation}
where $\xi \sim \operatorname{Bernoulli}(p)$ for all $i \in [n]$, and $p \in (0,1]$ 
is a parameter controlling the noise level: the smaller $p$ the stronger the noise in the stochastic gradients.
We use $[x]_j$ to denote the $j$-th component of a vector $x \in \mathbb{R}^d$. 
In our experiments, we initialize the optimization at $x^0 = [\sqrt{d}, 0, \dots, 0]^\top$. In all methods, we performed a grid search over the step size $\gamma \in \left\{2^{-16}, \dots, 2^4\right\}.$ In Rennala SGD, we tune the batch size parameter from the set $\left\{1, 5, 10, 15, 20, \dots, n\right\}$. In $m$-Synchronous, we also tune the number of active workers $m$ from the set $\left\{1, 5, 10, 15, 20, \dots, n\right\}$.

\subsection{Quadratic optimization problem with heterogeneous times}
\label{sec:heter_times}
In Figure~\ref{fig:gradient_noise}, we consider the fixed computation model (Assumption~\ref{ass:fixed_computation_model}) with the number of workers $n = 1000$. In each row of Figure~\ref{fig:gradient_noise}, we consider one computational scenario, and in each column, we fix the probability from \eqref{eq:zTeQKwqxLsZHYPWCC}. The main observation is that $m$-Synchronous SGD with a proper $m$, Asynchronous SGD, and Rennala SGD are competitive, in accordance with Section~\ref{sec:sync_optimal}. 

At the same time, Synchronous SGD does not match other methods in some of the plots. Consider the first row, when worker $i$ requires $\tau_i = \sqrt{i}$ seconds. Recall that Proposition \ref{prop:tau} implies that the optimal number of workers in Algorithm~\ref{alg:alg_server_m_star} is $m = \min\left\{\ceil{\nicefrac{\sigma^2}{\varepsilon}}, n\right\}$. Consequently, the number of active workers should increase as the gradient noise grows, until the moment when $\nicefrac{\sigma^2}{\varepsilon} \geq n.$ In the first row, we see exactly the predicted behavior: when $\nicefrac{\sigma^2}{\varepsilon} \ll n$ (Fig.~\ref{fig:plot1}), Synchronous SGD is slower because the optimal $m \ll n.$ However, once we start increasing the variance by taking a smaller $p,$ we can see that the gap between Synchronous SGD and other methods decreases since the optimal choice of $m$ becomes $n.$ In the second row, we see a similar behavior. In the third row, when worker $i$ requires $\tau_i = i^{1.2}$ seconds, the optimal choice of $m$ is $1$ due to Proposition~\ref{prop:g}, and Synchronous SGD with all participating workers cannot match other methods.




\begin{figure}[htbp]
    \centering
    \begin{subfigure}[b]{0.32\textwidth}
        \centering
        \includegraphics[width=\linewidth]{\experimentone}
        \caption{$p=0.01, n=10^3, \tau_i = \sqrt{i}$}
        \label{fig:plot1}
    \end{subfigure}
    \hfill
    \begin{subfigure}[b]{0.32\textwidth}
        \centering
        \includegraphics[width=\linewidth]{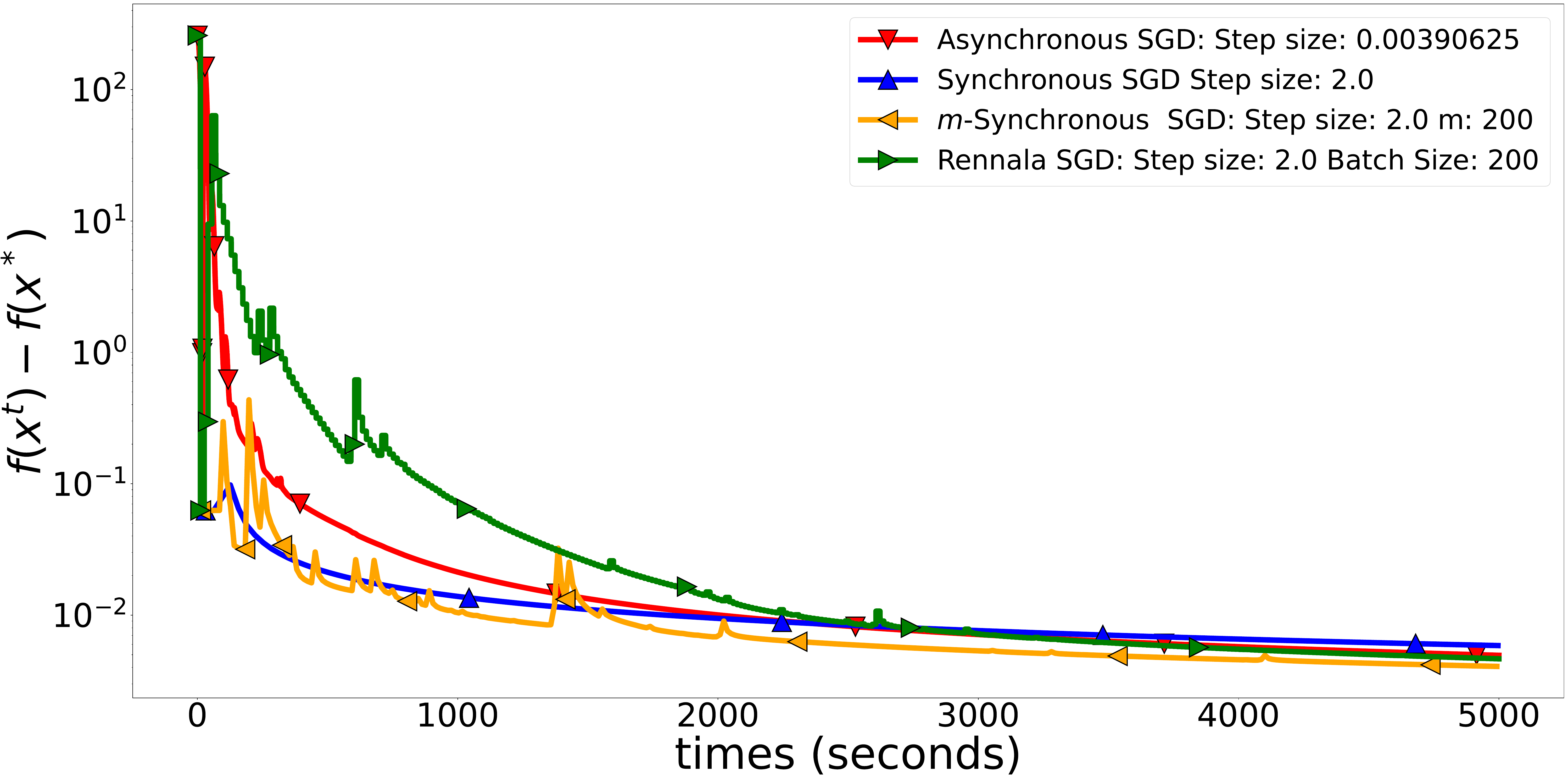}
        \caption{$p=0.001, n=10^3, \tau_i = \sqrt{i}$}
        \label{fig:plot2}
    \end{subfigure}
    \hfill
    \begin{subfigure}[b]{0.32\textwidth}
        \centering
        \includegraphics[width=\linewidth]{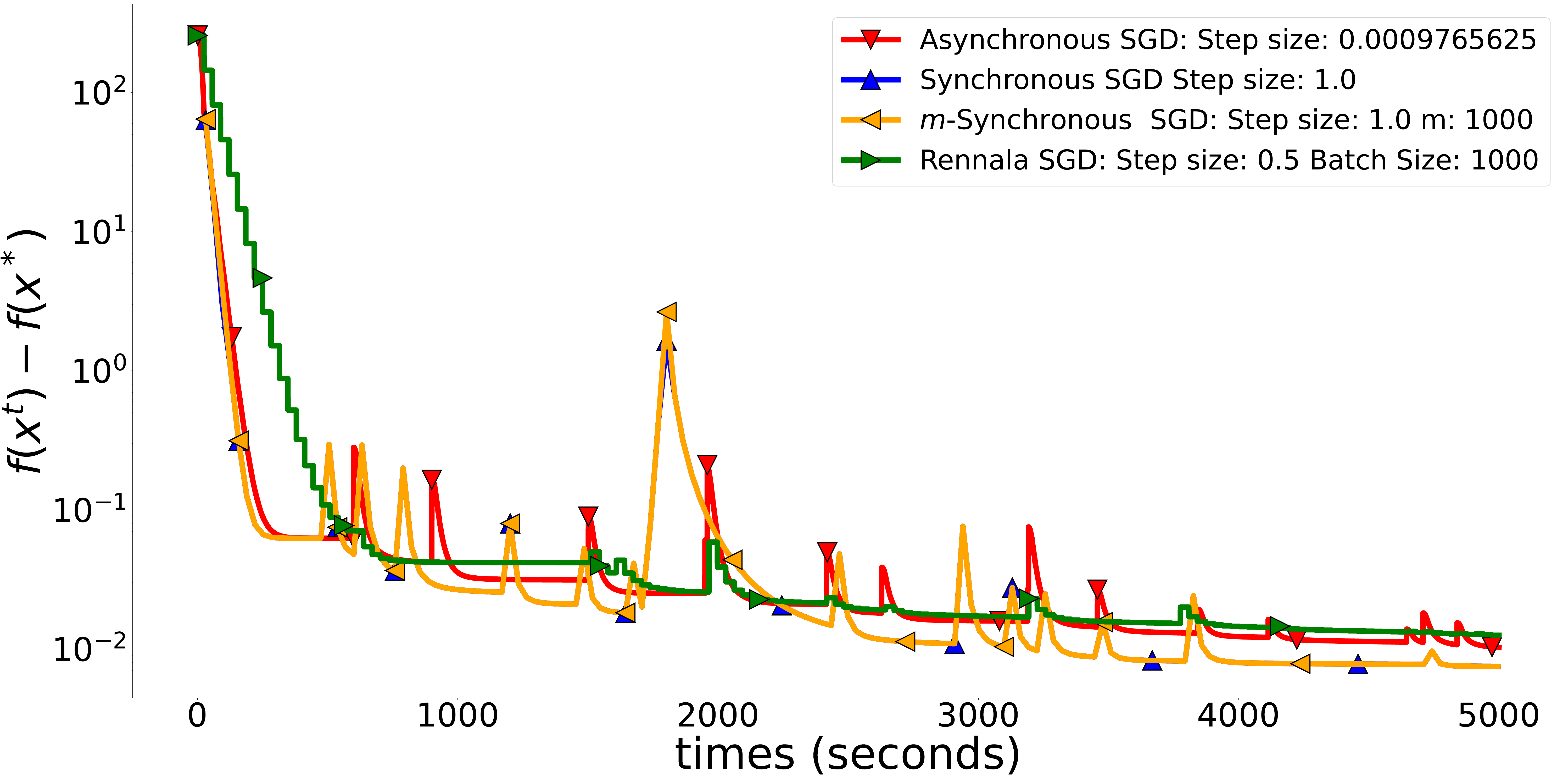}
        \caption{$p=0.0001, n=10^3, \tau_i = \sqrt{i}$}
        \label{fig:plot3}
    \end{subfigure}

    \vspace{1em} 

    \begin{subfigure}[b]{0.32\textwidth}
        \centering
        \includegraphics[width=\linewidth]{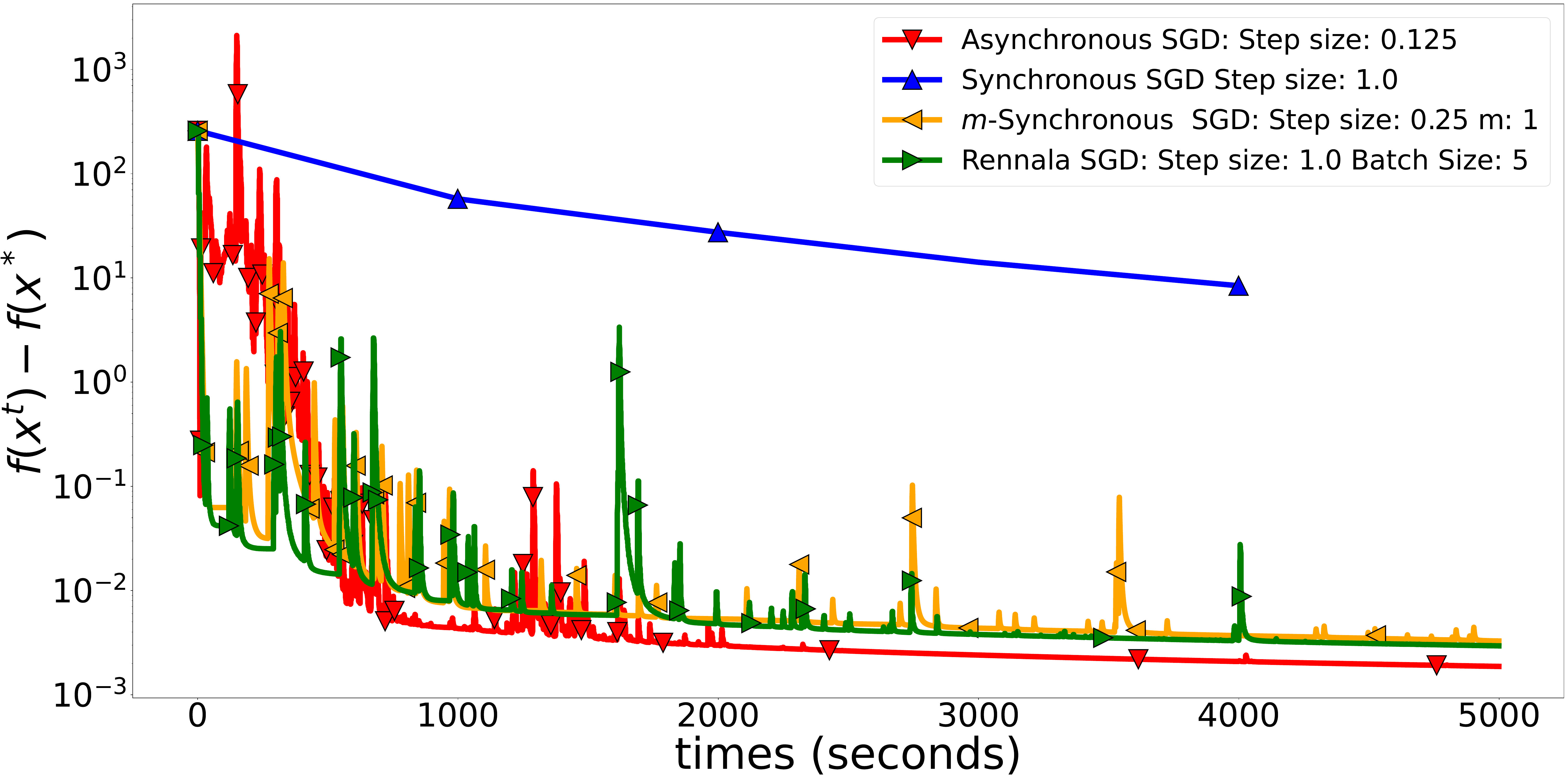}
        \caption{$p=0.01, n=10^3, \tau_i = i$}
        \label{fig:plot4}
    \end{subfigure}
    \hfill
    \begin{subfigure}[b]{0.32\textwidth}
        \centering
        \includegraphics[width=\linewidth]{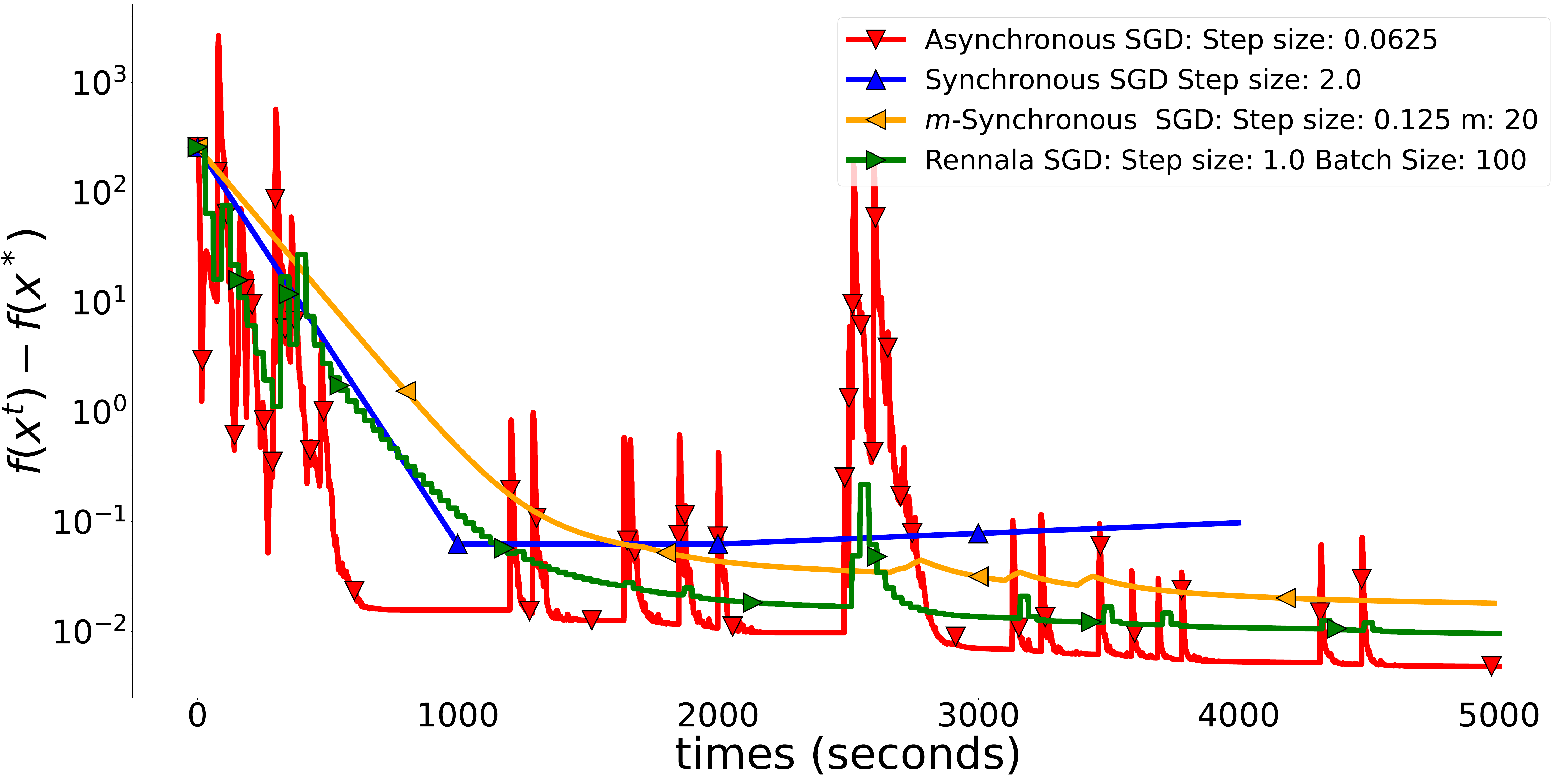}
        \caption{$p=0.001, n=10^3, \tau_i = i$}
        \label{fig:plot5}
    \end{subfigure}
    \hfill
    \begin{subfigure}[b]{0.32\textwidth}
        \centering
        \includegraphics[width=\linewidth]{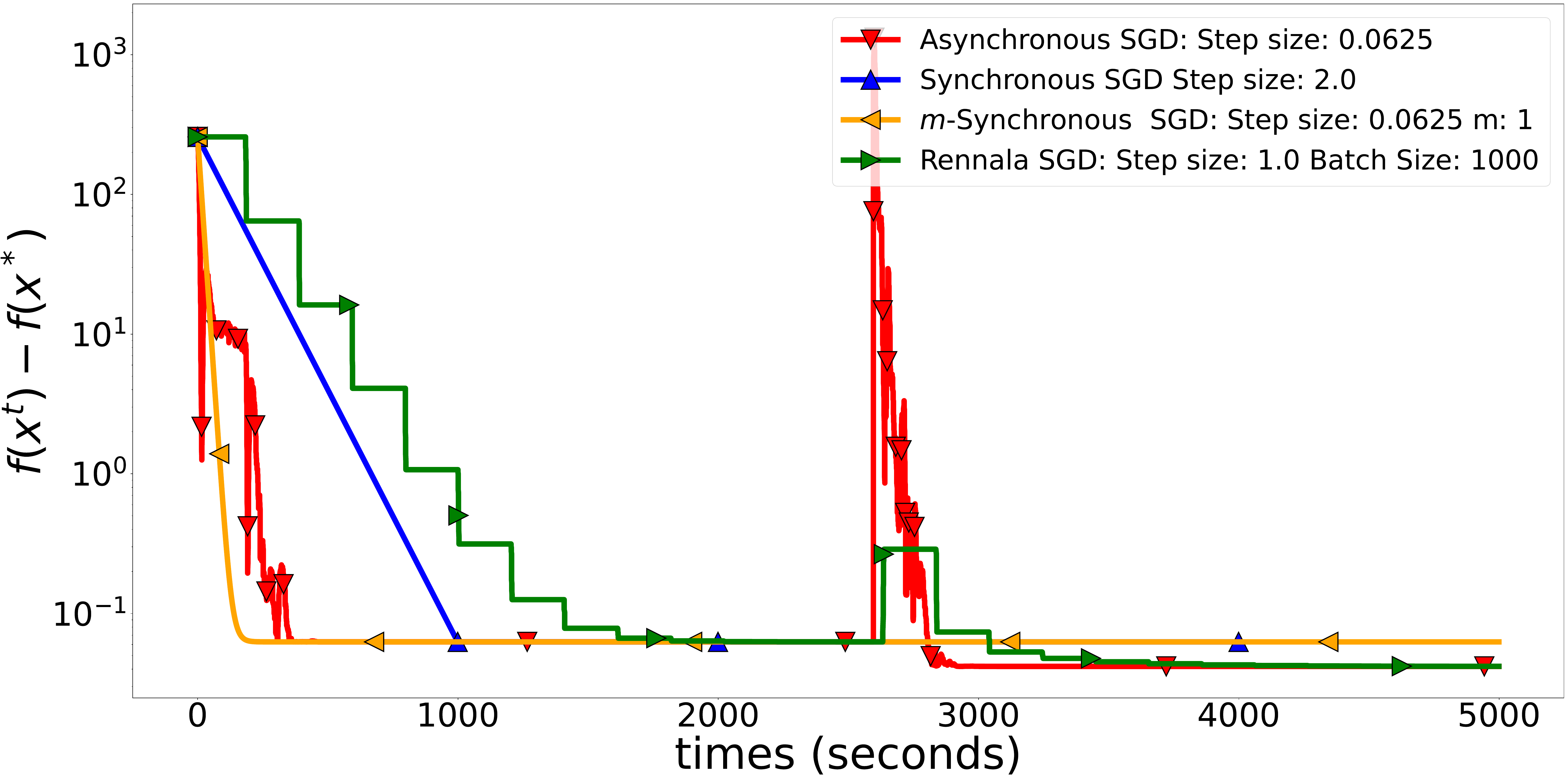}
        \caption{$p=0.0001, n=10^3, \tau_i = i$}
        \label{fig:plot6}
    \end{subfigure}

    \vspace{1em}

    \begin{subfigure}[b]{0.32\textwidth}
        \centering
        \includegraphics[width=\linewidth]{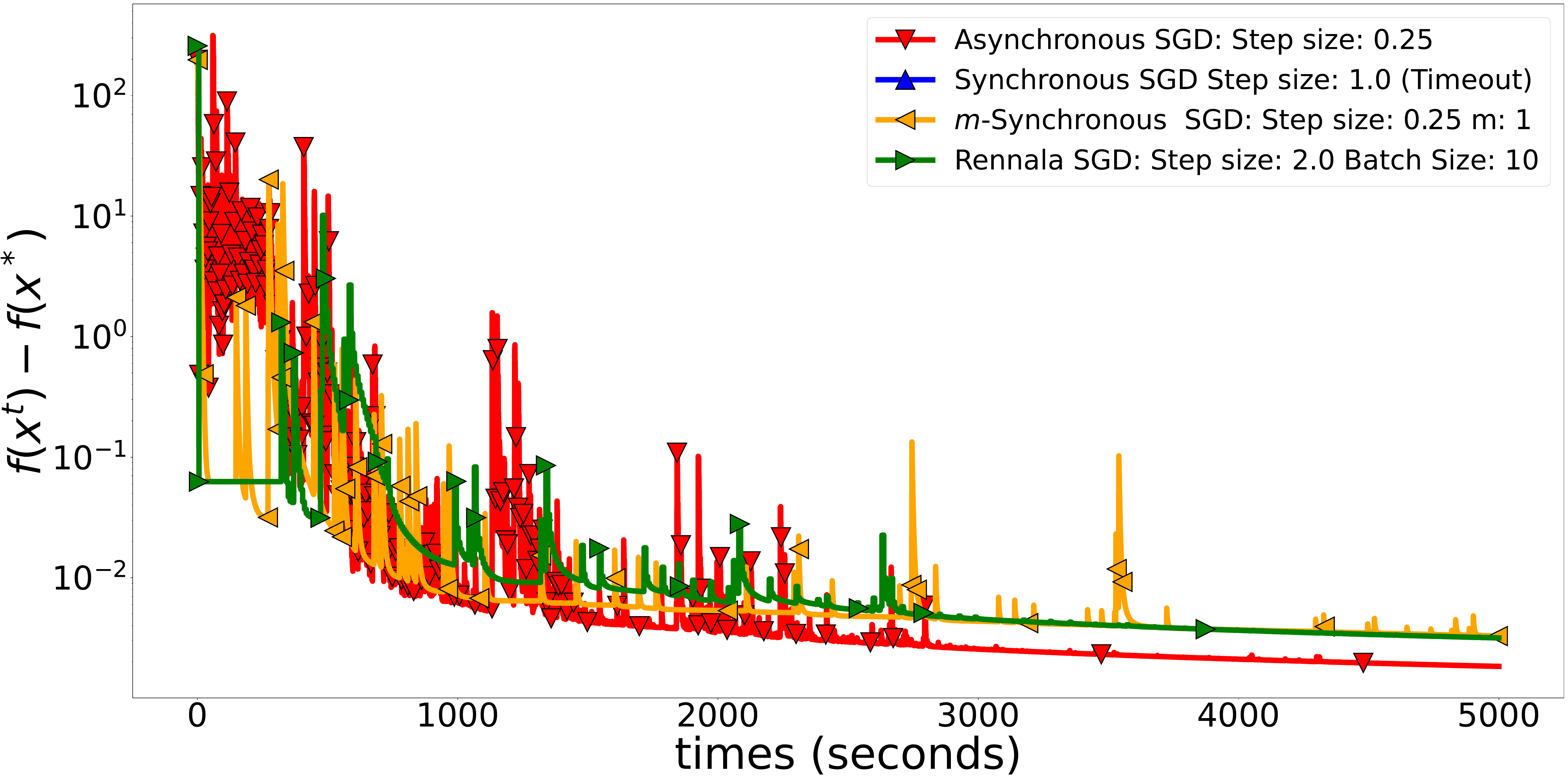}
        \caption{$p=0.01, n=10^3, \tau_i = i^{1.2}$}
        \label{fig:plot7}
    \end{subfigure}
    \hfill
    \begin{subfigure}[b]{0.32\textwidth}
        \centering
        \includegraphics[width=\linewidth]{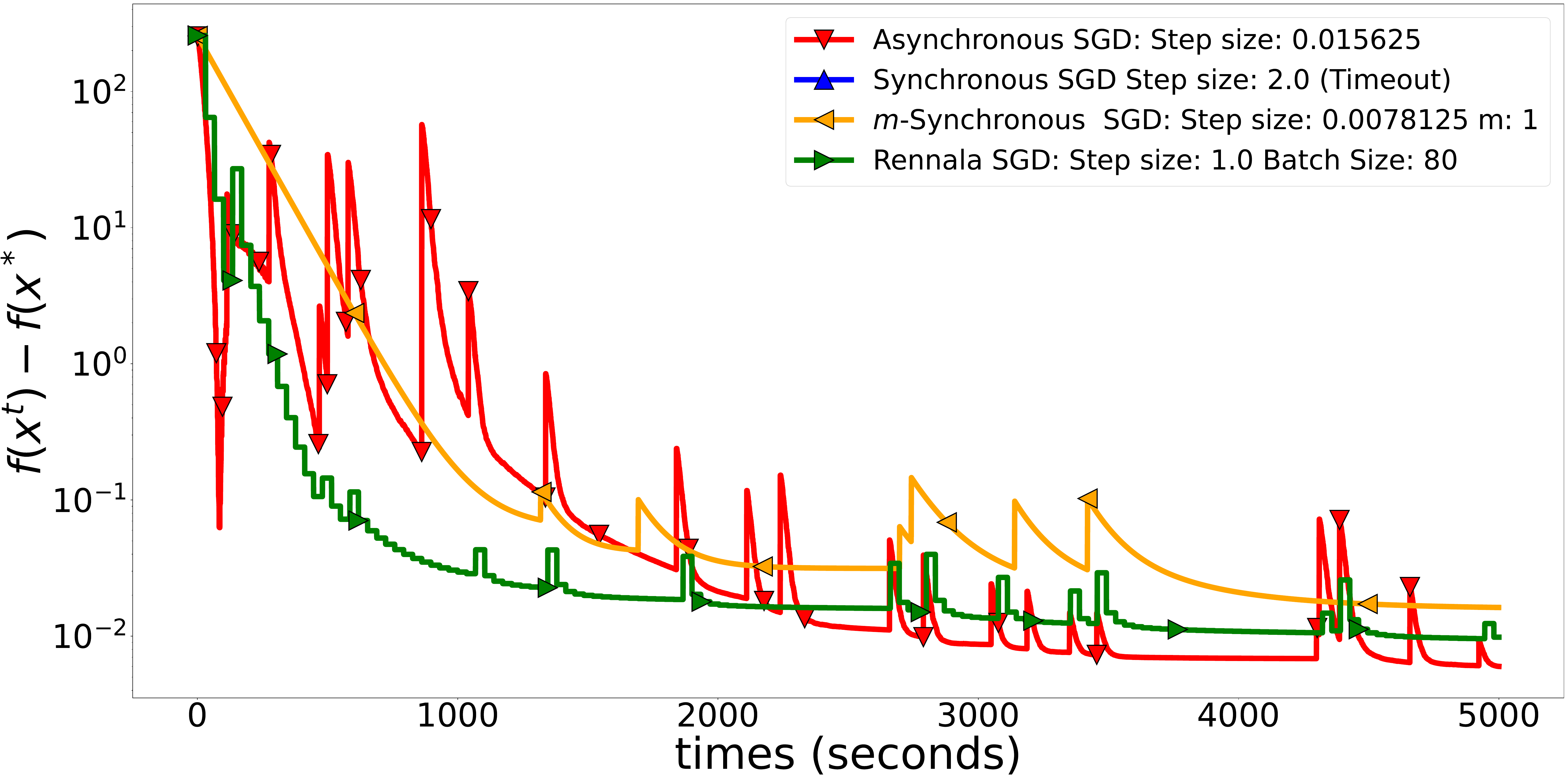}
        \caption{$p=0.001, n=10^3, \tau_i = i^{1.2}$}
    \end{subfigure}
    \hfill
    \begin{subfigure}[b]{0.32\textwidth}
        \centering
        \includegraphics[width=\linewidth]{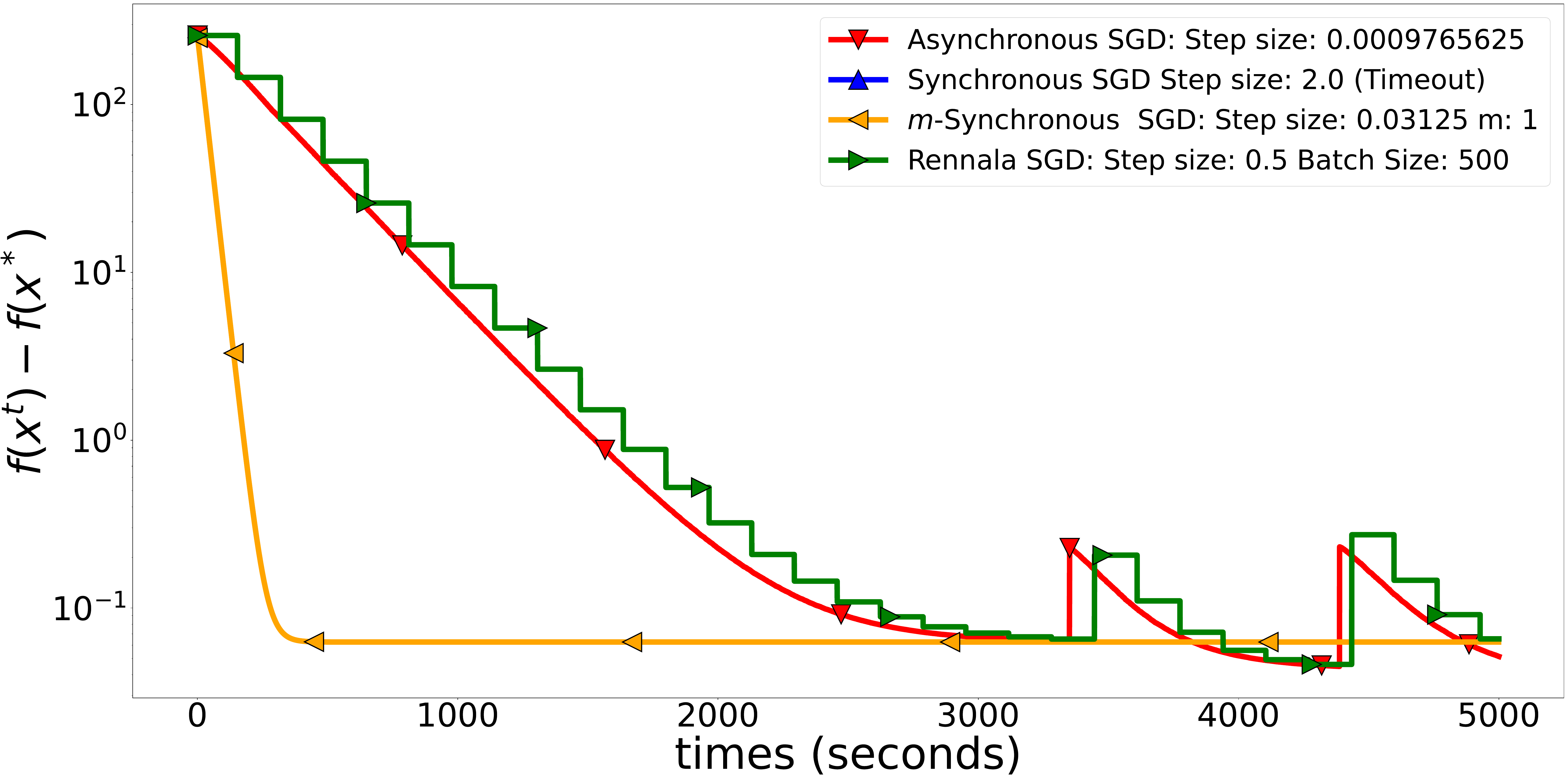}
        \caption{$p=0.0001, n=10^3, \tau_i = i^{1.2}$}
        \label{fig:plot9}
    \end{subfigure}
    \caption{Comparison of the methods on the quadratic optimization problem under different computation time scenarios and noise levels.}
    \label{fig:gradient_noise}
\end{figure}
\clearpage
\subsection{Quadratic optimization with increasing number of workers}
In the previous section, we fix the number of workers to $n = 1000.$ In this section, we compare the behavior of methods with $n \in \{100, 1000, 10000\}.$ In Figure~\ref{fig:scale_n}, we observe that $m$-Synchronous, as expected, is robust to the increasing number of workers, while the gap between Synchronous SGD and other methods increases.
\begin{figure}[htbp]
    \centering
    \begin{subfigure}[b]{0.3\textwidth}
        \centering
        \includegraphics[width=\linewidth]{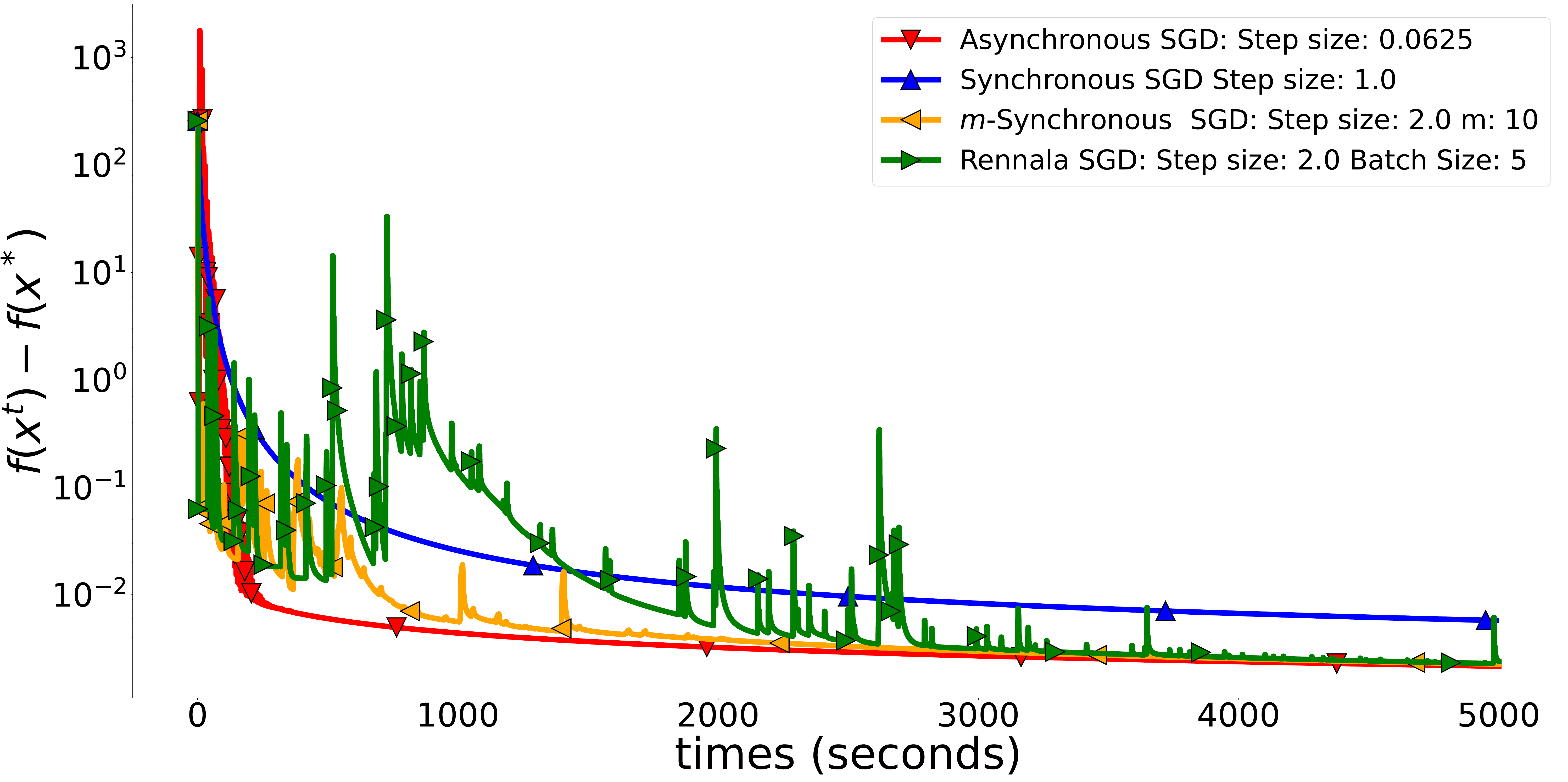}
        \caption{$n=100$}
        \label{fig:img1_3}
    \end{subfigure}
    \hfill
    \begin{subfigure}[b]{0.3\textwidth}
        \centering
        \includegraphics[width=\linewidth]{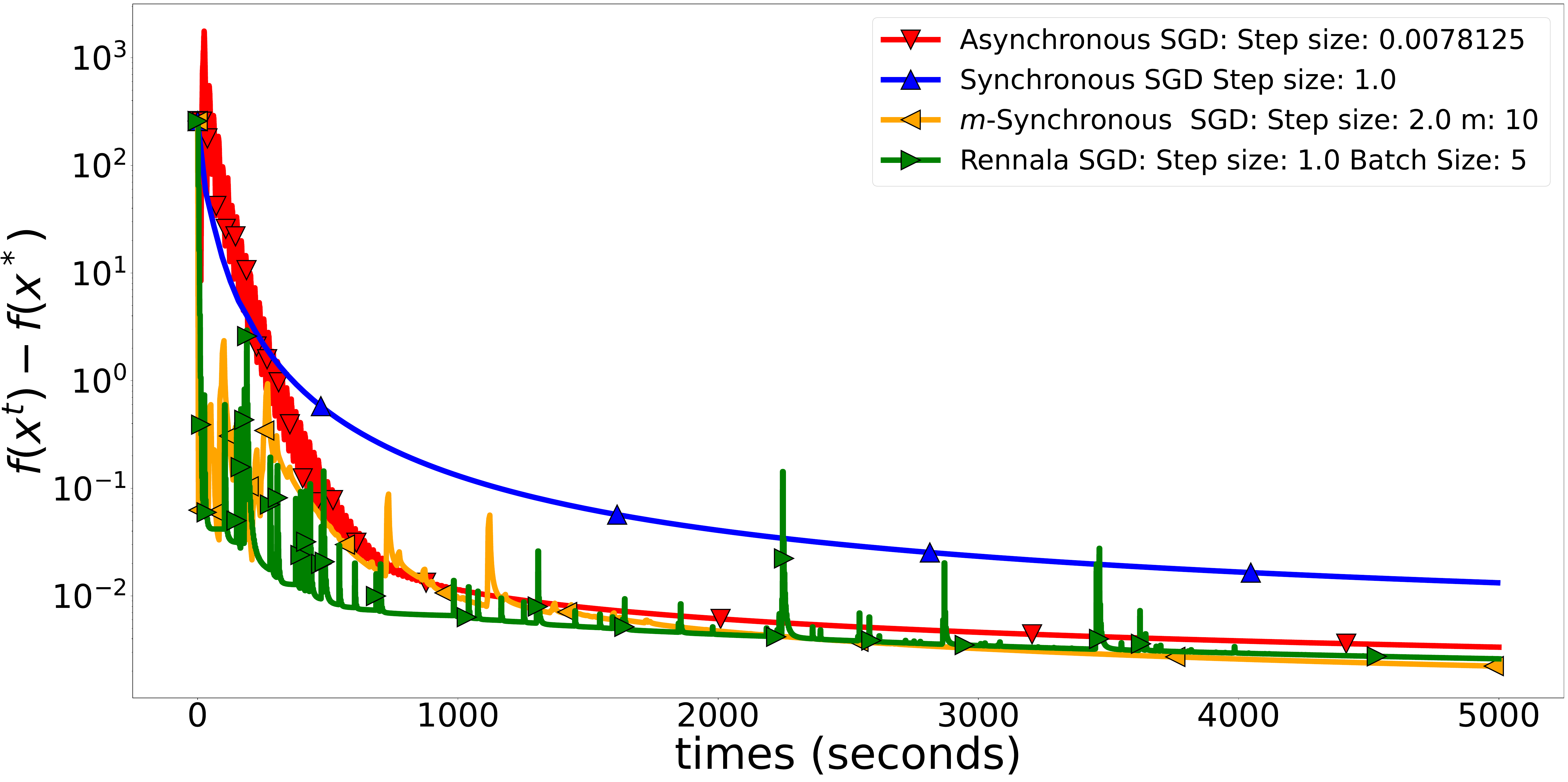}
        \caption{$n=1000$}
        \label{fig:img2_3}
    \end{subfigure}
    \hfill
    \begin{subfigure}[b]{0.3\textwidth}
        \centering
        \includegraphics[width=\linewidth]{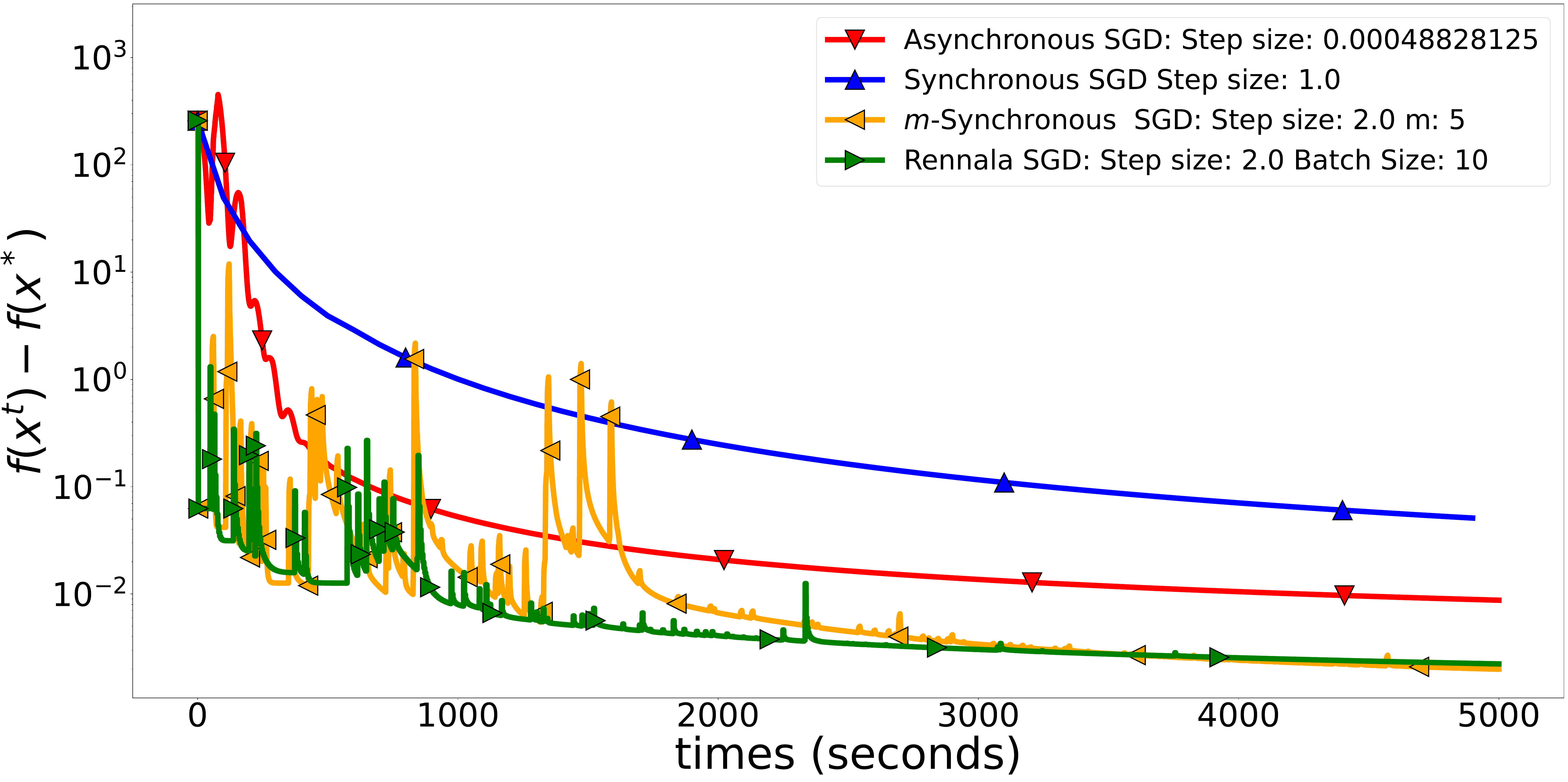}
        \caption{$n=10000$}
        \label{fig:img3_3}
    \end{subfigure}
    \caption{Comparison of the methods on the quadratic optimization problem with different number of workers ($p=0.01, \tau_i = \sqrt{i}$).}
    \label{fig:scale_n}
\end{figure}

\subsection{Quadratic optimization with random computational times}
In this section, we consider an experiment setup similar to Section~\ref{sec:heter_times}, with the only difference that we examine the performance of optimization methods when computational times are modeled as random variables drawn from different probability distributions.

\begin{enumerate}
\item First, we assume that random computational times are drawn from a normal distribution with parameter $\mu_i$ and variance $\sigma^2$, truncated to the non-negative half-line $[0, \infty),$ which we denote as $\cN_{\geq 0}(\mu_i, \sigma^2).$ ($\Exp{\bar{\tau}_i} \neq \mu_i,$ in general). 

\item Next, we consider $\bar{\tau}_i$ drawn from a Gamma distribution with shape parameter $k = \tau_i^2 / \sigma^2$ and 
scale parameter $\theta = \sigma^2 / \tau_i$. Under this parametrization, the mean and variance are $\Exp{\bar{\tau}_i} = k \theta = \tau_i$ and $\operatorname{Var}(\bar{\tau}_i) = k \theta^2 = \sigma^2,$ accordingly. 

\item Finally, we take $\bar{\tau}_i$ drawn from a uniform distribution $\textnormal{Unif}[\tau_i - \sigma, \tau_i + \sigma].$ 
\end{enumerate}

We present the results in Figures~\ref{fig:sqrt_random}, \ref{fig:linear_random}, and \ref{fig:quad_unif_noise}. In the first row of Figure~\ref{fig:sqrt_random}, as expected, we get the same result as in Figure~\ref{fig:plot1}. In all noise regimes, we observe that $m$-Synchronous SGD is robust and that its performance is comparable to that of the asynchronous methods. We arrive at the same conclusion in Figure~\ref{fig:linear_random} when we change $\tau_i$ to $\tau_i = i$. Figure~\ref{fig:quad_unif_noise} also concur. However, in Figure~\ref{fig:quad_unif_noise}, even Synchronous SGD with all participating workers has a good performance when the means are the same, i.e., $\tau_i = 1$, which supports our theoretical results.

\begin{figure}[htbp]
    \centering
    \begin{subfigure}[b]{0.3\textwidth}
        \centering
        \includegraphics[width=\linewidth]{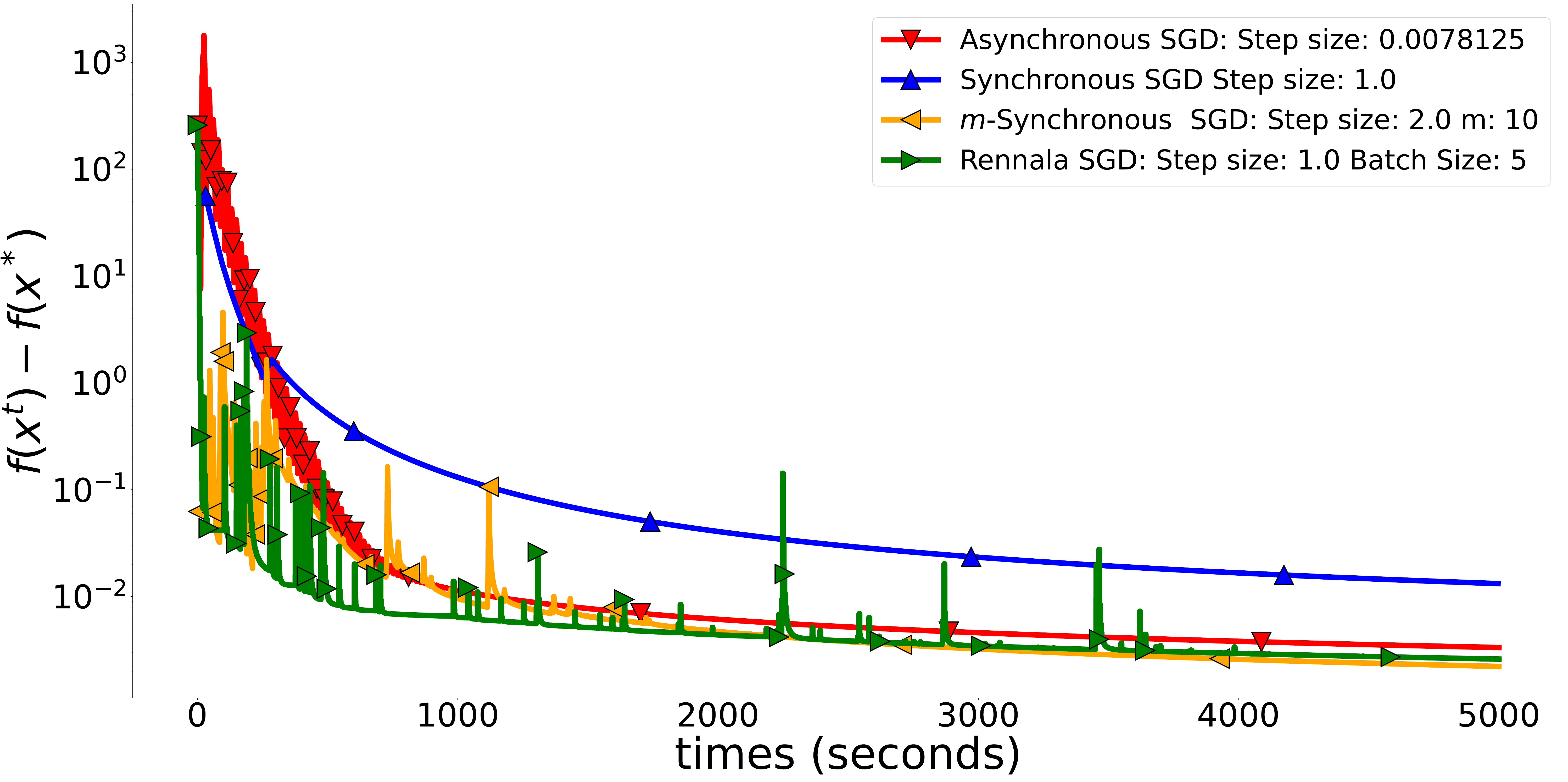}
        \caption{$\bar{\tau}_i = \tau_i$}
    \end{subfigure}
    \begin{subfigure}[b]{0.3\textwidth}
        \centering
        \includegraphics[width=\linewidth]{final_results/final__nn_1000_dim_1000_noise_0_01_delay_sqrt_time5000.pdf}
        \caption{$\bar{\tau}_i = \tau_i$}
    \end{subfigure}
    \begin{subfigure}[b]{0.3\textwidth}
        \centering
        \includegraphics[width=\linewidth]{final_results/final__nn_1000_dim_1000_noise_0_01_delay_sqrt_time5000.pdf}
        \caption{$\bar{\tau}_i = \tau_i$}
    \end{subfigure}

    \medskip

    \begin{subfigure}[b]{0.3\textwidth}
        \centering
        \includegraphics[width=\linewidth]{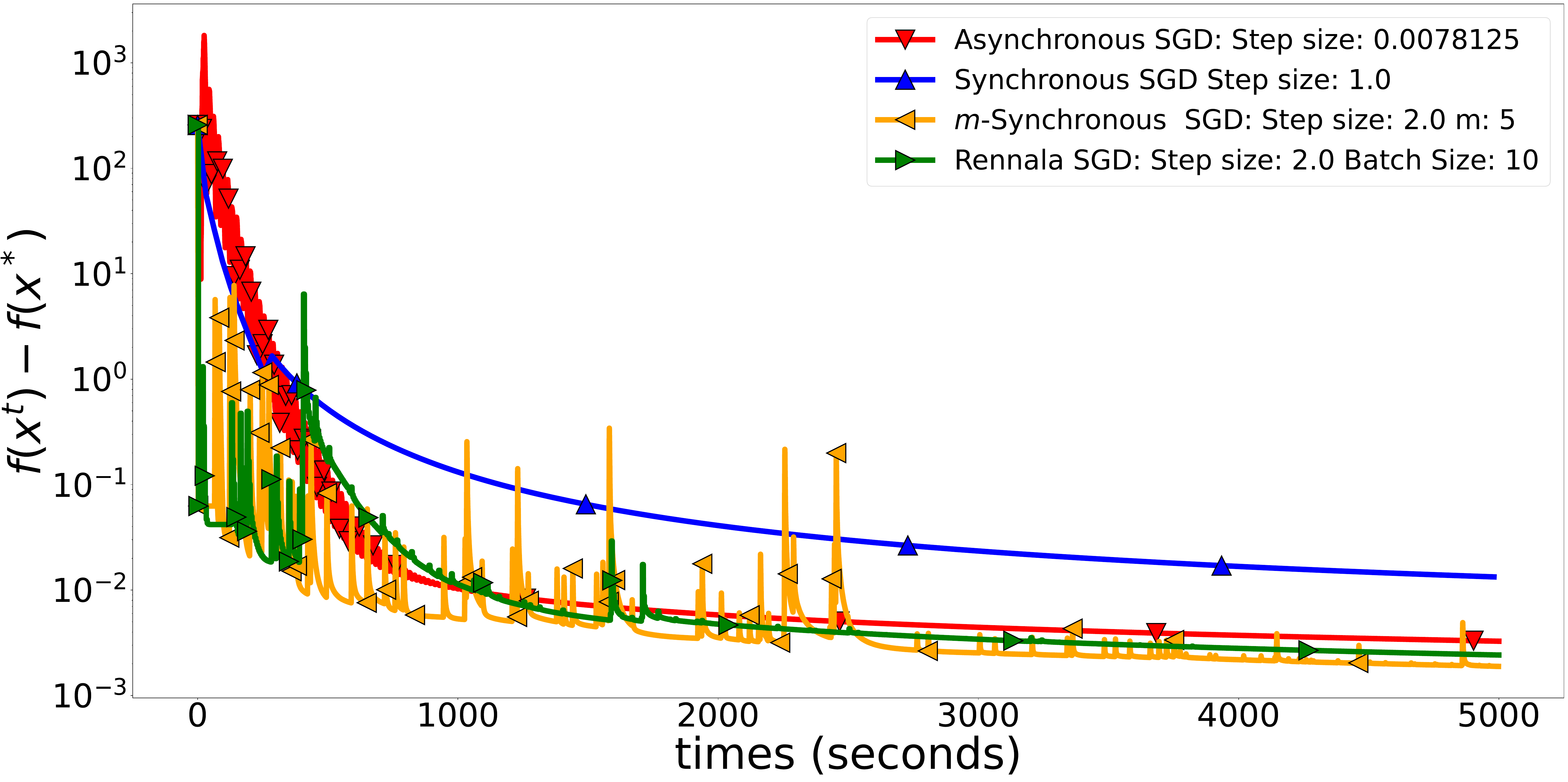}
        \caption{$\bar{\tau}_i \sim \cN_{\geq 0}(\sqrt{i}, \sigma^2), \sigma = 0.1$}
    \end{subfigure}
    \begin{subfigure}[b]{0.3\textwidth}
        \centering
        \includegraphics[width=\linewidth]{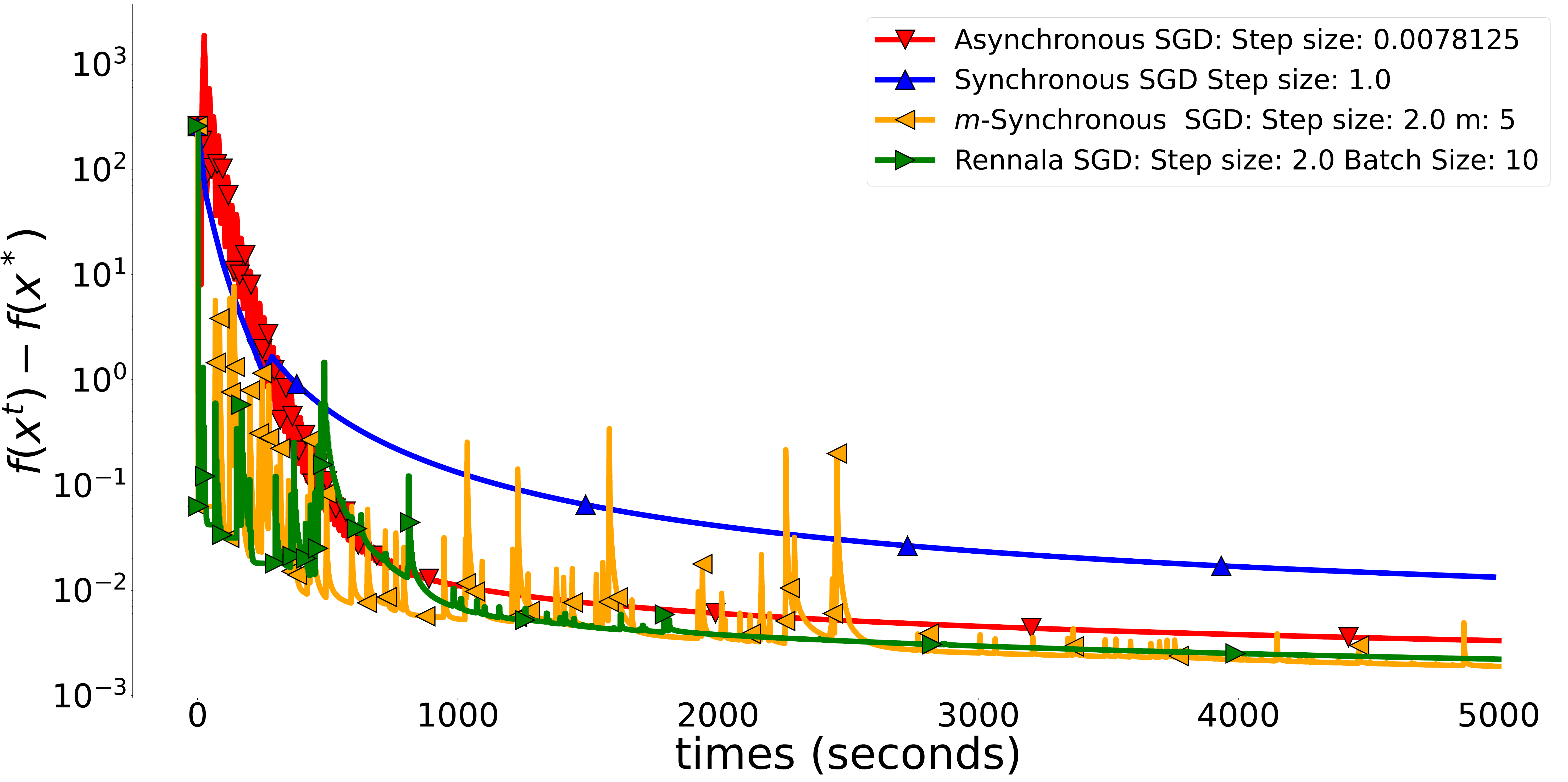}
        \caption{$\bar{\tau}_i \sim \text{Gamma}(\frac{\tau_i^2}{0.1^2}, \frac{0.1^2}{\tau_i})$}
    \end{subfigure}
    \begin{subfigure}[b]{0.3\textwidth}
        \centering
        \includegraphics[width=\linewidth]{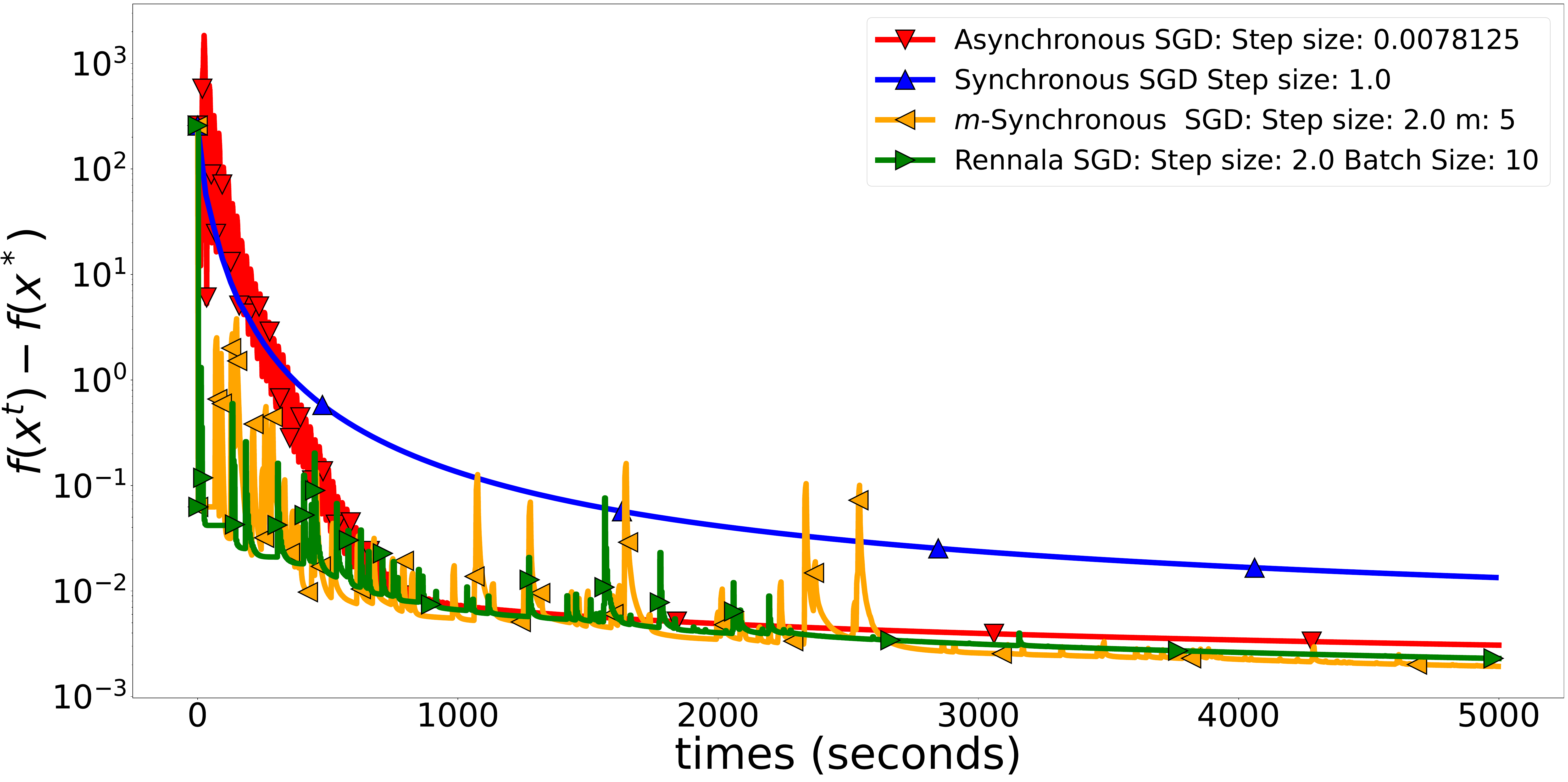}
        \caption{$\bar{\tau}_i \sim \tau_i + \textnormal{Unif}([-0.5, 0.5])$}
    \end{subfigure}

    \medskip
    
    \begin{subfigure}[b]{0.3\textwidth}
        \centering
        \includegraphics[width=\linewidth]{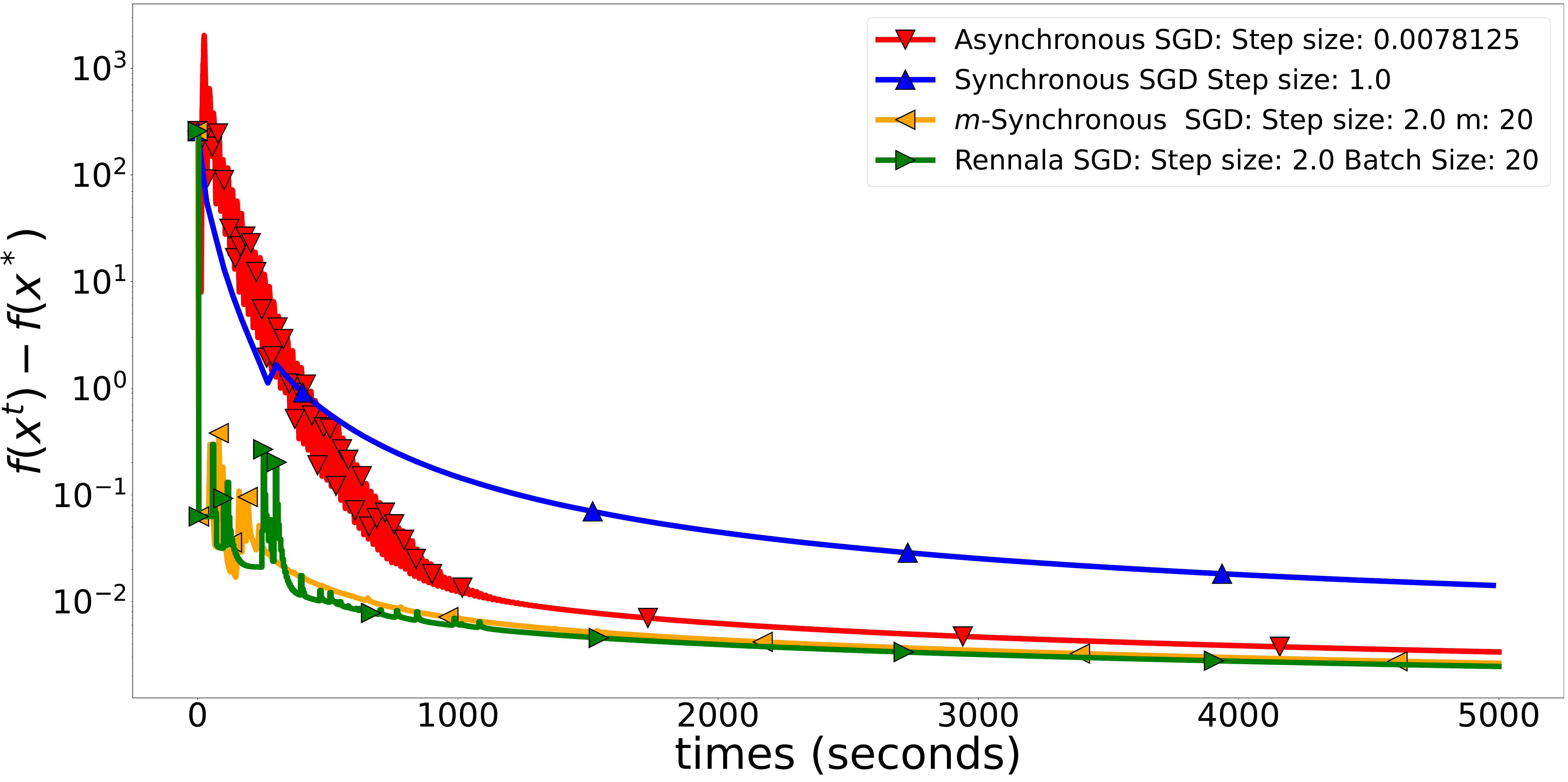}
        \caption{$\bar{\tau}_i \sim \cN_{\geq 0}(\sqrt{i}, \sigma^2), \sigma = 1$}
    \end{subfigure}
    \begin{subfigure}[b]{0.3\textwidth}
        \centering
        \includegraphics[width=\linewidth]{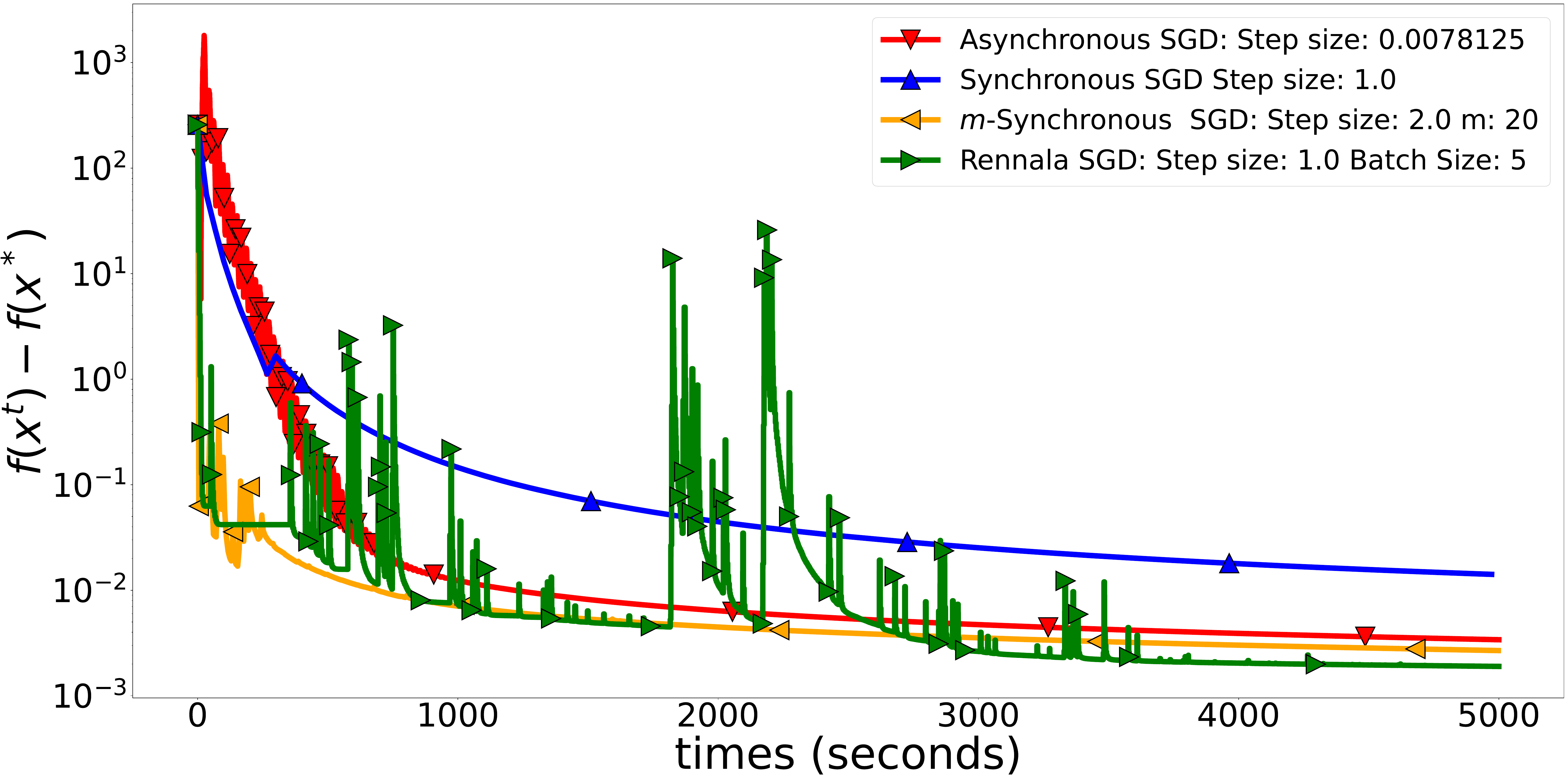}
        \caption{$\bar{\tau}_i \sim \text{Gamma}(\frac{\tau_i^2}{1^2}, \frac{1^2}{\tau_i})$}
    \end{subfigure}
    \begin{subfigure}[b]{0.3\textwidth}
        \centering
        \includegraphics[width=\linewidth]{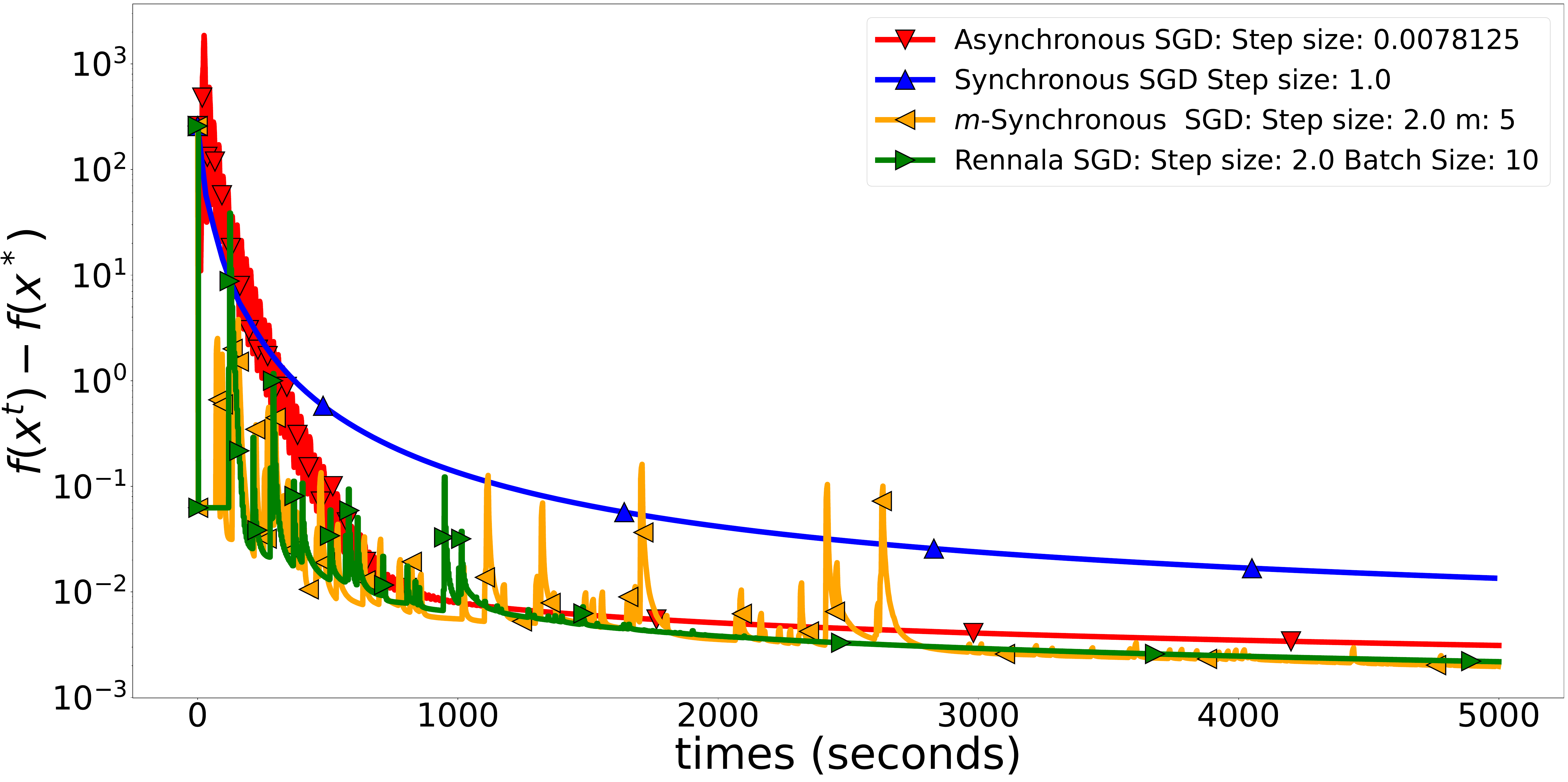}
        \caption{$\bar{\tau}_i \sim \tau_i + \textnormal{Unif}([-0.7, 0.7])$}
    \end{subfigure}

    \caption{Comparison of the methods on a quadratic optimization problem with different types of random computation times. We take $p = 0.01$, $n = 1000$, and $\tau_i = \sqrt{i}$. Three noise regimes are considered: (i) no noise in the first row; (ii) weak noise in the second row; and (iii) stronger noise in the third row. In the first column, we sample computation times from the truncated normal distribution. In the second column, from the gamma distribution. In the third column, from the uniform distribution.}
    \label{fig:sqrt_random}
\end{figure}

\begin{figure}[h]
    \centering
    \begin{subfigure}[b]{0.3\textwidth}
        \centering
        \includegraphics[width=\linewidth]{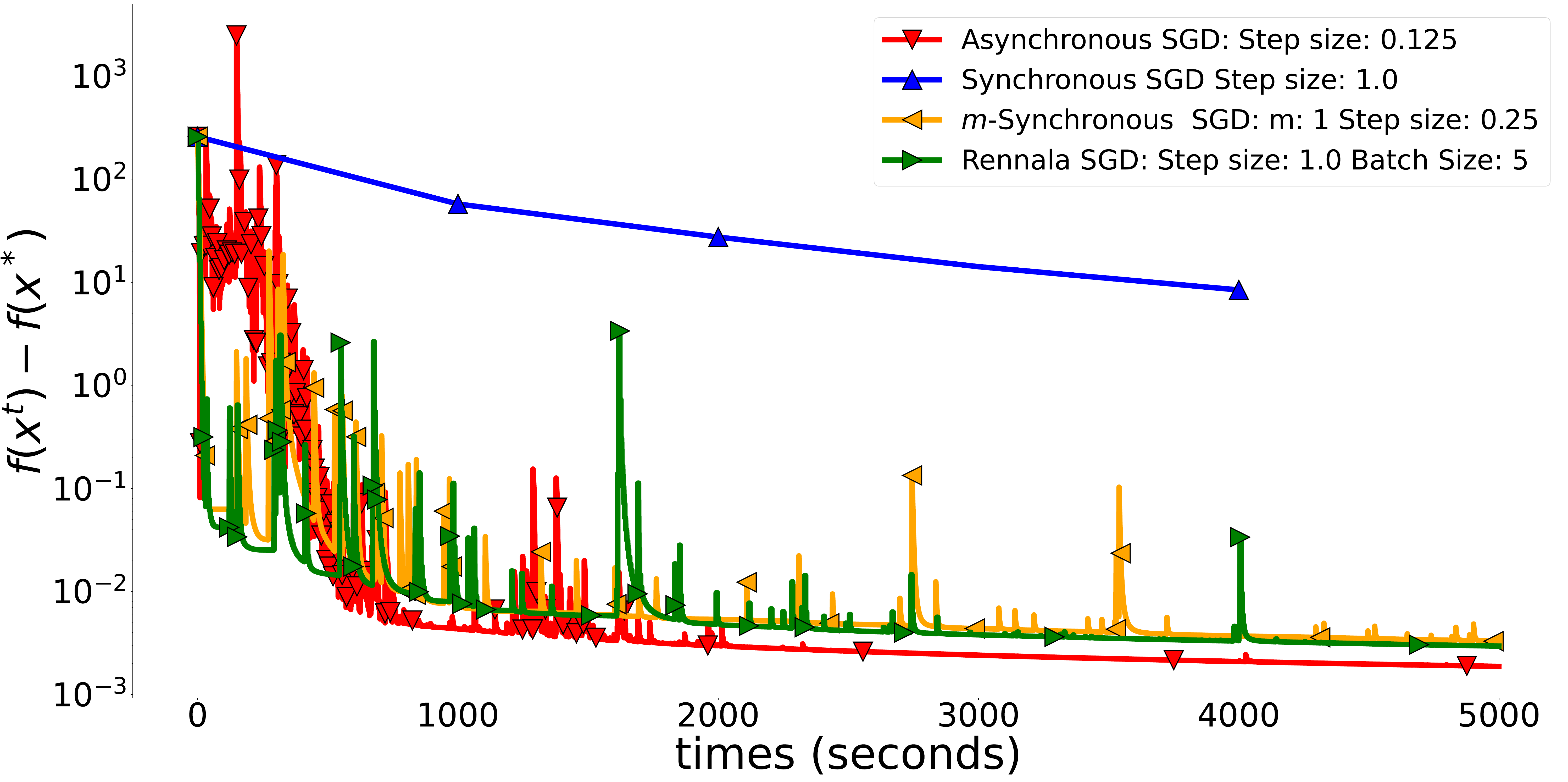}
        \caption{$\bar{\tau}_i = \tau_i$}
    \end{subfigure}
    \begin{subfigure}[b]{0.3\textwidth}
        \centering
        \includegraphics[width=\linewidth]{final_results/remake_1_noise_001_delays_linear_num_nodes_1000_distr_none_sigma_00_time5000.pdf}
        \caption{$\bar{\tau}_i = \tau_i$}
    \end{subfigure}
    \begin{subfigure}[b]{0.3\textwidth}
        \centering
        \includegraphics[width=\linewidth]{final_results/remake_1_noise_001_delays_linear_num_nodes_1000_distr_none_sigma_00_time5000.pdf}
        \caption{$\bar{\tau}_i = \tau_i$}
    \end{subfigure}
    
    \medskip

    \begin{subfigure}[b]{0.3\textwidth}
        \centering
        \includegraphics[width=\linewidth]{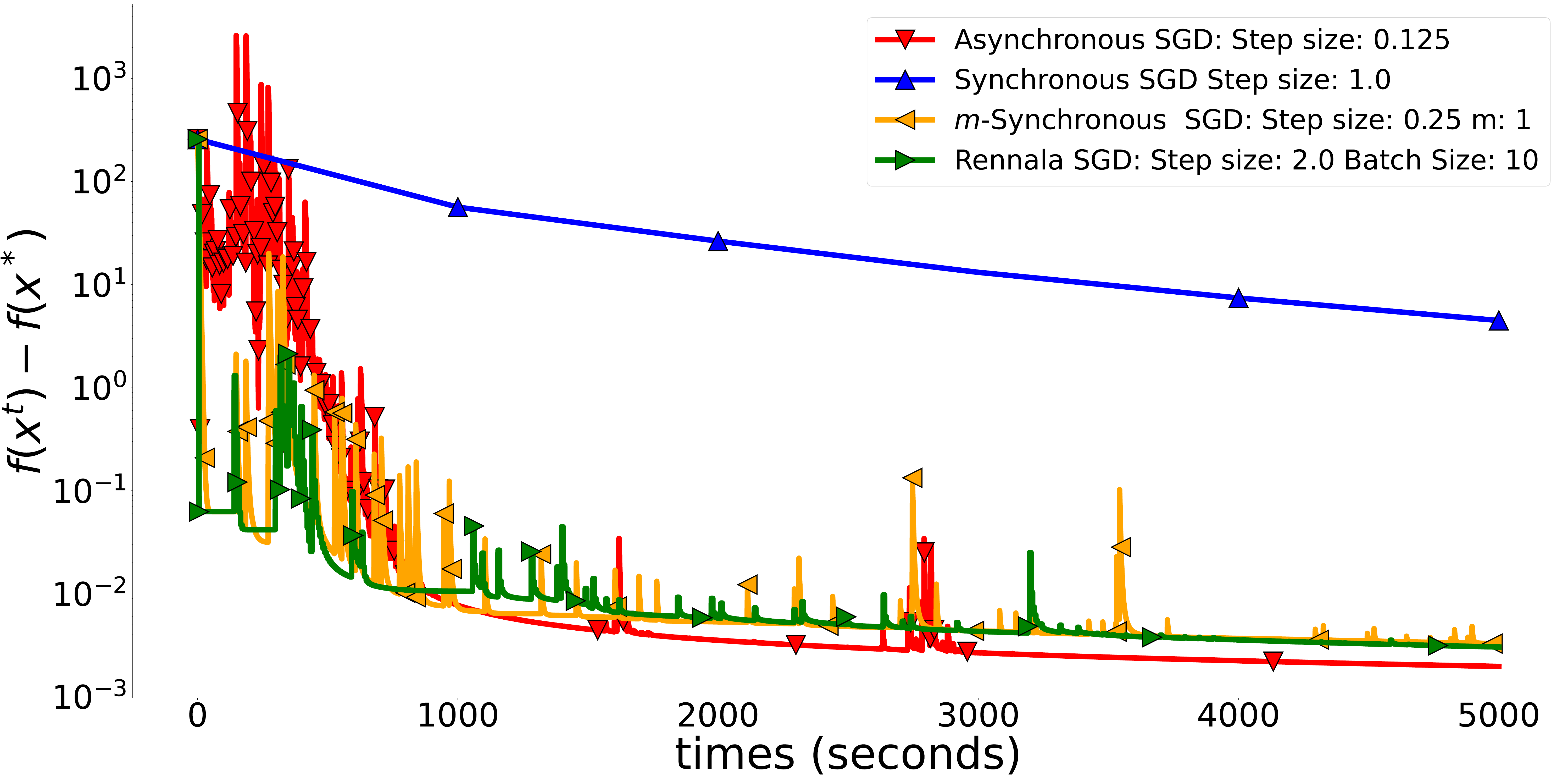}
        \caption{$\bar{\tau}_i \sim \cN_{\geq 0}(i, \sigma^2), \sigma = 0.1$}
    \end{subfigure}
    \begin{subfigure}[b]{0.3\textwidth}
        \centering
        \includegraphics[width=\linewidth]{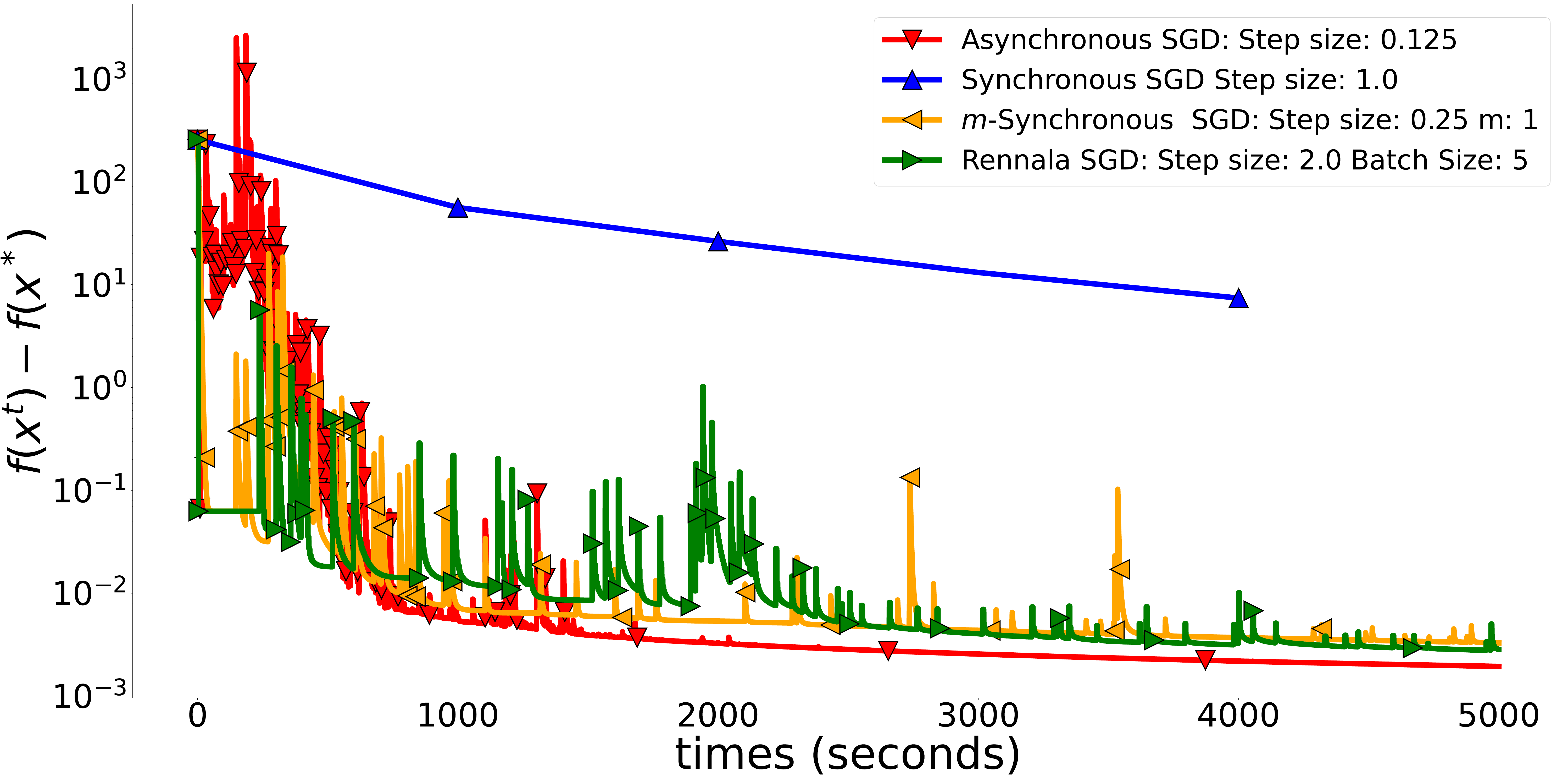}
        \caption{$\bar{\tau}_i \sim \text{Gamma}(\frac{\tau_i^2}{0.1^2}, \frac{0.1^2}{\tau_i})$}
    \end{subfigure}
    \begin{subfigure}[b]{0.3\textwidth}
        \centering
        \includegraphics[width=\linewidth]{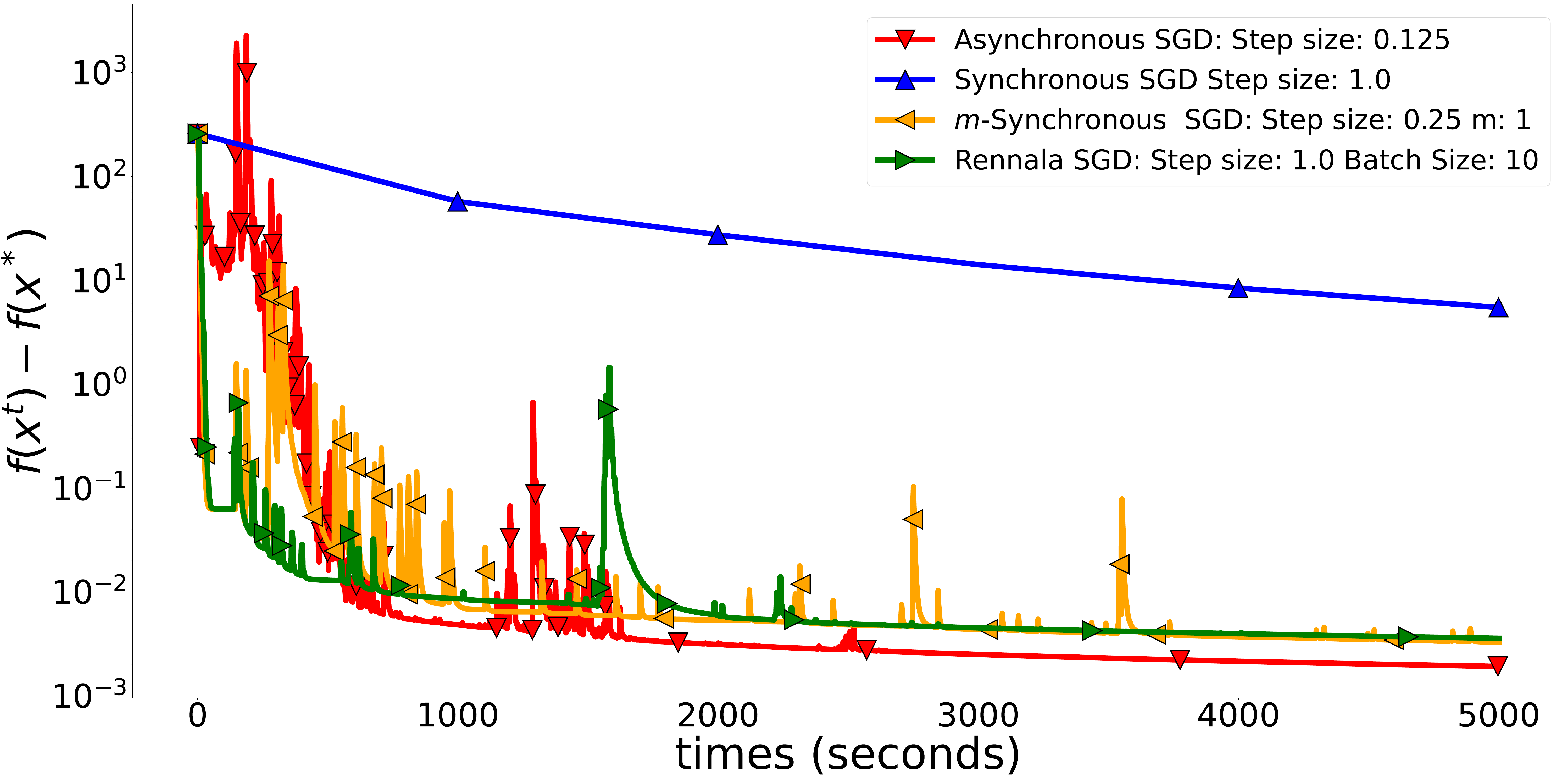}
        \caption{$\bar{\tau}_i \sim \tau_i + \textnormal{Unif}([-0.5, 0.5])$}
    \end{subfigure}

    \medskip

    \begin{subfigure}[b]{0.3\textwidth}
        \centering
        \includegraphics[width=\linewidth]{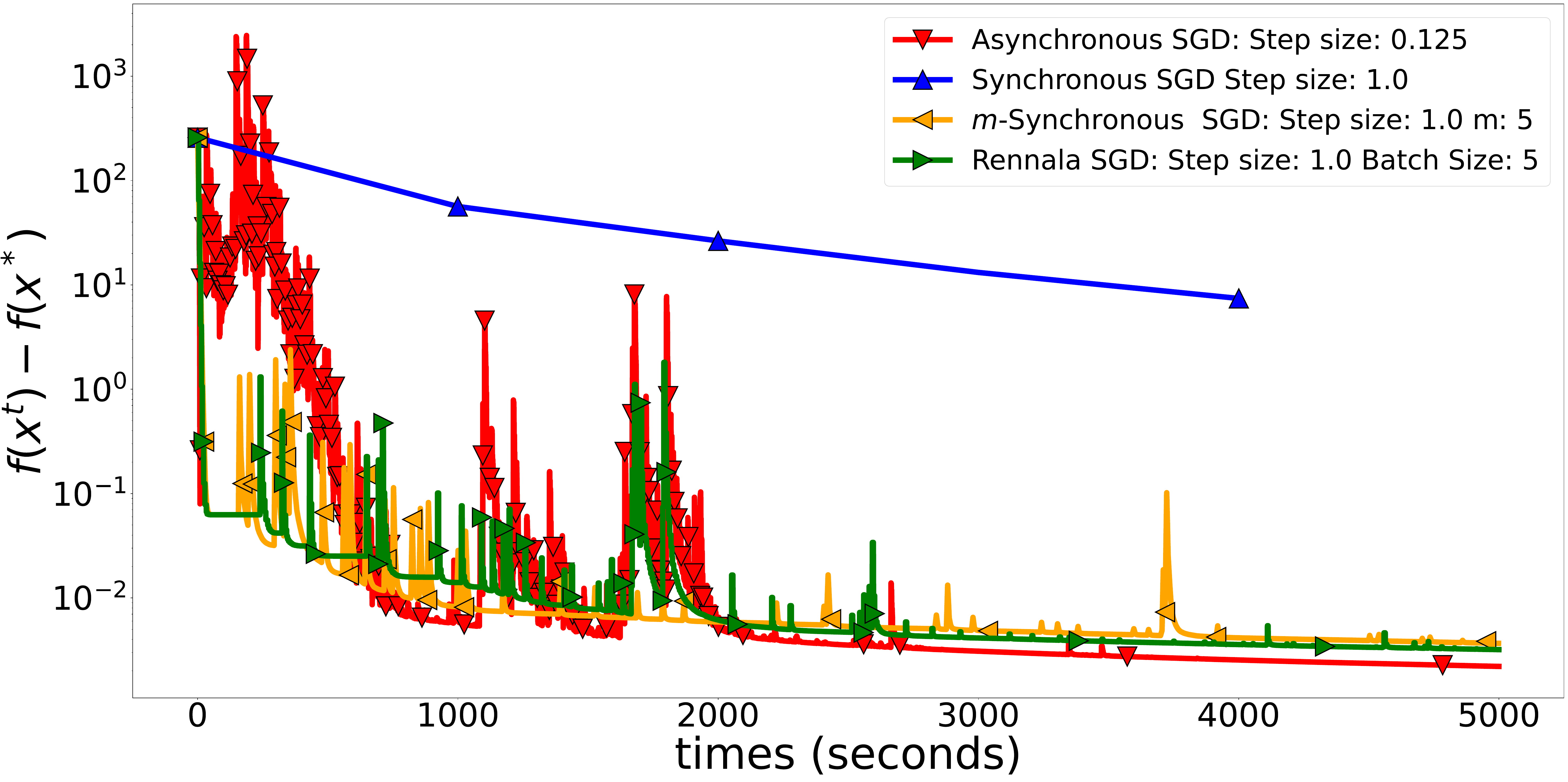}
        \caption{$\bar{\tau}_i \sim \cN_{\geq 0}(i, \sigma^2), \sigma = 1$}
    \end{subfigure}
    \begin{subfigure}[b]{0.3\textwidth}
        \centering
        \includegraphics[width=\linewidth]{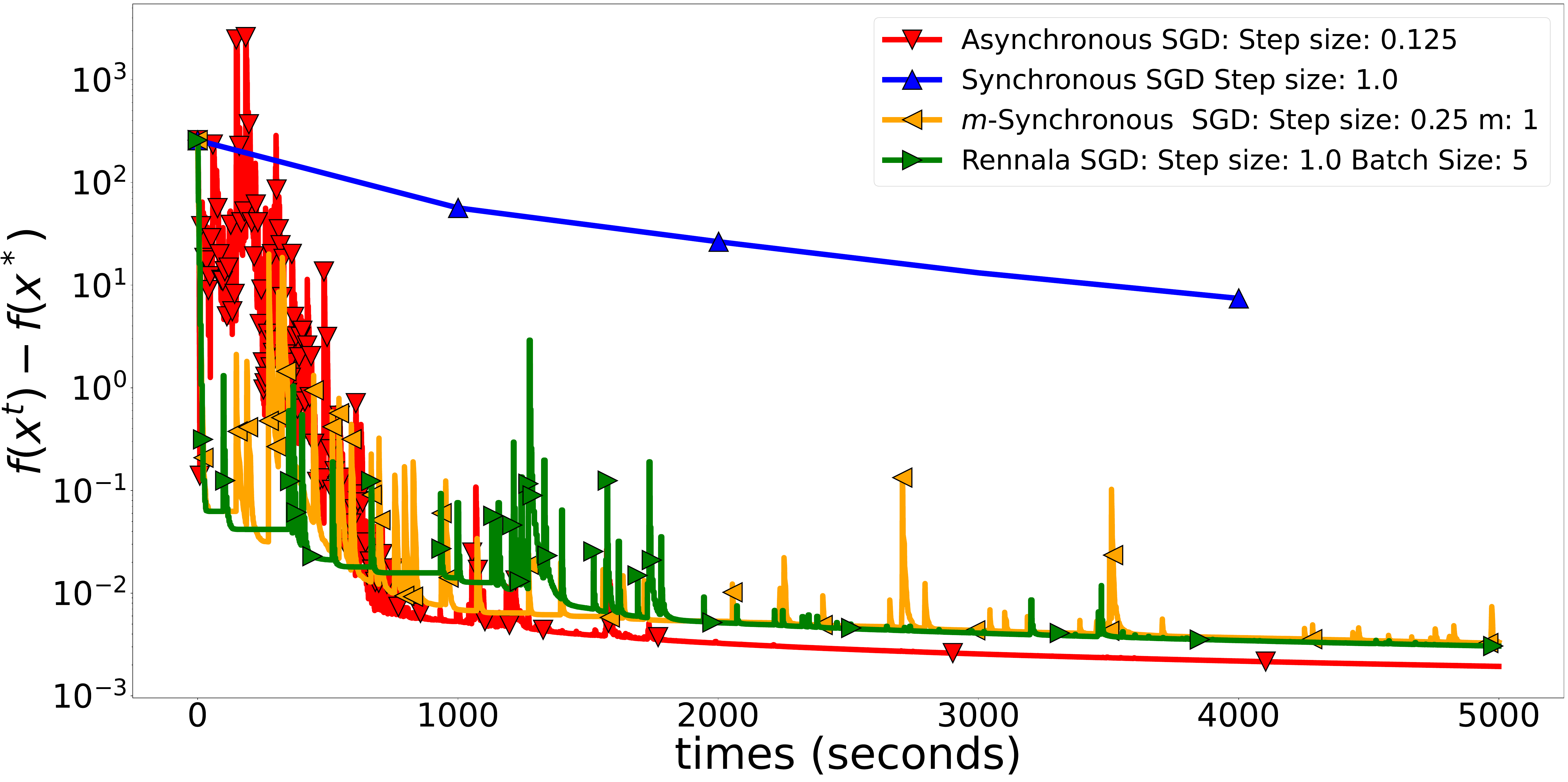}
        \caption{$\bar{\tau}_i \sim \text{Gamma}(\frac{\tau_i^2}{1^2}, \frac{1^2}{\tau_i})$}
        \label{fig:asfqaf}
    \end{subfigure}
    \begin{subfigure}[b]{0.3\textwidth}
        \centering
        \includegraphics[width=\linewidth]{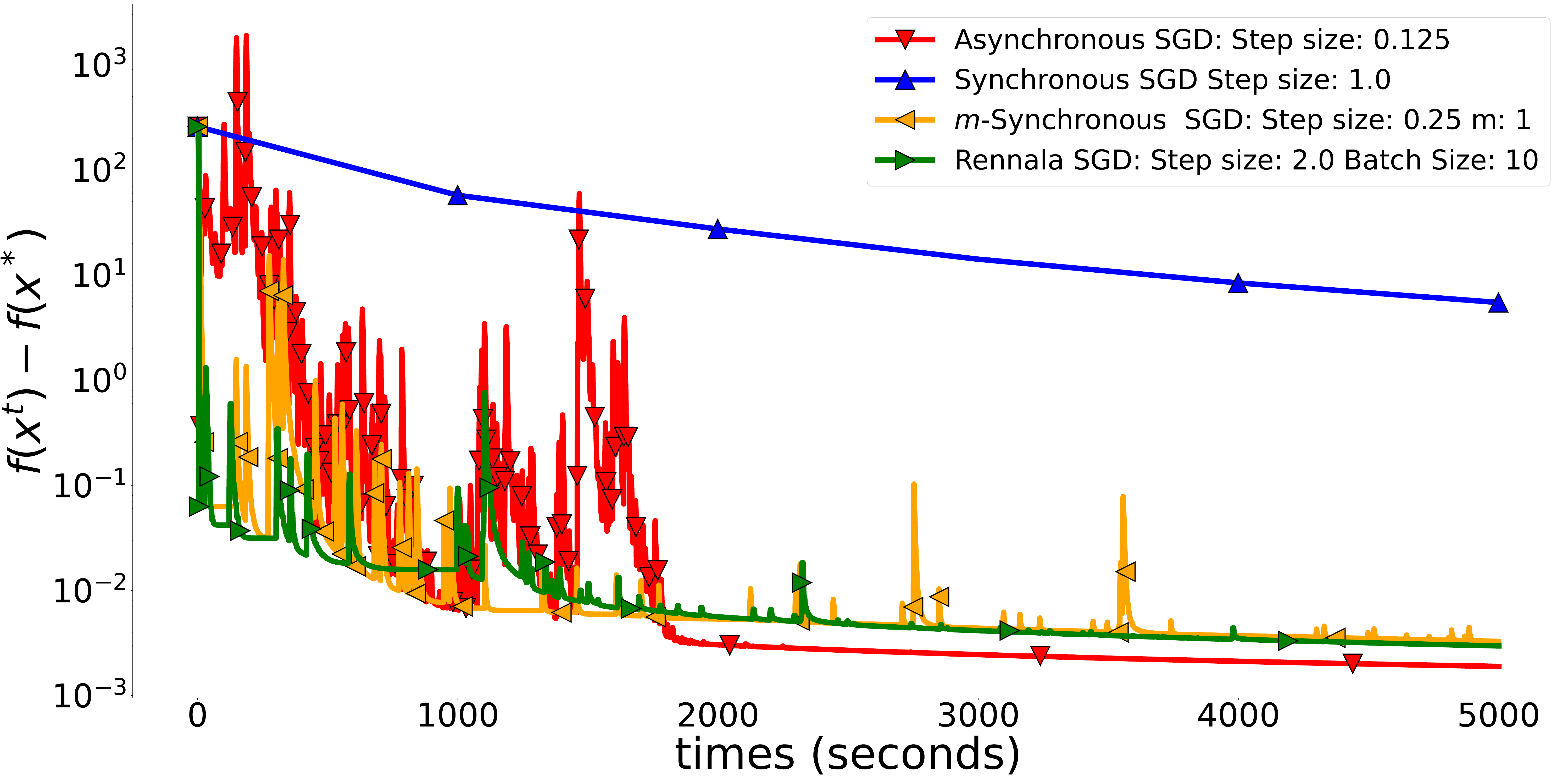}
        \caption{$\bar{\tau}_i \sim \tau_i + \textnormal{Unif}([-0.7, 0.7])$}
    \end{subfigure}

    \caption{Comparison of the methods on a quadratic optimization problem with different types of random computation times. We take $p = 0.01$, $n = 1000$, and $\tau_i = i$. Three noise regimes are considered: (i) no noise in the first row; (ii) weak noise in the second row; and (iii) stronger noise in the third row. In the first column, we sample computation times from the truncated normal distribution. In the second column, from the gamma distribution. In the third column, from the uniform distribution.}
    \label{fig:linear_random}
\end{figure}

\begin{figure}[h]
    \centering
    \begin{subfigure}[b]{0.3\textwidth}
        \centering
        \includegraphics[width=\linewidth]{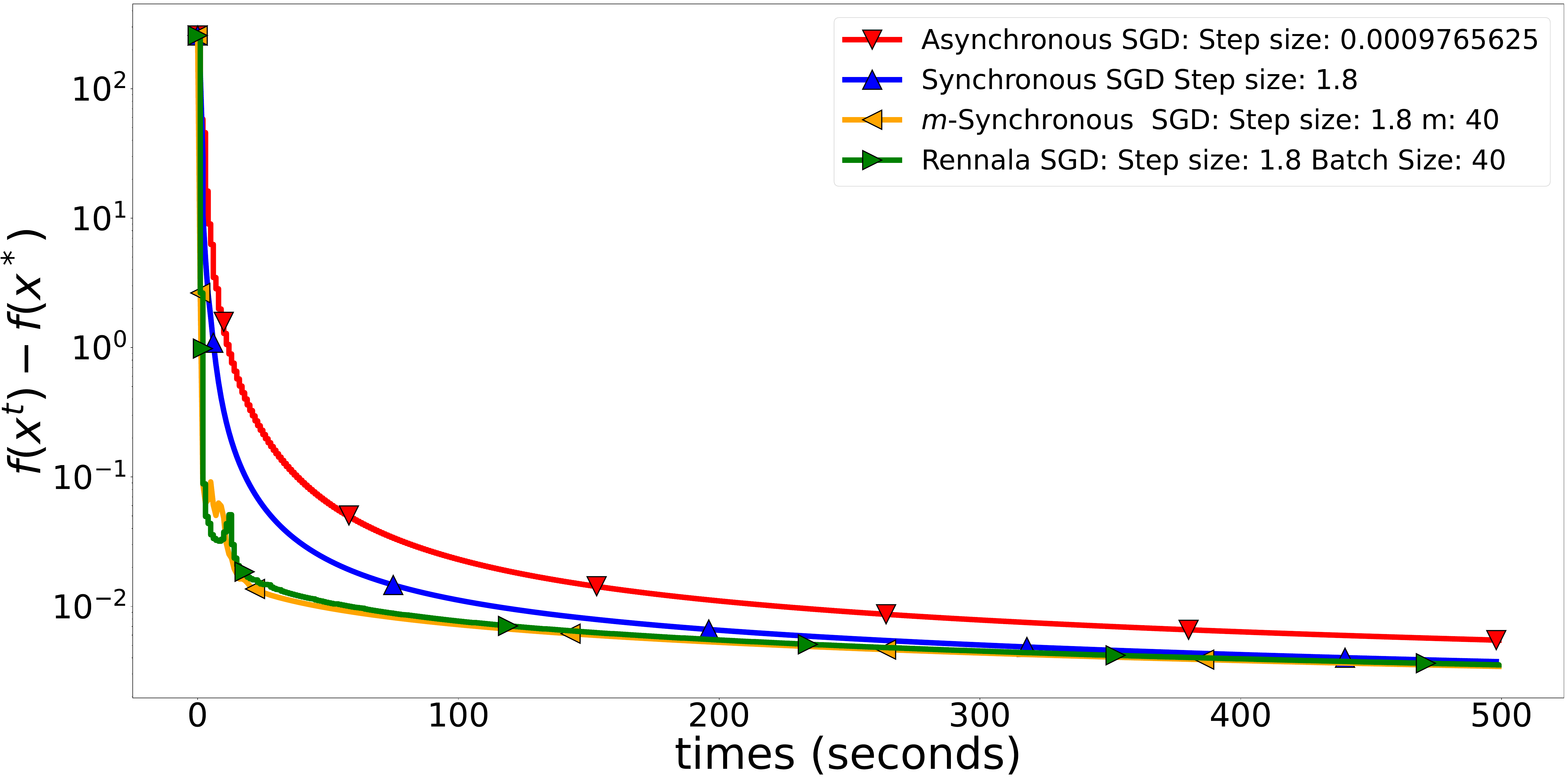}
        \caption{$\bar{\tau}_i = 1$}
    \end{subfigure}
    \hfill
    \begin{subfigure}[b]{0.3\textwidth}
        \centering
        \includegraphics[width=\linewidth]{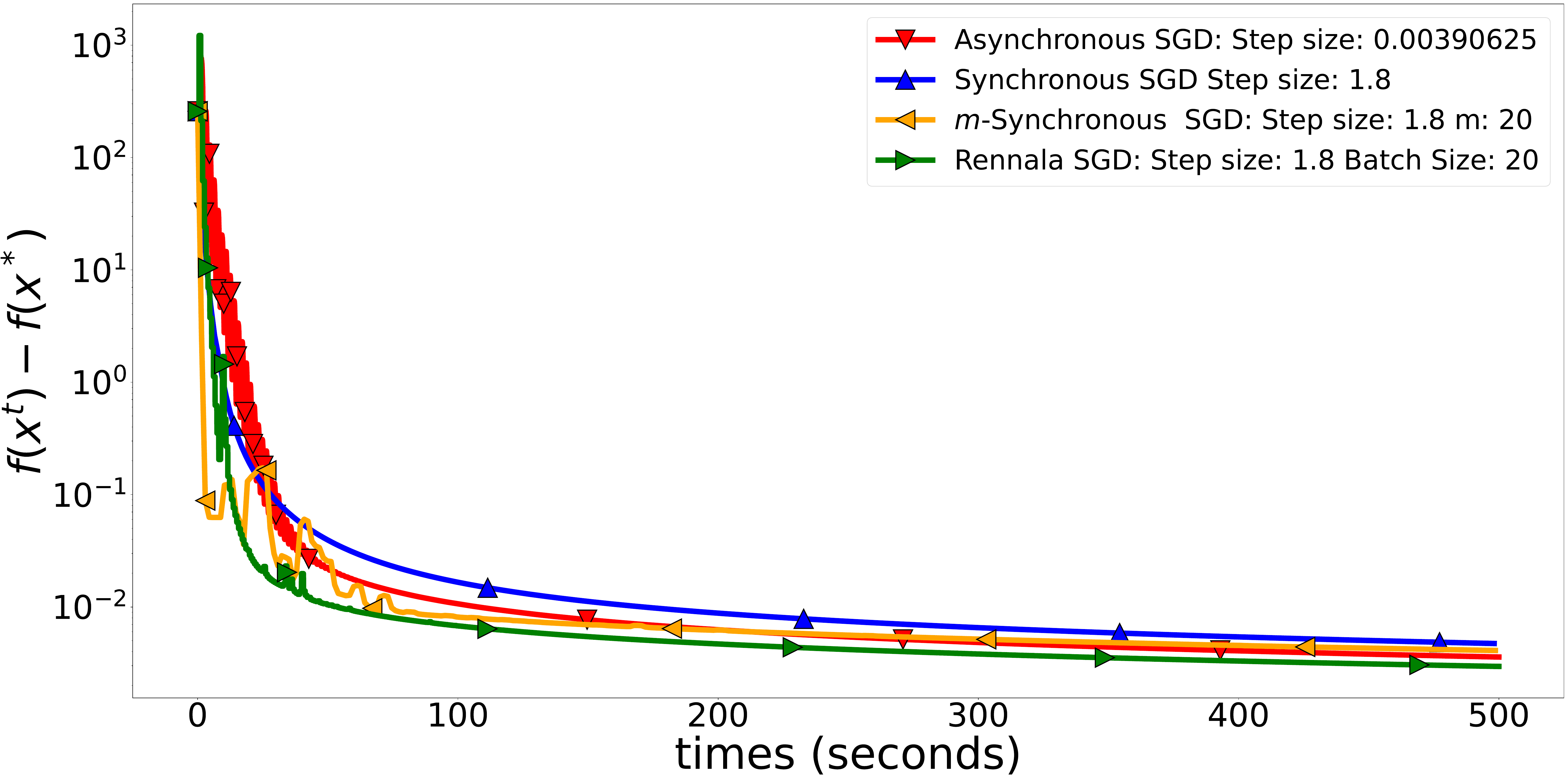}
        \caption{$\bar{\tau}_i \sim 1 + \textnormal{Unif}([-0.5, 0.5])$}
    \end{subfigure}
    \hfill
    \begin{subfigure}[b]{0.3\textwidth}
        \centering
        \includegraphics[width=\linewidth]{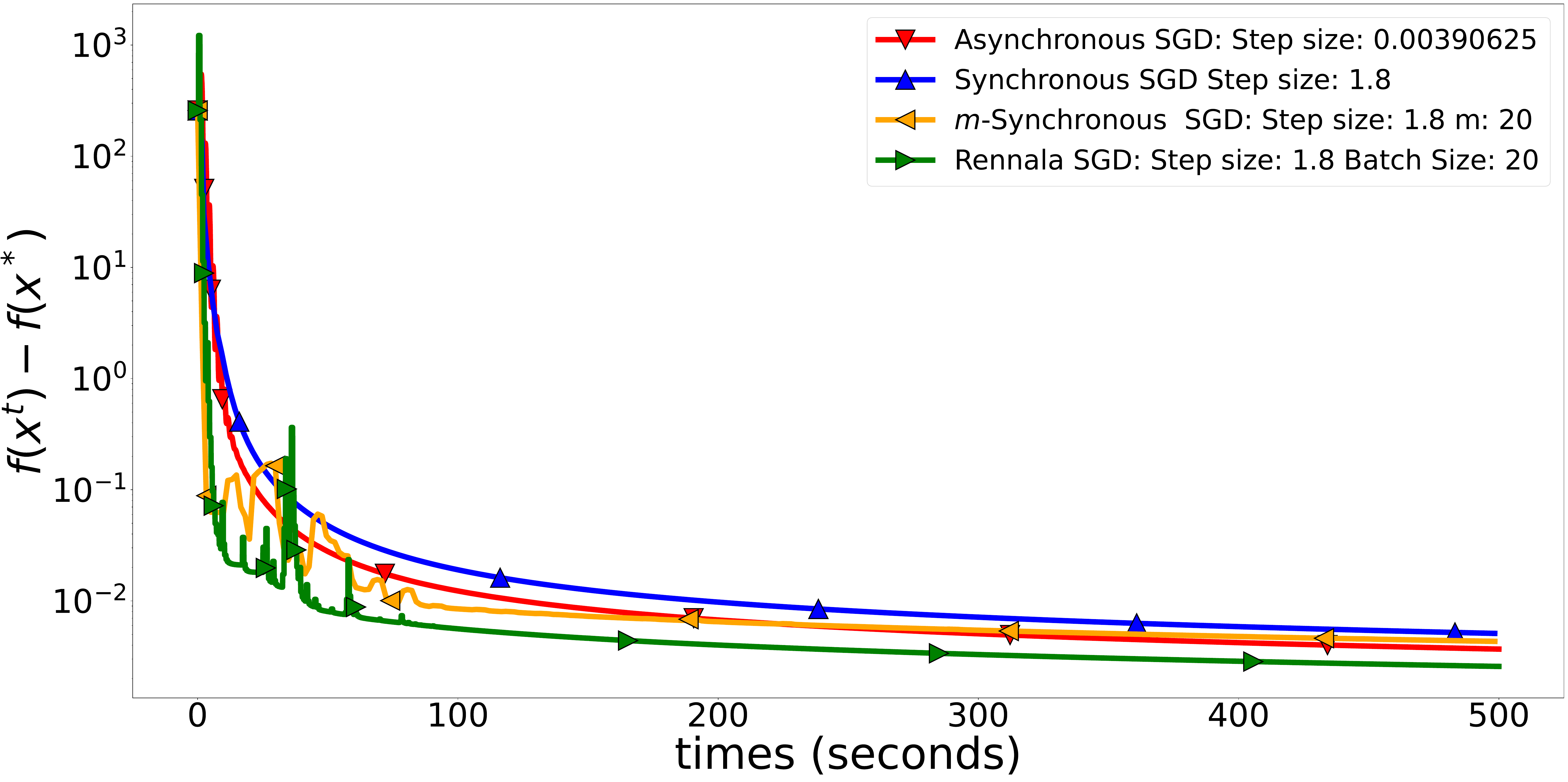}
        \caption{$\bar{\tau}_i \sim 1 + \textnormal{Unif}([-0.7, 0.7])$}
    \end{subfigure}
    \caption{Comparison of the methods on the quadratic optimization problem with the same mean $\tau_i = 1$ and uniform noise ($p=0.01, \tau_i = 1, n = 1000$). In these experiments, for all methods, we choose the best step size from the set $\left\{2^{-16}, \dots, 2^{-2}, 2^{-1}, 1, 1.2, 1.6, 1.8, 2\right\}.$}
    \label{fig:quad_unif_noise}
\end{figure}

\clearpage
\newpage
\subsection{Small-Scale Machine Learning Task}
In this section, we train a two-layer neural network (NN), Linear$(\text{input\_dim},32)\rightarrow\text{ReLU}\rightarrow\text{Linear}(32,\text{num\_classes})$, with the logistic loss on the \emph{CIFAR-10} dataset \citep{krizhevsky2009learning} and \emph{MNIST} \citep{lecun2010mnist} to compare the performance of the methods on tasks where workers compute stochastic gradients via uniformly random data sampling with a sample size of $128.$

We compare two computational scenarios: \\ i) $\bar{\tau}_i \sim \textnormal{Unif}(1 - \sigma, 1 + \sigma),$ the computations times are random variables from the uniform distribution with equal means; \\ ii) $\bar{\tau}_i \sim \cN_{\geq 0}(\mu_i;\sigma^2),$ $\mu_i = \sqrt{i},$ the computations times are random variables from the truncated normal distribution.

The results for CIFAR10 are presented in Figure~\ref{fig:cifar_comparison_uniform}, we consider the first computational scenario $\bar{\tau}_i \sim \textnormal{Unif}(1 - \sigma, 1 + \sigma)$ and vary $\sigma$ to check the methods' robustness to these variations. When $\sigma = 0.5$ and $\sigma = 0.7,$ all methods are comparable, with Asynchronous SGD being faster by at most $\approx 2\times$ times, which is expected in practical tasks. One important observation is that even Synchronous SGD is as fast as the optimal Rennala SGD method. Notice that a $2\times$ factor can be ``caught up'' through more efficient implementations and communication (see Section~\ref{sec:why}, Section~\ref{sec:nanogpt_another}, and \citep{sergeev2018horovod}).

In Figure~\ref{fig:cifar_comparison}, we consider the second computational scenario $\bar{\tau}_i \sim \cN_{\geq 0}(\sqrt{i};\sigma^2)$ and also vary $\sigma.$ The conclusions are the same as in the previous computation scenario, with the only main difference that $m$-Synchronous SGD is more robust than Synchronous SGD when $\tau_i = \sqrt{i}$ are heterogeneous. In Figures~\ref{fig:mnist_comparison_uniform} and \ref{fig:mnist_comparison}, we repeat the same experiments on MNIST, and observe that the behavior and conclusions are consistent across the datasets.

\begin{figure}[h]
\centering

\begin{subfigure}[b]{0.4\textwidth}
    \includegraphics[width=\linewidth]{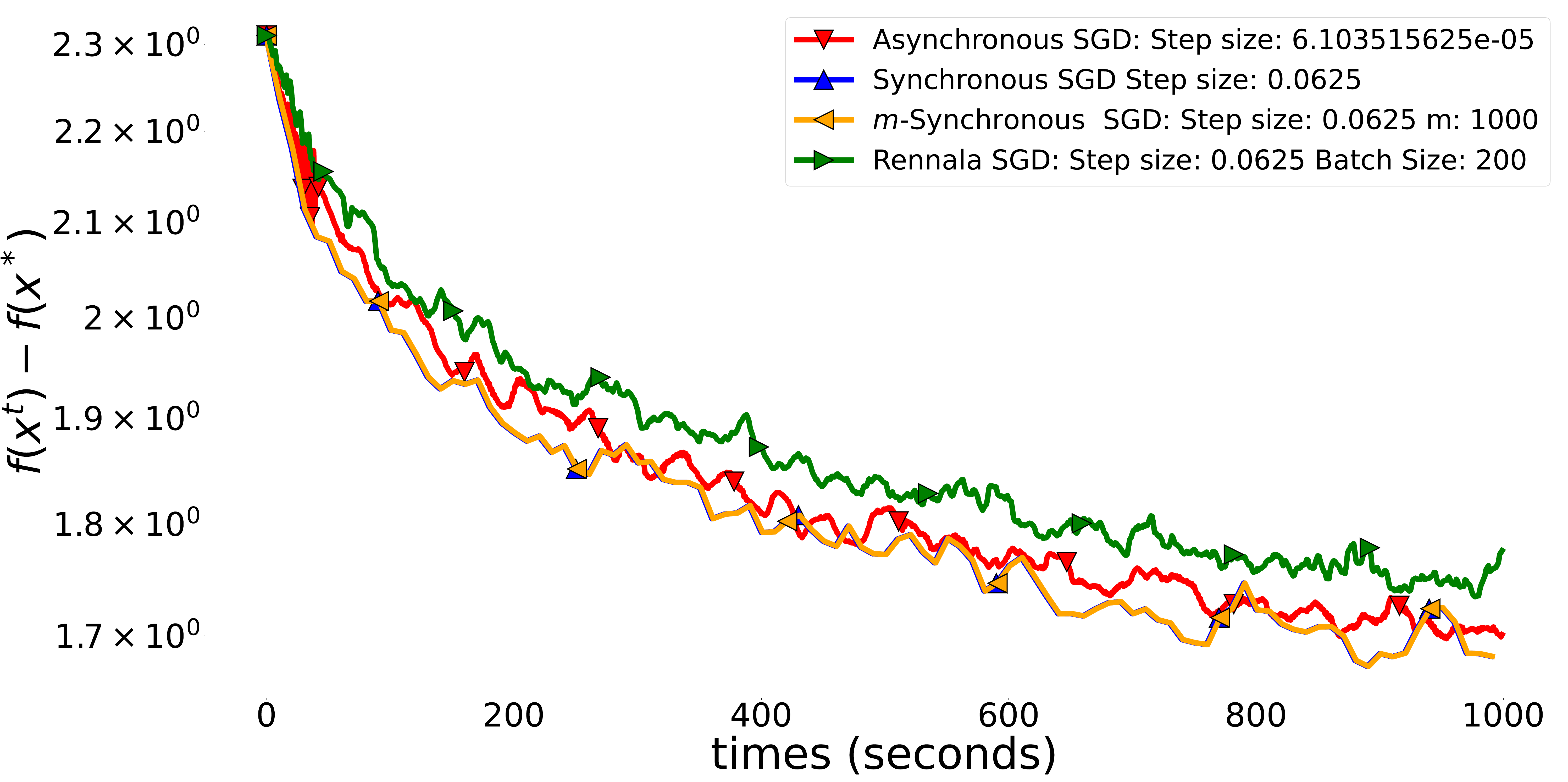}
    \caption*{$\bar{\tau}_i = 1$}
\end{subfigure}
\begin{subfigure}[b]{0.4\textwidth}
    \includegraphics[width=\linewidth]{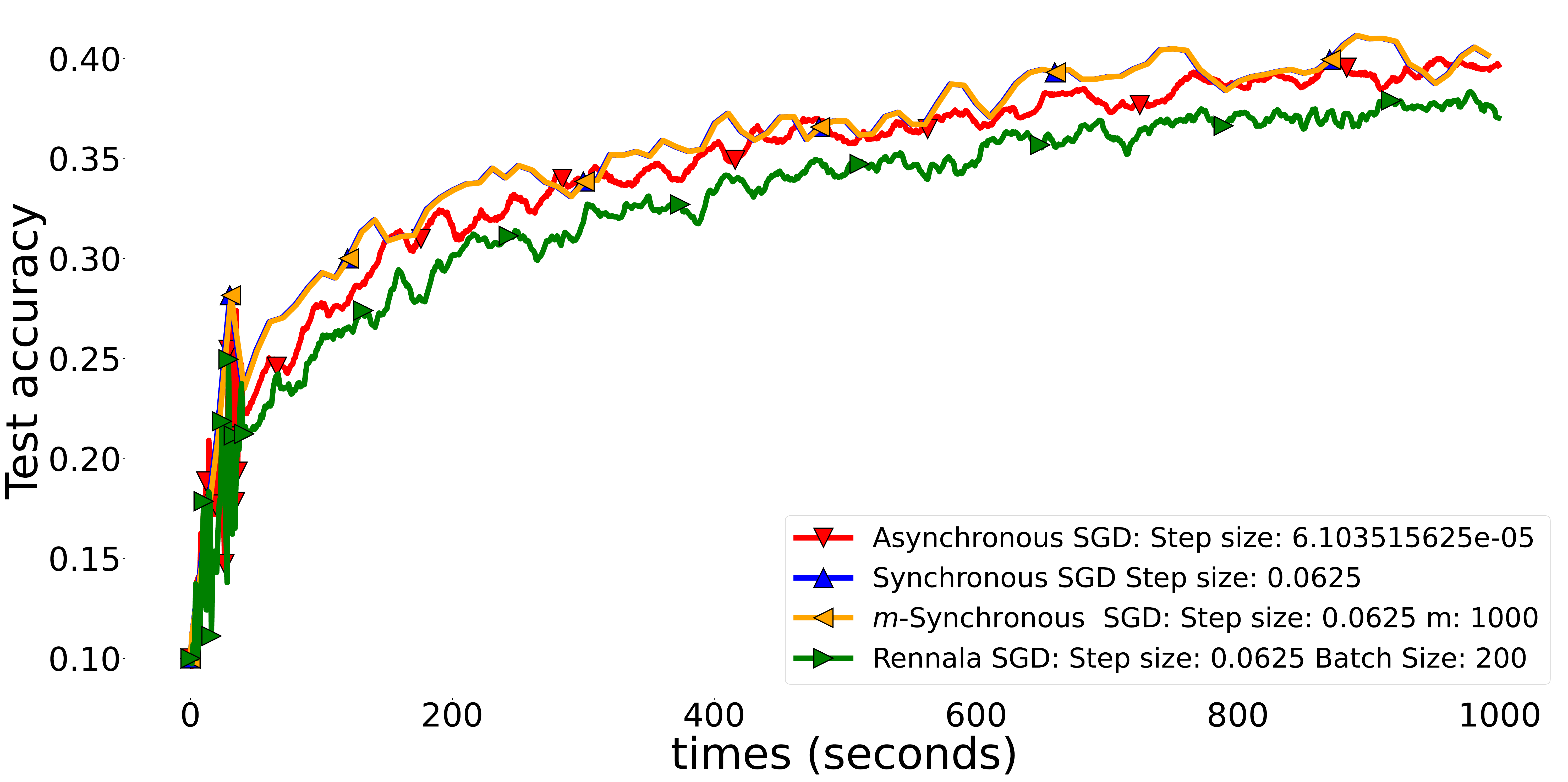}
    \caption*{$\bar{\tau}_i = 1$}
\end{subfigure}

\medskip

\begin{subfigure}[b]{0.4\textwidth}
    \includegraphics[width=\linewidth]{final_resultss/cifar_2_noise_00_delays_equal_num_nodes_1000_distr_uniform_sigma_05_time1000}
    \caption*{$\bar{\tau}_i \sim 1 + \textnormal{Unif}([-0.5, 0.5])$}
\end{subfigure}
\begin{subfigure}[b]{0.4\textwidth}
    \includegraphics[width=\linewidth]{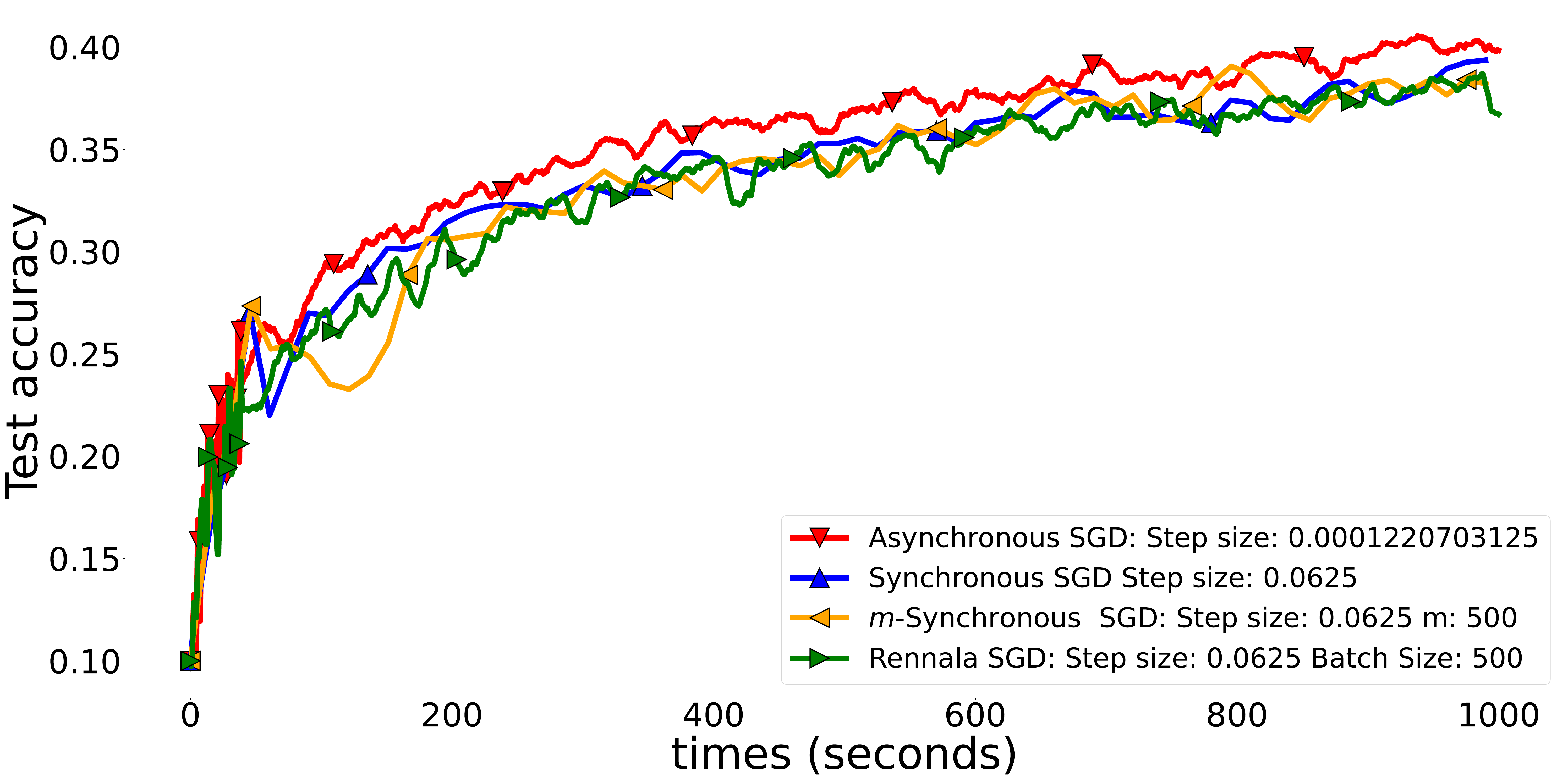}
    \caption*{$\bar{\tau}_i \sim 1 + \textnormal{Unif}([-0.5, 0.5])$}
\end{subfigure}

\medskip

\begin{subfigure}[b]{0.4\textwidth}
    \includegraphics[width=\linewidth]{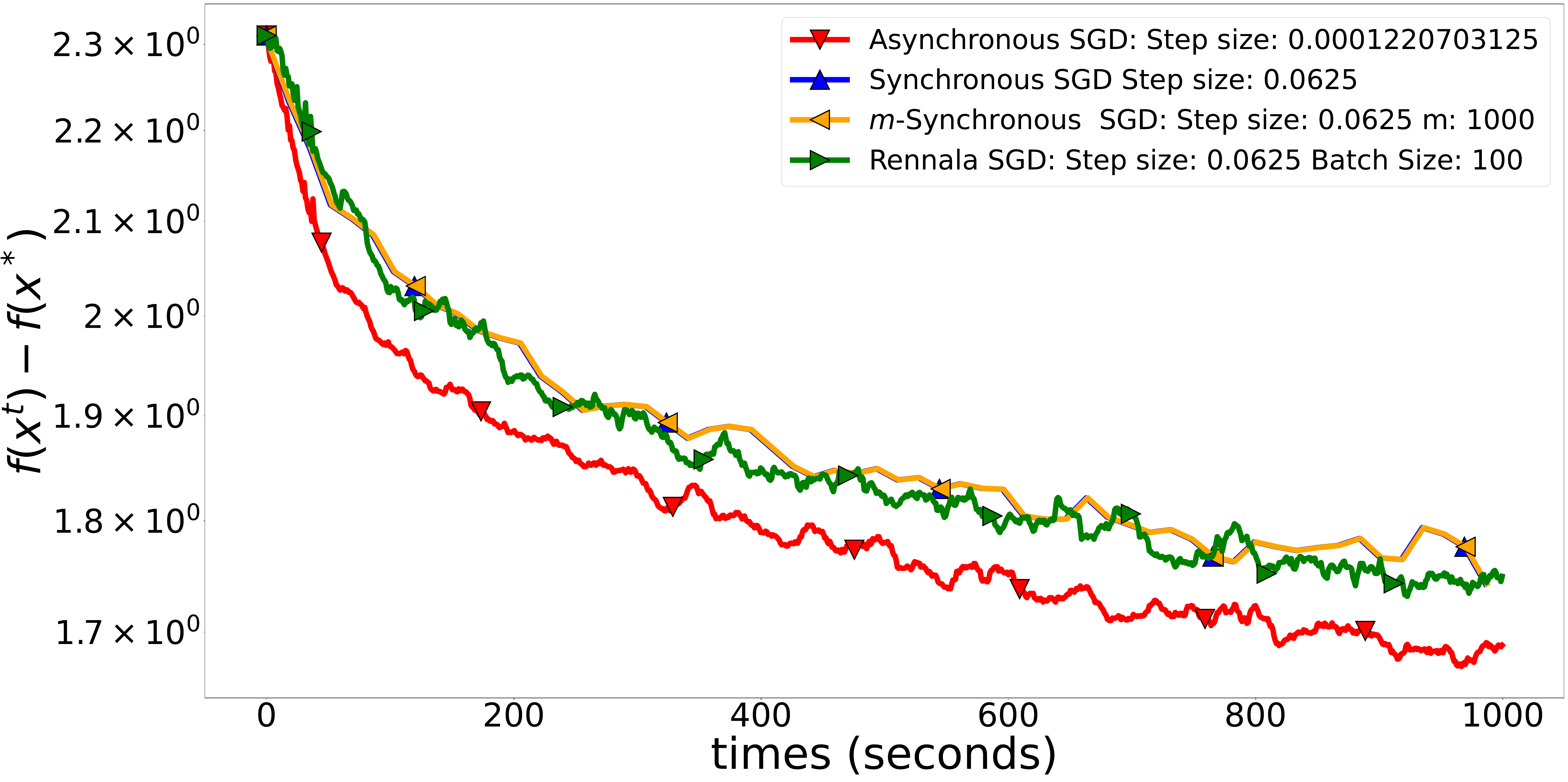}
    \caption*{$\bar{\tau}_i \sim 1 + \textnormal{Unif}([-0.7, 0.7])$}
\end{subfigure}
\begin{subfigure}[b]{0.4\textwidth}
    \includegraphics[width=\linewidth]{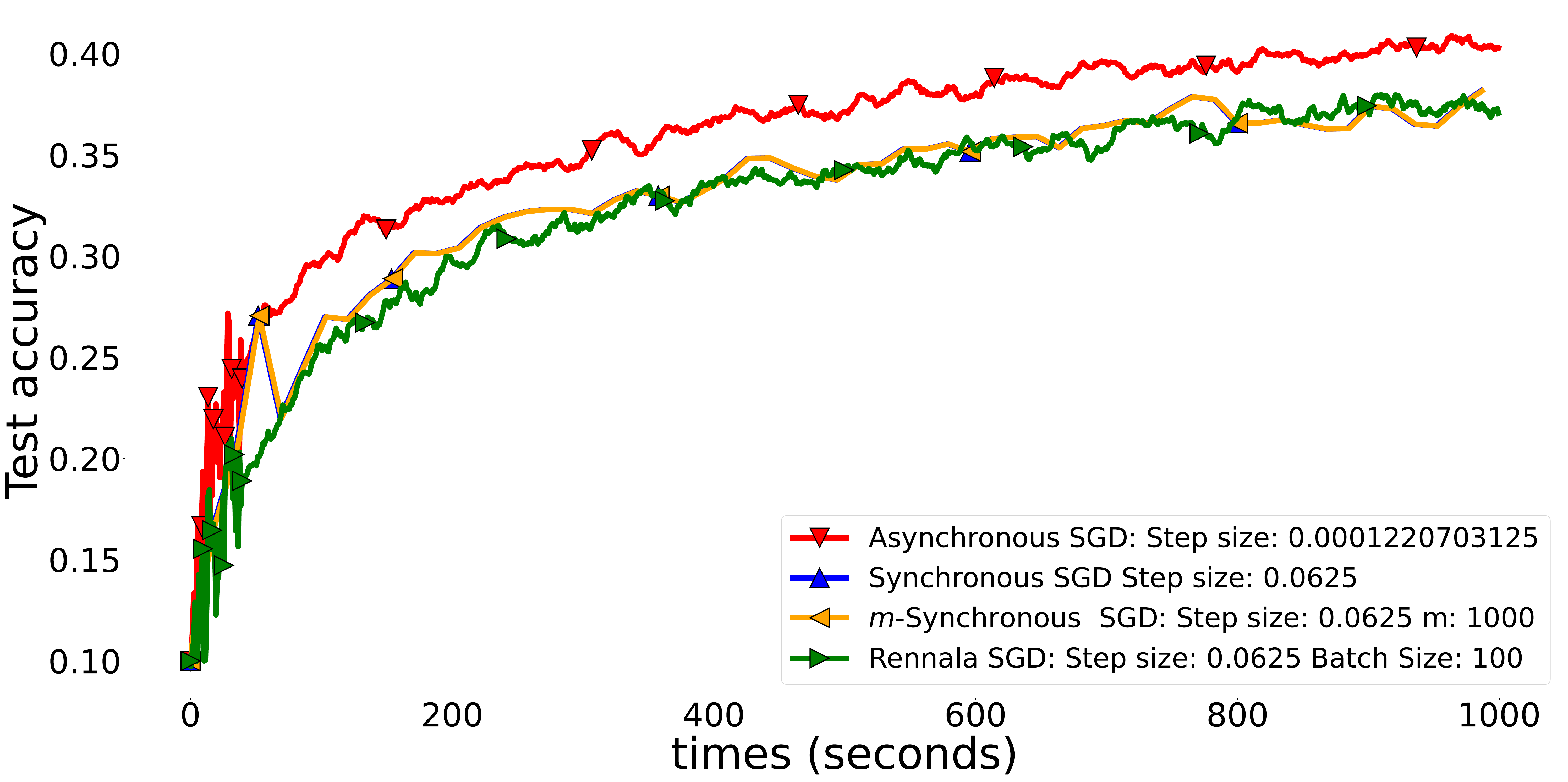}
    \caption*{$\bar{\tau}_i \sim 1 + \textnormal{Unif}([-0.7, 0.7])$}
\end{subfigure}

\caption{Experiments with the CIFAR-10 dataset ($n = 1000$). Each row corresponds to a fixed noise level. The first column corresponds to function values, while the second column corresponds to accuracies.}
\label{fig:cifar_comparison_uniform}
\end{figure}

\begin{figure}[h]
\centering

\begin{subfigure}[b]{0.4\textwidth}
    \includegraphics[width=\linewidth]{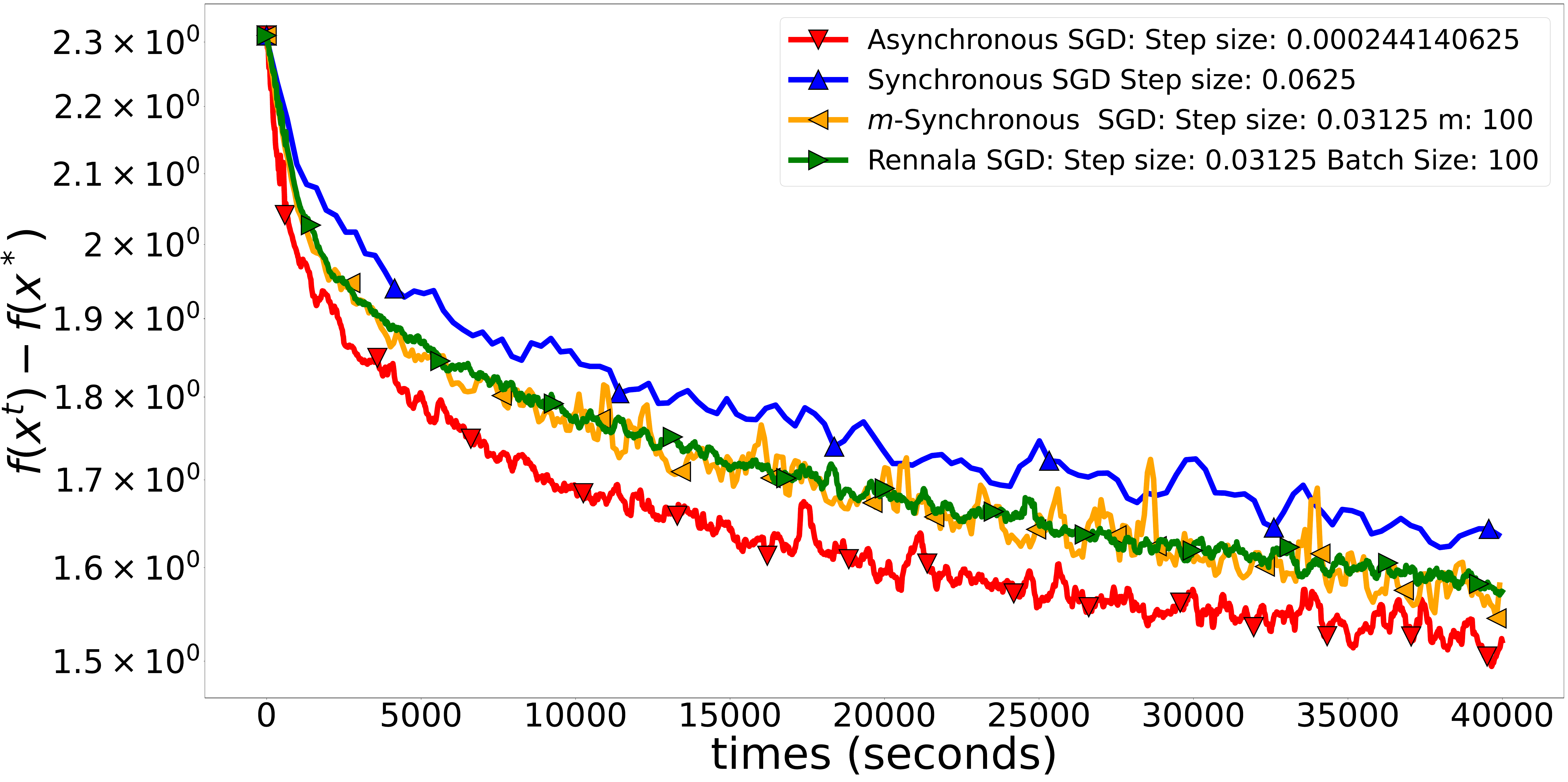}
    \caption*{$\bar{\tau}_i = \sqrt{i}$}
\end{subfigure}
\begin{subfigure}[b]{0.4\textwidth}
    \includegraphics[width=\linewidth]{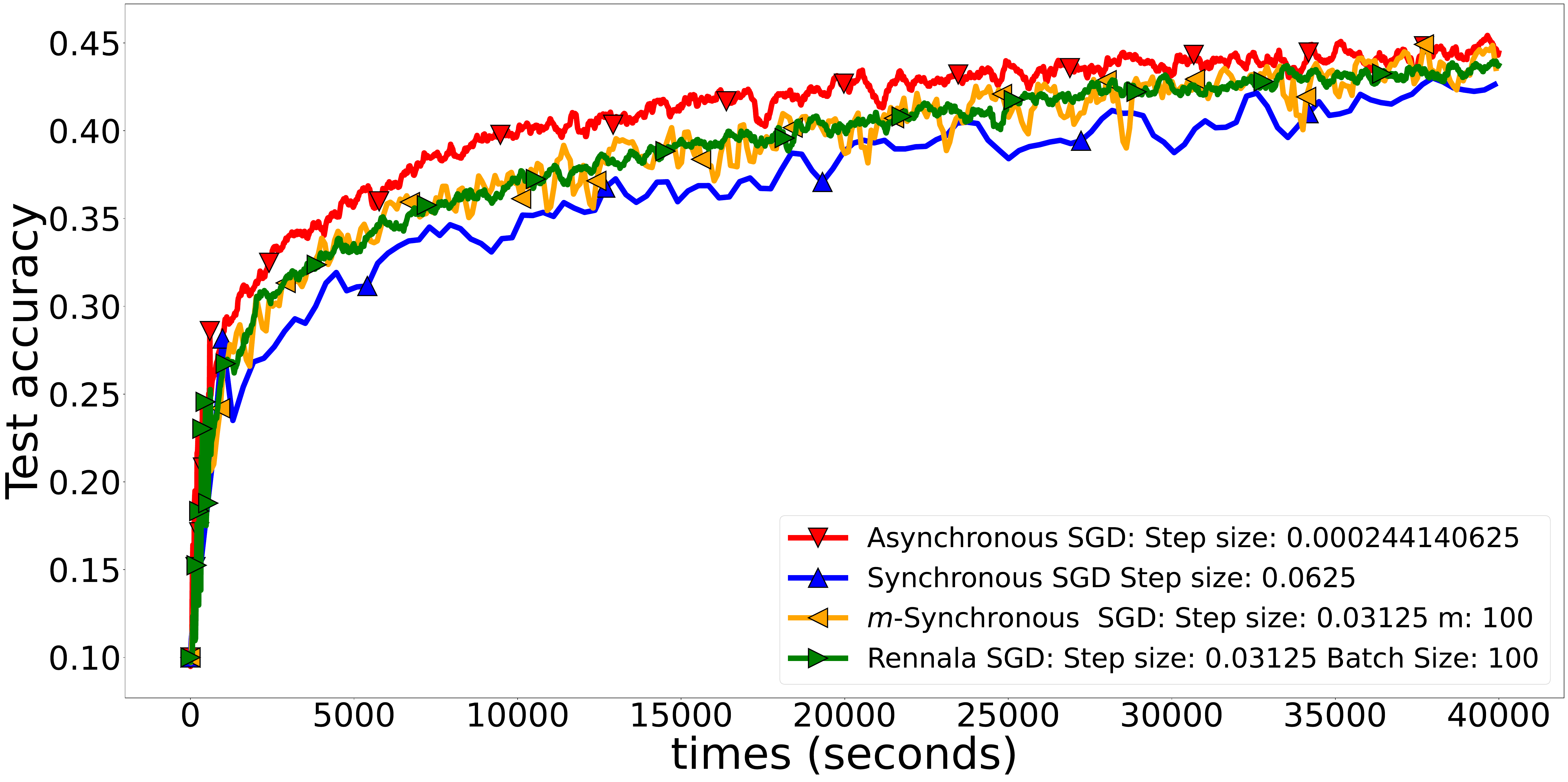}
    \caption*{$\bar{\tau}_i = \sqrt{i}$}
\end{subfigure}

\medskip

\begin{subfigure}[b]{0.4\textwidth}
    \includegraphics[width=\linewidth]{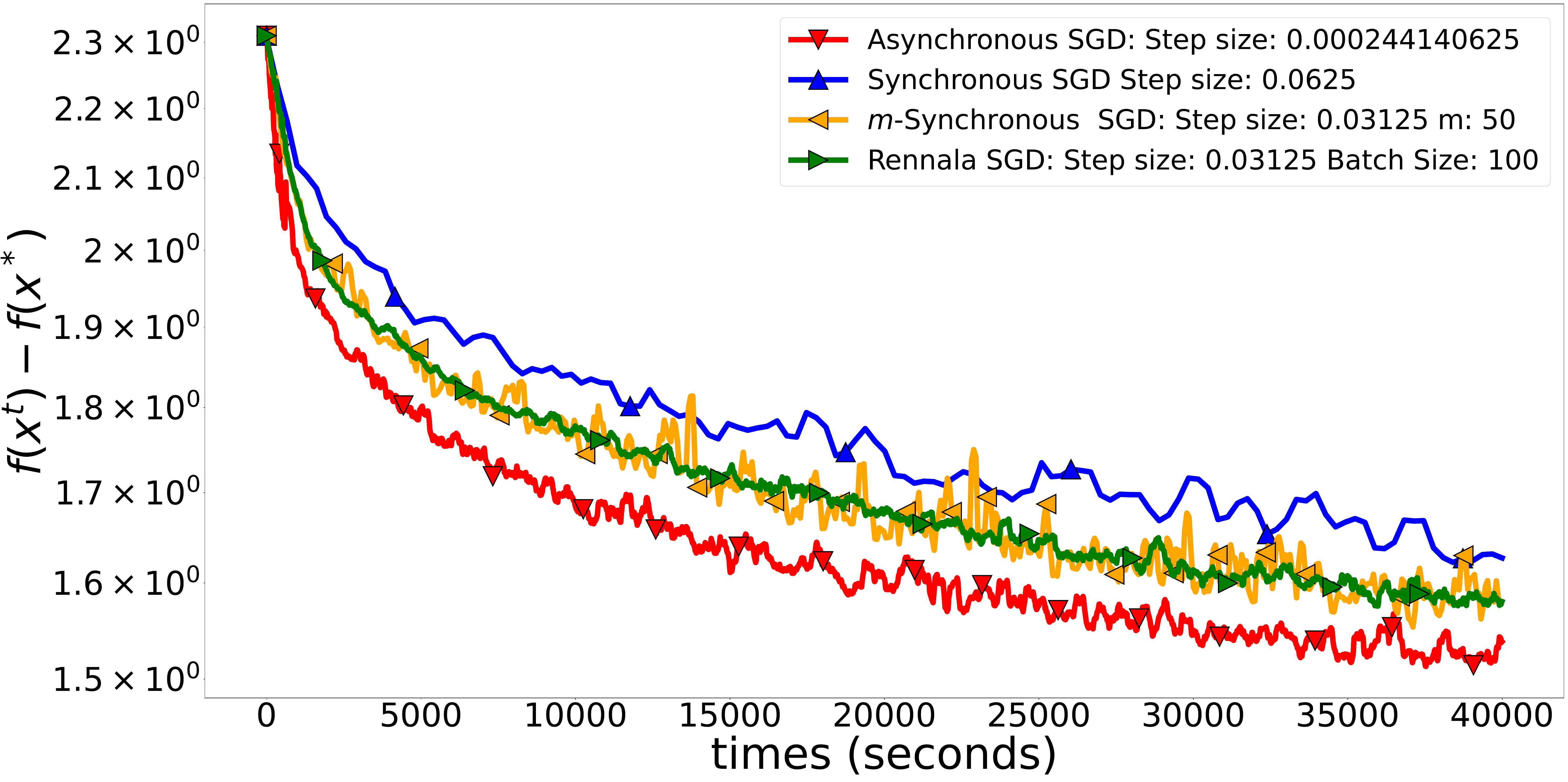}
    \caption*{$\bar{\tau}_i \sim \cN_{\geq 0}(\sqrt{i}, \sigma^2), \,\, \sigma = 0.1$}
\end{subfigure}
\begin{subfigure}[b]{0.4\textwidth}
    \includegraphics[width=\linewidth]{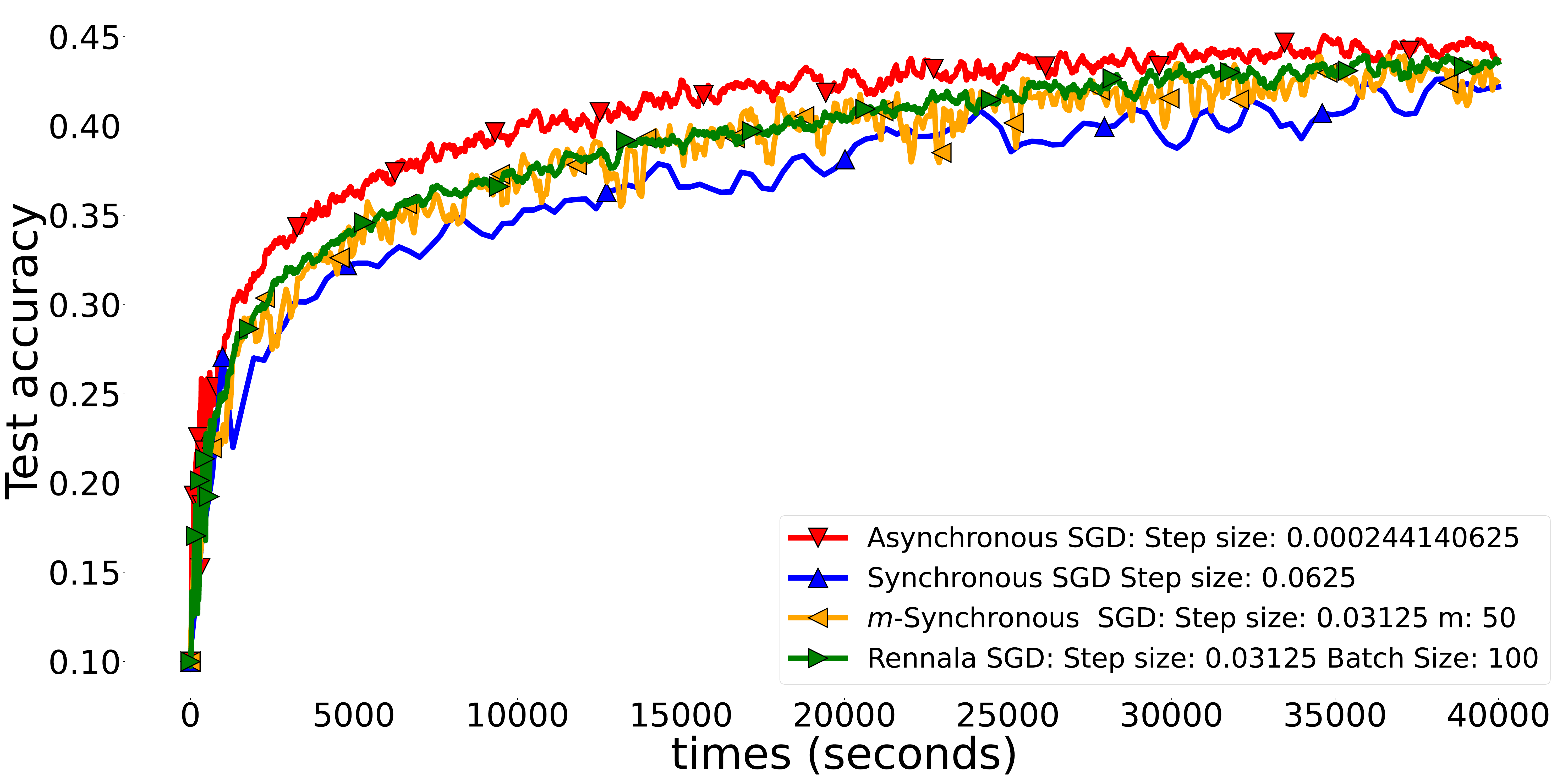}
    \caption*{$\bar{\tau}_i \sim \cN_{\geq 0}(\sqrt{i}, \sigma^2), \,\, \sigma = 0.1$}
\end{subfigure}

\medskip

\begin{subfigure}[b]{0.4\textwidth}
    \includegraphics[width=\linewidth]{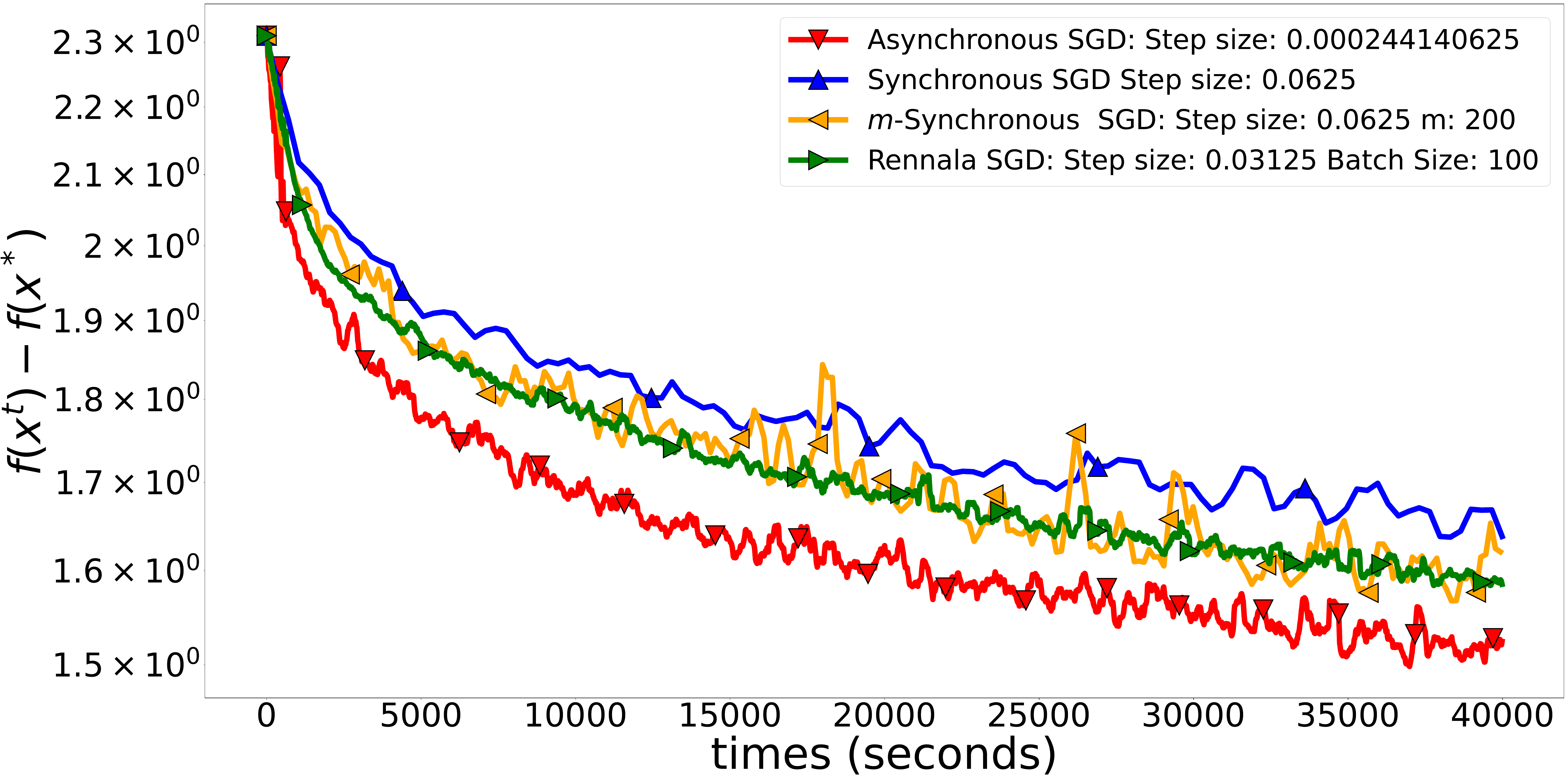}
    \caption*{$\bar{\tau}_i \sim \cN_{\geq 0}(\sqrt{i}, \sigma^2), \,\, \sigma = 1$}
\end{subfigure}
\begin{subfigure}[b]{0.4\textwidth}
    \includegraphics[width=\linewidth]{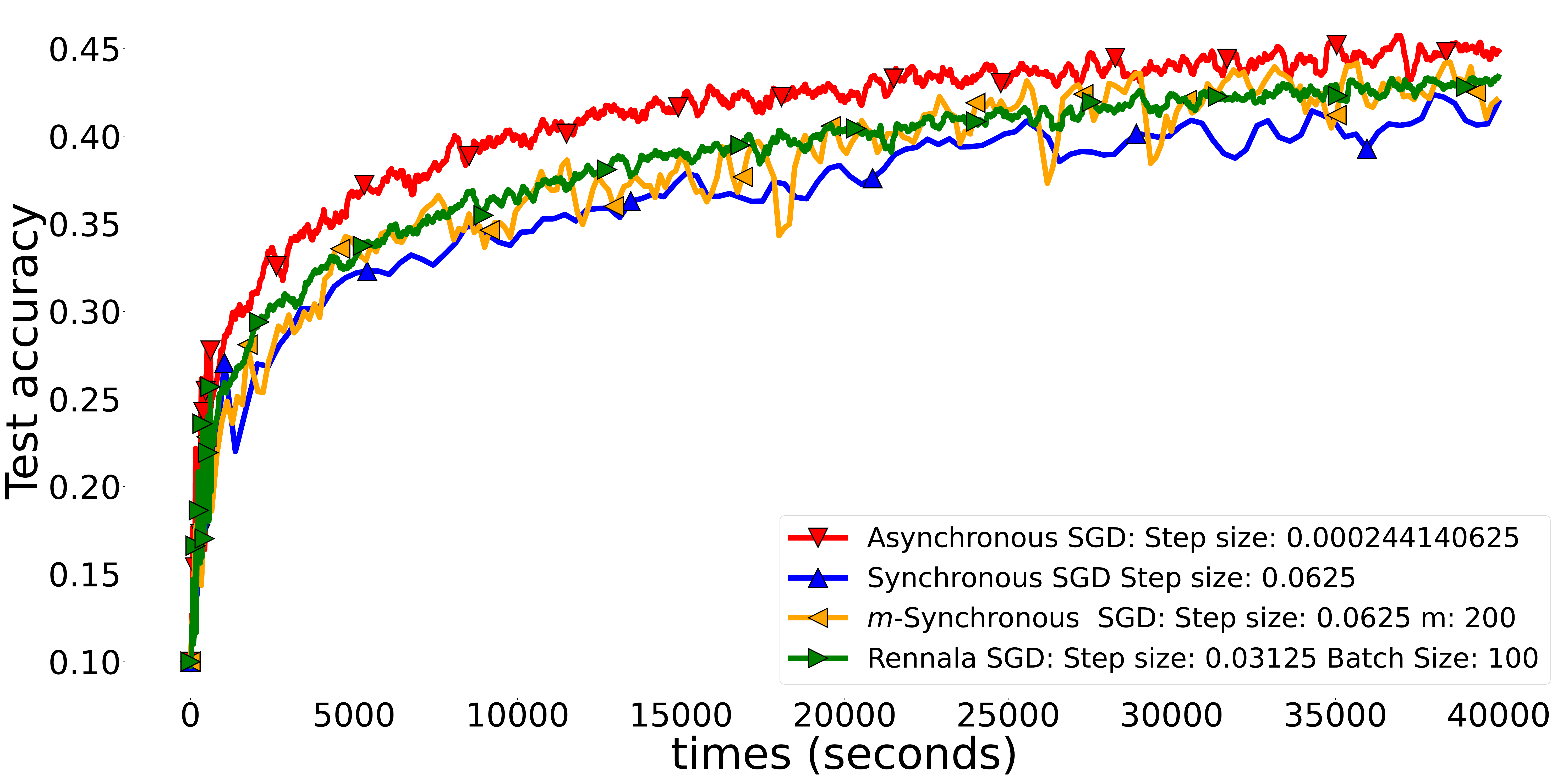}
    \caption*{$\bar{\tau}_i \sim \cN_{\geq 0}(\sqrt{i}, \sigma^2), \,\, \sigma = 1$}
\end{subfigure}

\caption{Experiments with the CIFAR-10 dataset ($n = 1000$). Each row corresponds to a fixed noise level. The first column corresponds to function values, while the second column corresponds to accuracies.}
\label{fig:cifar_comparison}
\end{figure}

\begin{figure}[h]
\centering

\begin{subfigure}[b]{0.4\textwidth}
    \includegraphics[width=\linewidth]{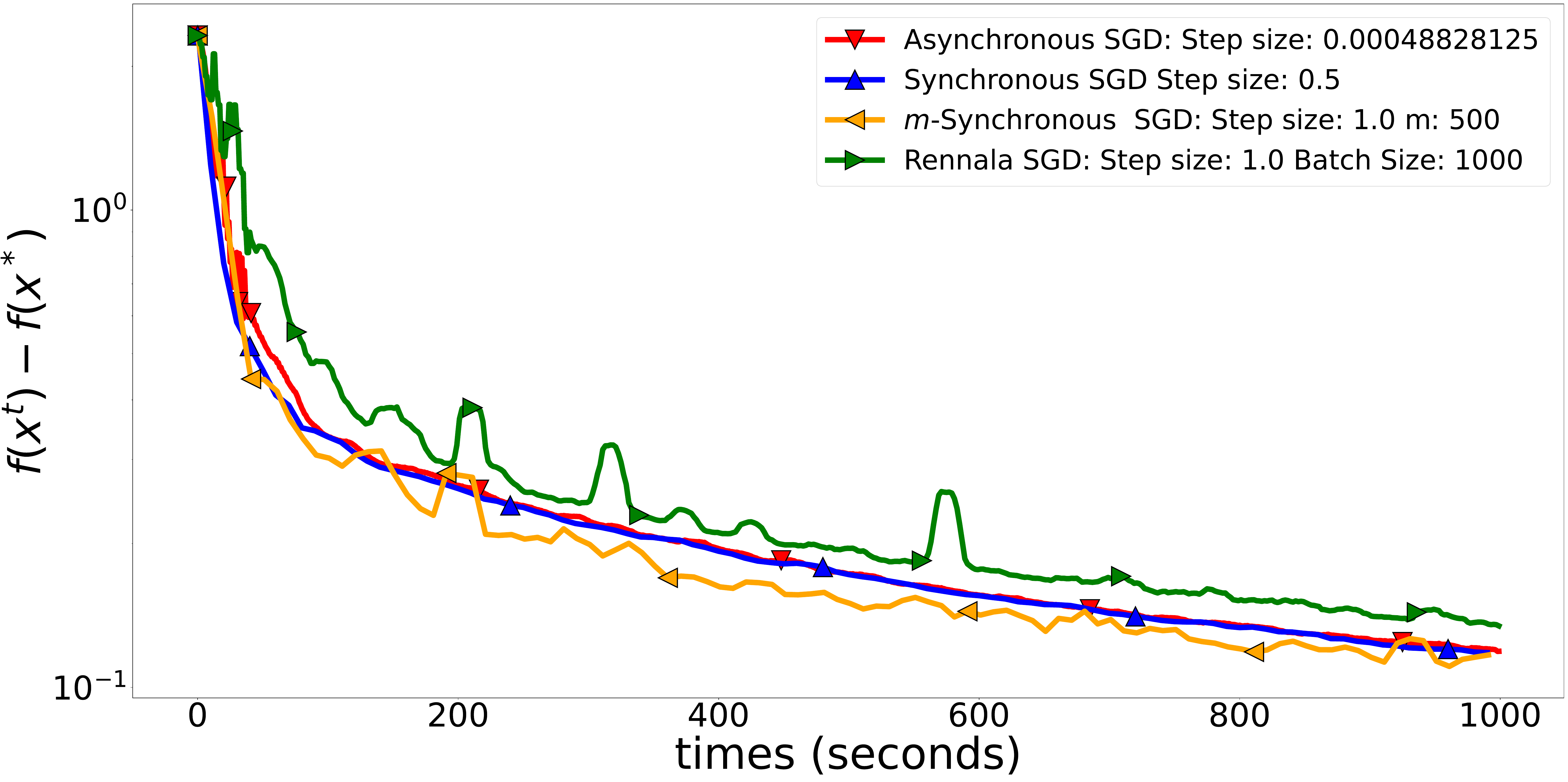}
    \caption*{$\bar{\tau}_i = 1$}
\end{subfigure}
\begin{subfigure}[b]{0.4\textwidth}
    \includegraphics[width=\linewidth]{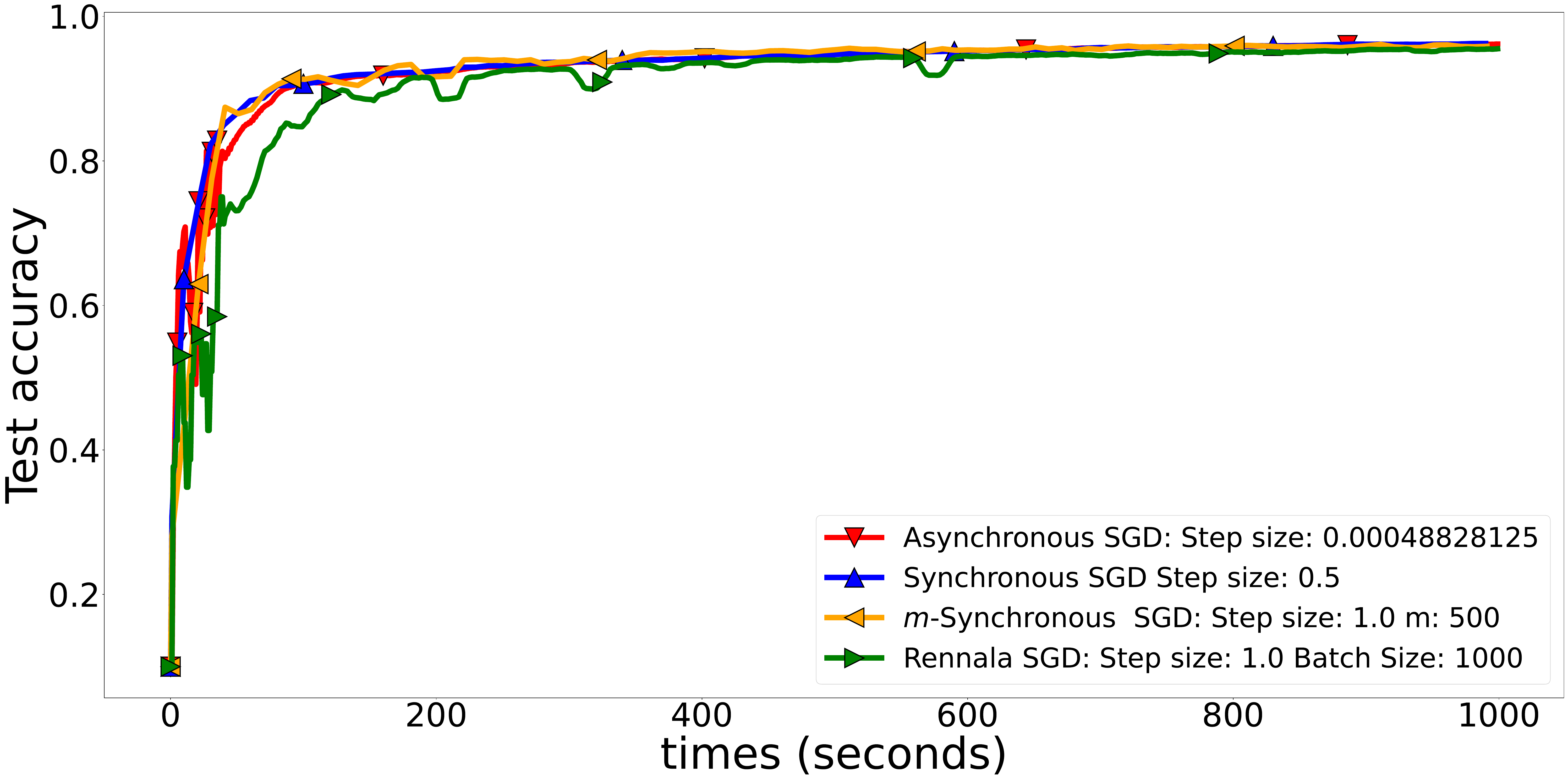}
    \caption*{$\bar{\tau}_i = 1$}
\end{subfigure}

\medskip

\begin{subfigure}[b]{0.4\textwidth}
    \includegraphics[width=\linewidth]{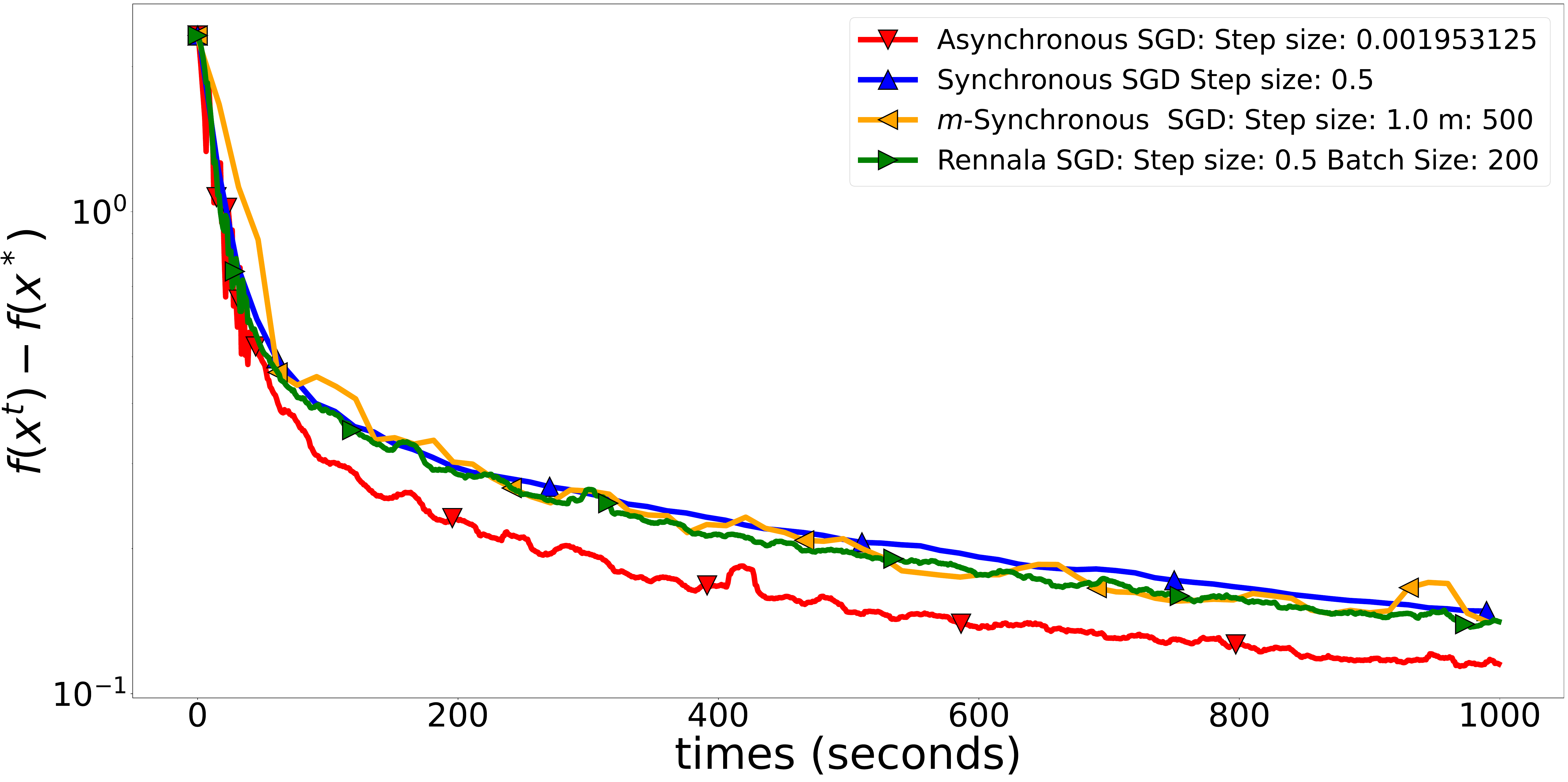}
    \caption*{$\bar{\tau}_i \sim 1 + \textnormal{Unif}([-0.5, 0.5])$}
\end{subfigure}
\begin{subfigure}[b]{0.4\textwidth}
    \includegraphics[width=\linewidth]{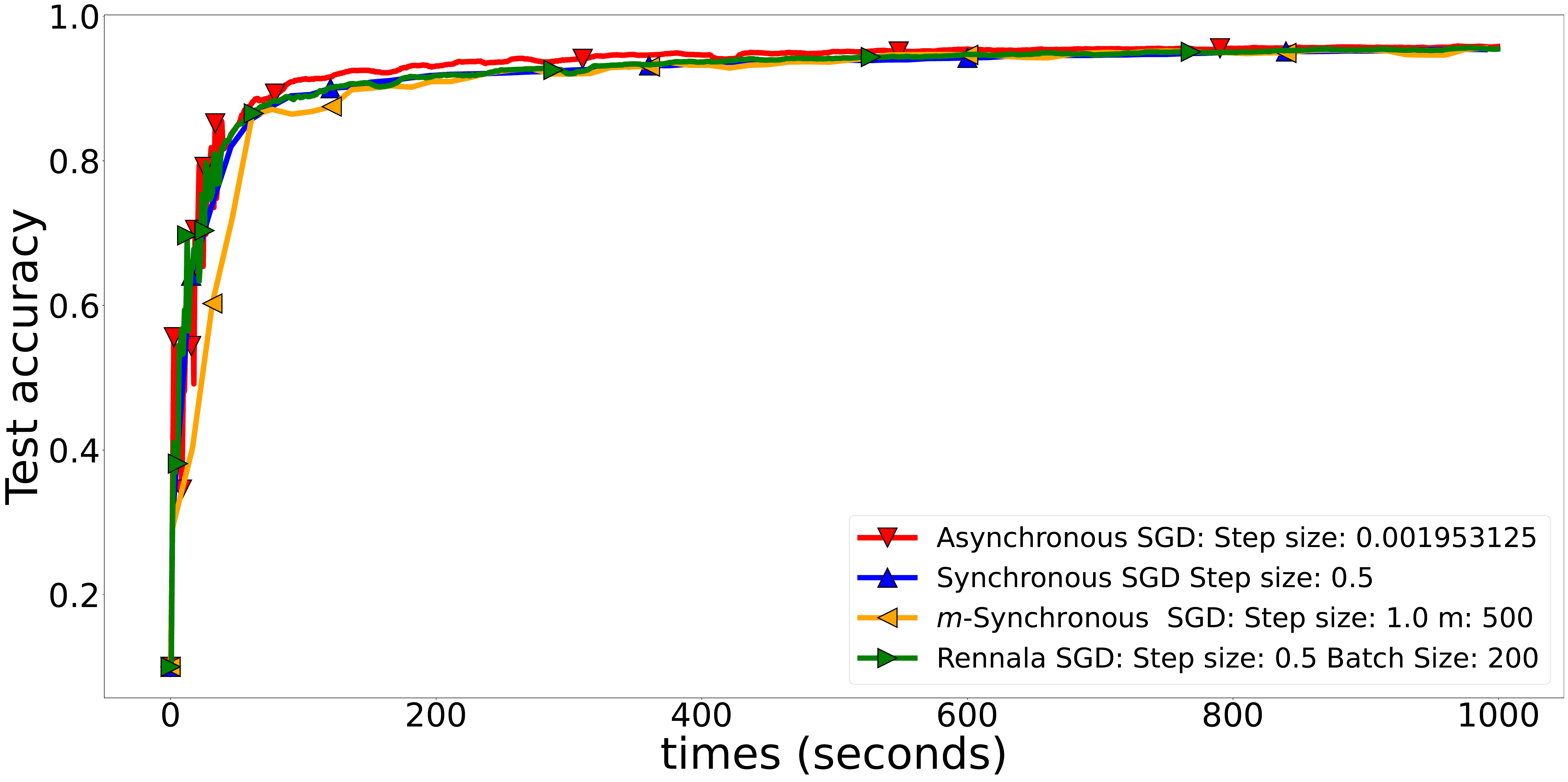}
    \caption*{$\bar{\tau}_i \sim 1 + \textnormal{Unif}([-0.5, 0.5])$}
\end{subfigure}

\medskip

\begin{subfigure}[b]{0.4\textwidth}
    \includegraphics[width=\linewidth]{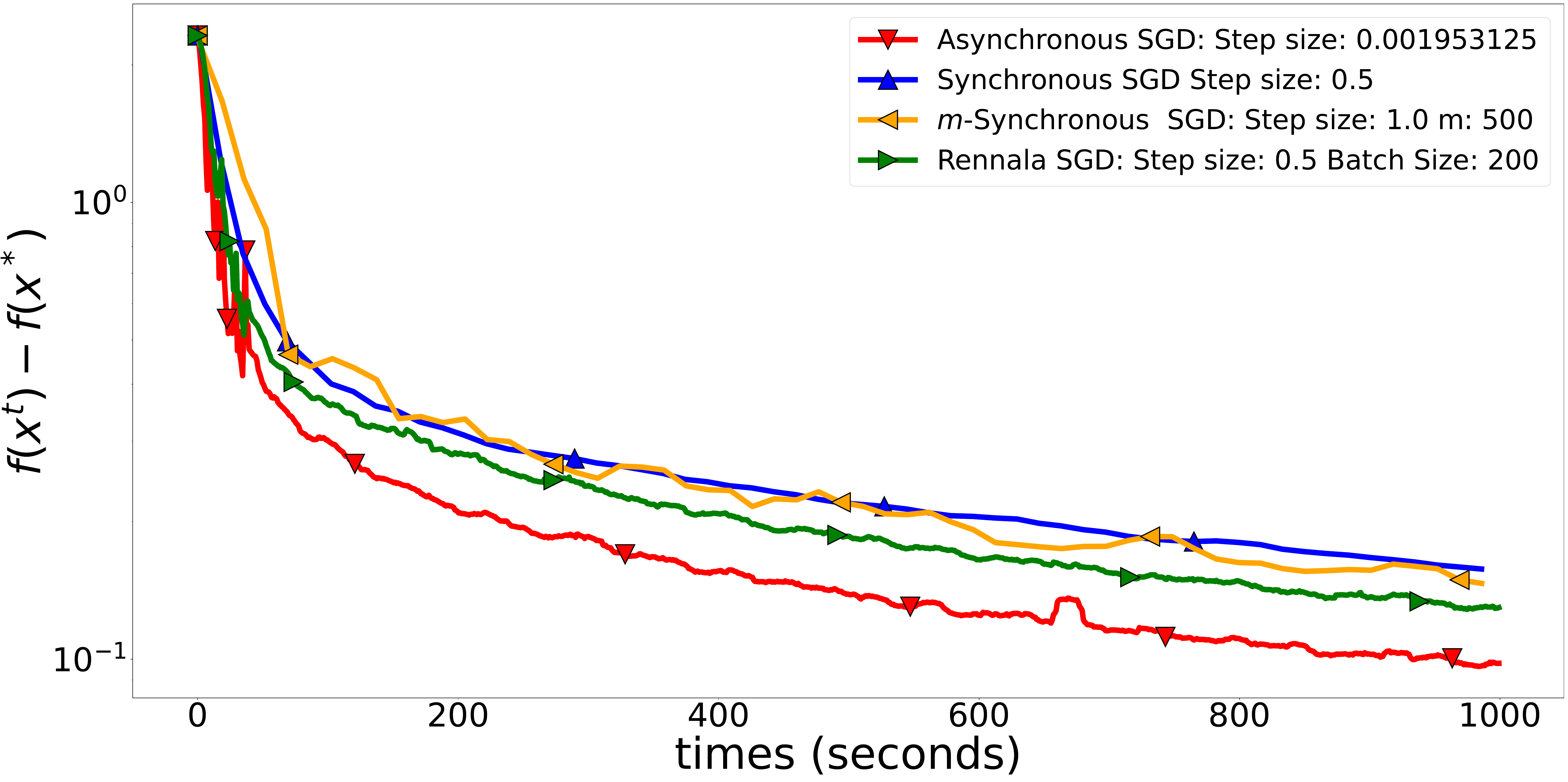}
    \caption*{$\bar{\tau}_i \sim 1 + \textnormal{Unif}([-0.7, 0.7])$}
\end{subfigure}
\begin{subfigure}[b]{0.4\textwidth}
    \includegraphics[width=\linewidth]{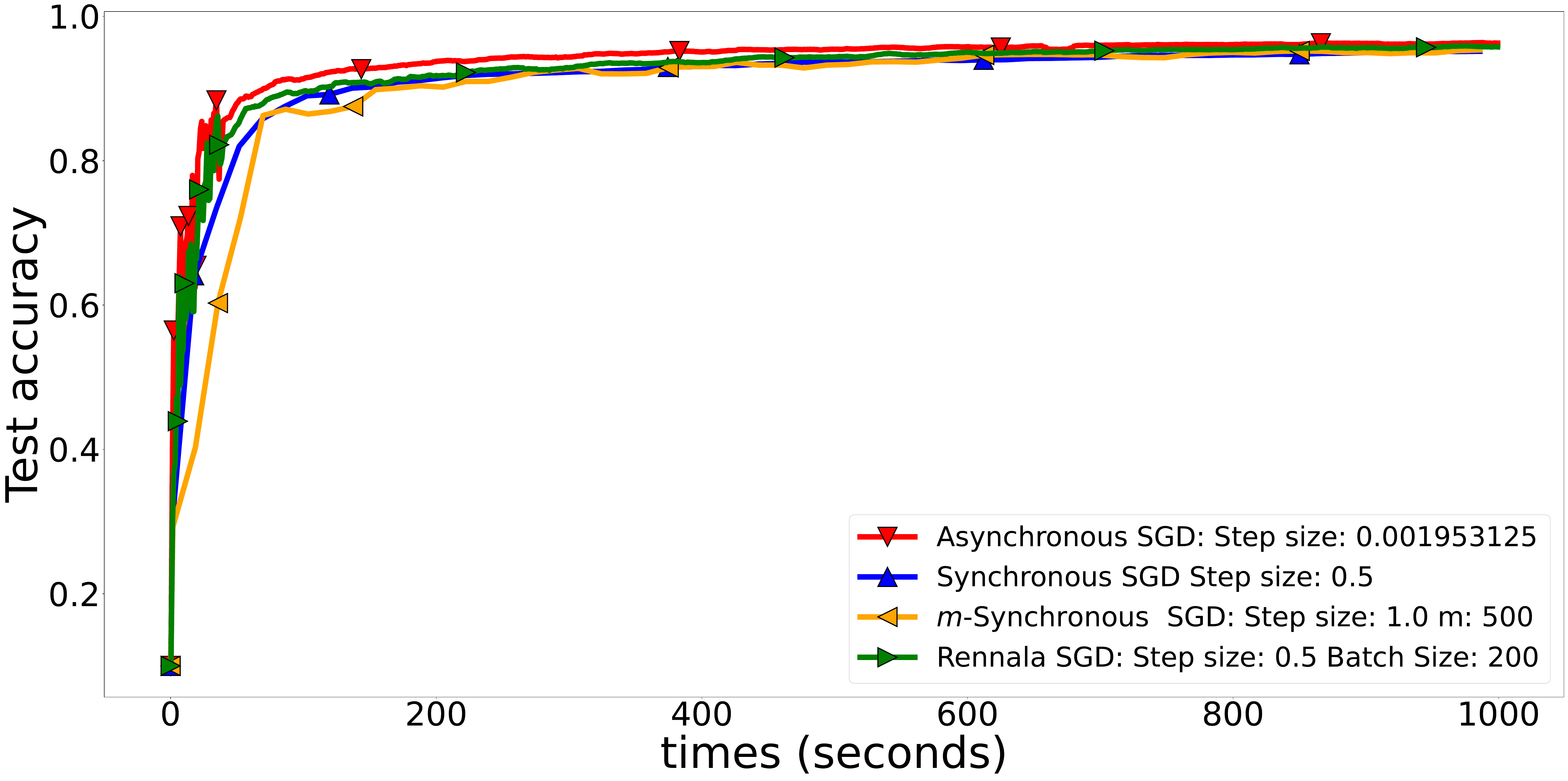}
    \caption*{$\bar{\tau}_i \sim 1 + \textnormal{Unif}([-0.7, 0.7])$}
\end{subfigure}

\caption{Experiments with the MNIST dataset ($n = 1000$). Each row corresponds to a fixed noise level.
The first column corresponds to function values, while the second column corresponds to accuracies.}
\label{fig:mnist_comparison_uniform}
\end{figure}

\begin{figure}[h]
\centering

\begin{subfigure}[b]{0.4\textwidth}
    \includegraphics[width=\linewidth]{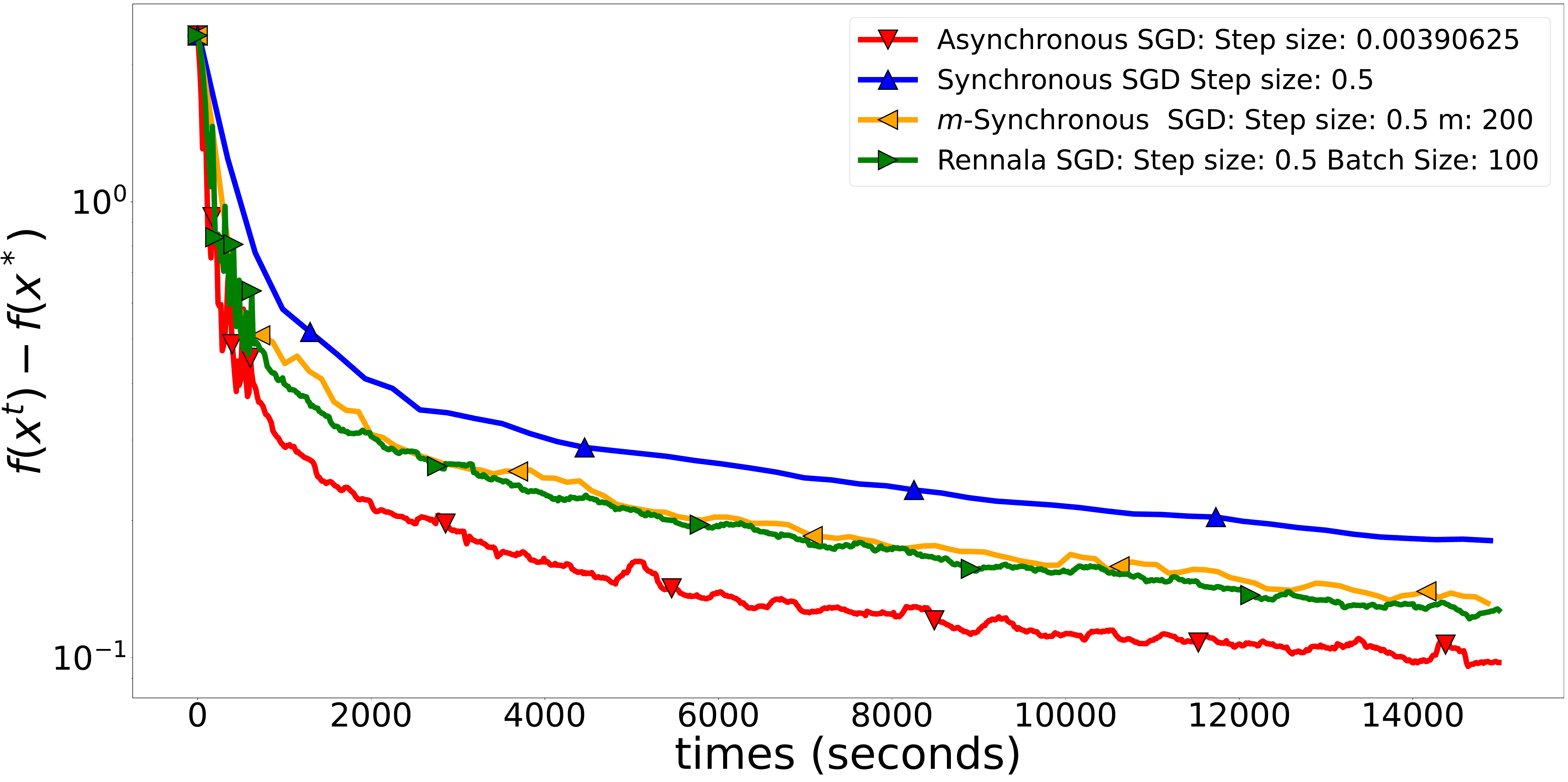}
    \caption*{$\bar{\tau}_i = \sqrt{i}$}
\end{subfigure}
\begin{subfigure}[b]{0.4\textwidth}
    \includegraphics[width=\linewidth]{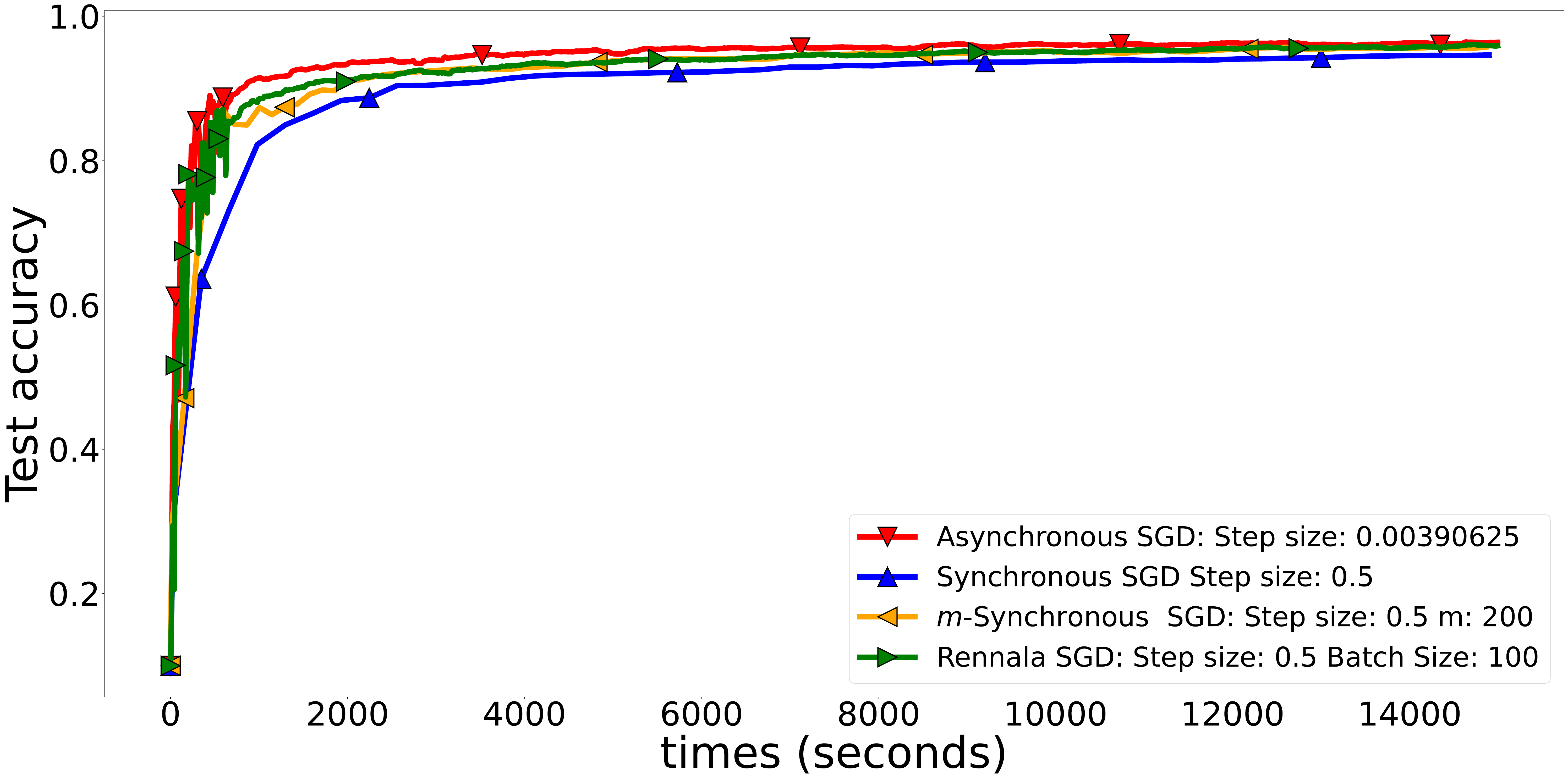}
    \caption*{$\bar{\tau}_i = \sqrt{i}$}
\end{subfigure}

\medskip

\begin{subfigure}[b]{0.4\textwidth}
    \includegraphics[width=\linewidth]{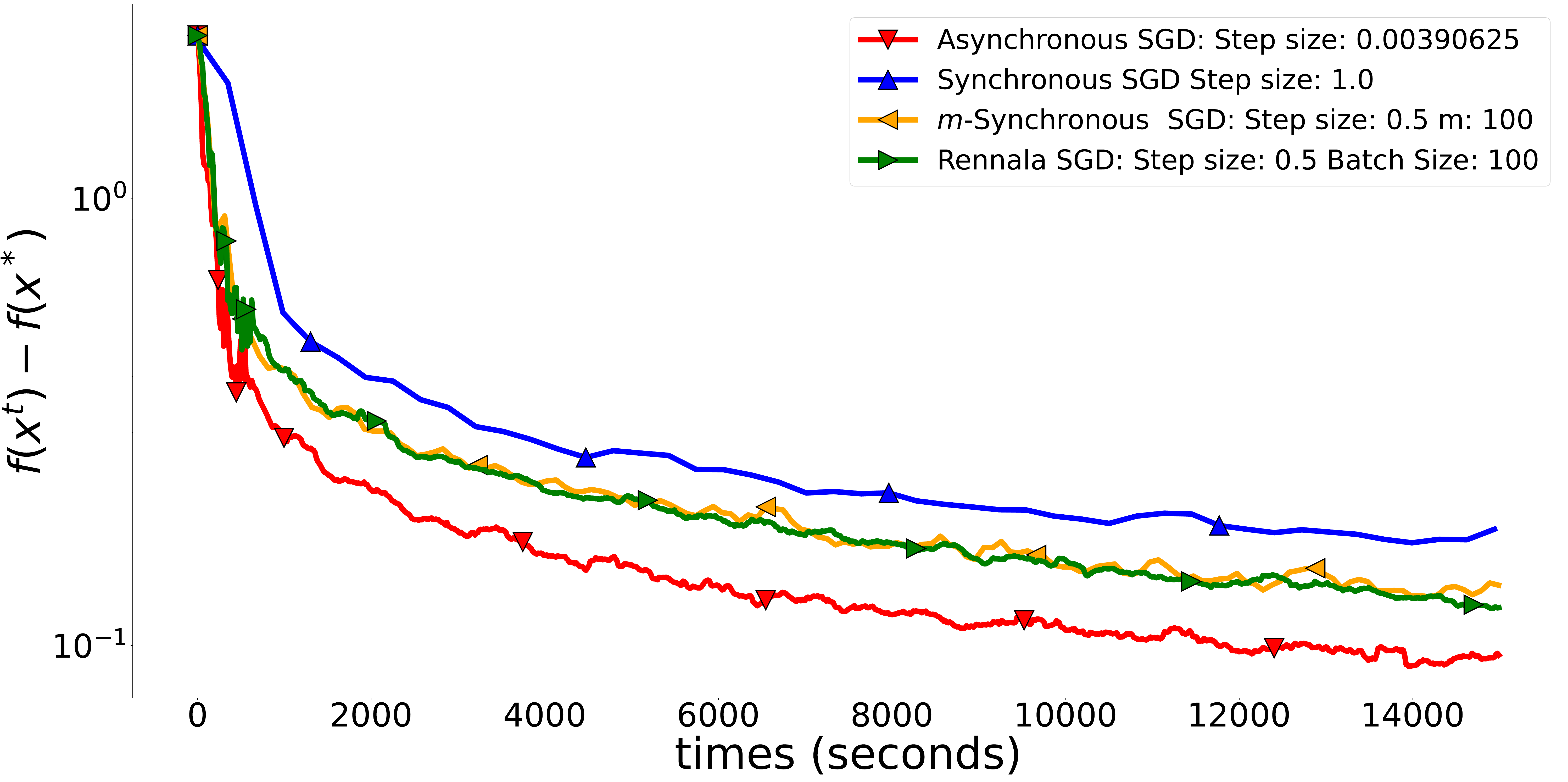}
    \caption*{$\bar{\tau}_i \sim \cN_{\geq 0}(\sqrt{i}, 0.1^2)$}
\end{subfigure}
\begin{subfigure}[b]{0.4\textwidth}
    \includegraphics[width=\linewidth]{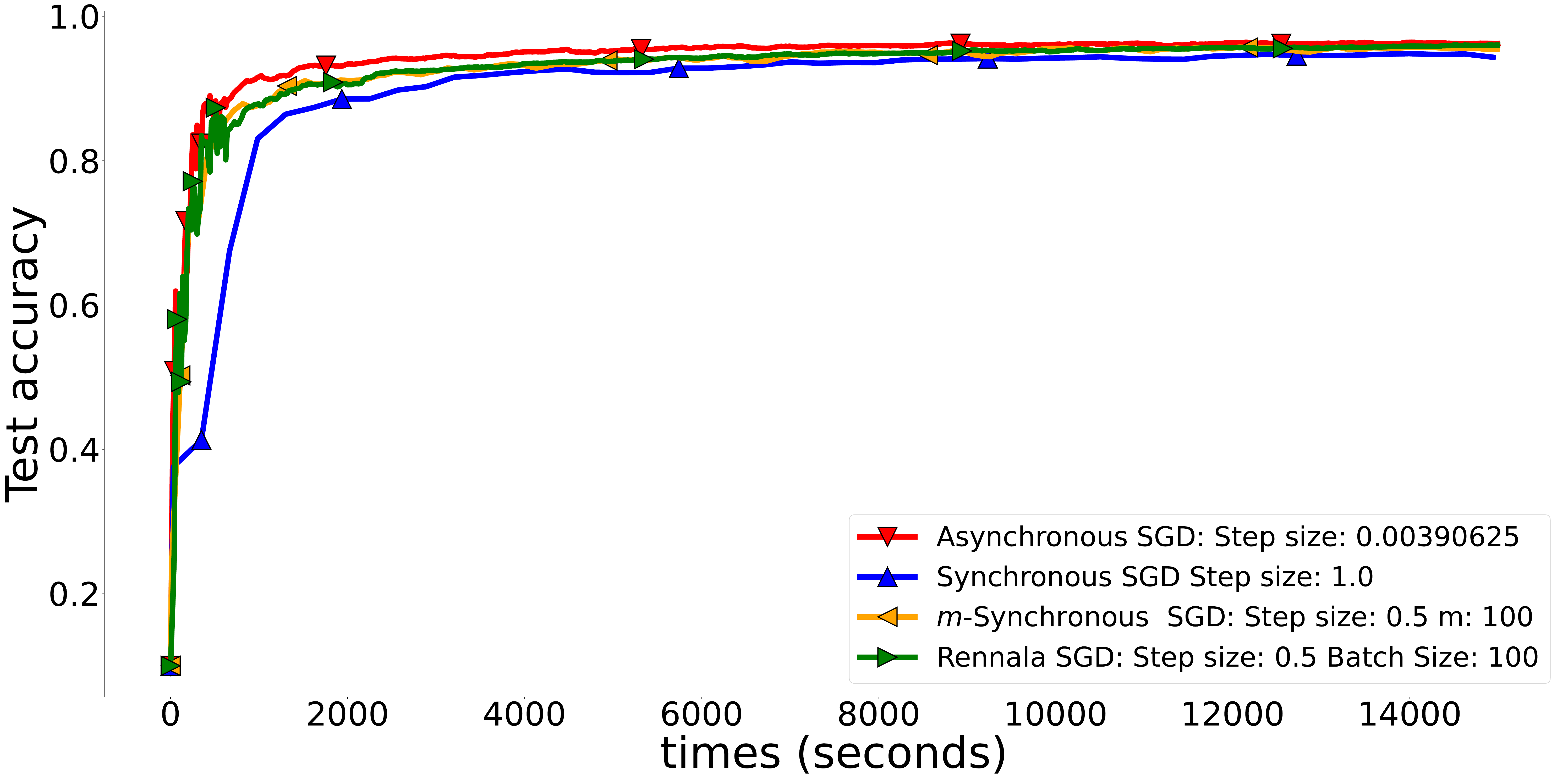}
    \caption*{$\bar{\tau}_i \sim \cN_{\geq 0}(\sqrt{i}, 0.1^2)$}
\end{subfigure}

\medskip

\begin{subfigure}[b]{0.4\textwidth}
    \includegraphics[width=\linewidth]{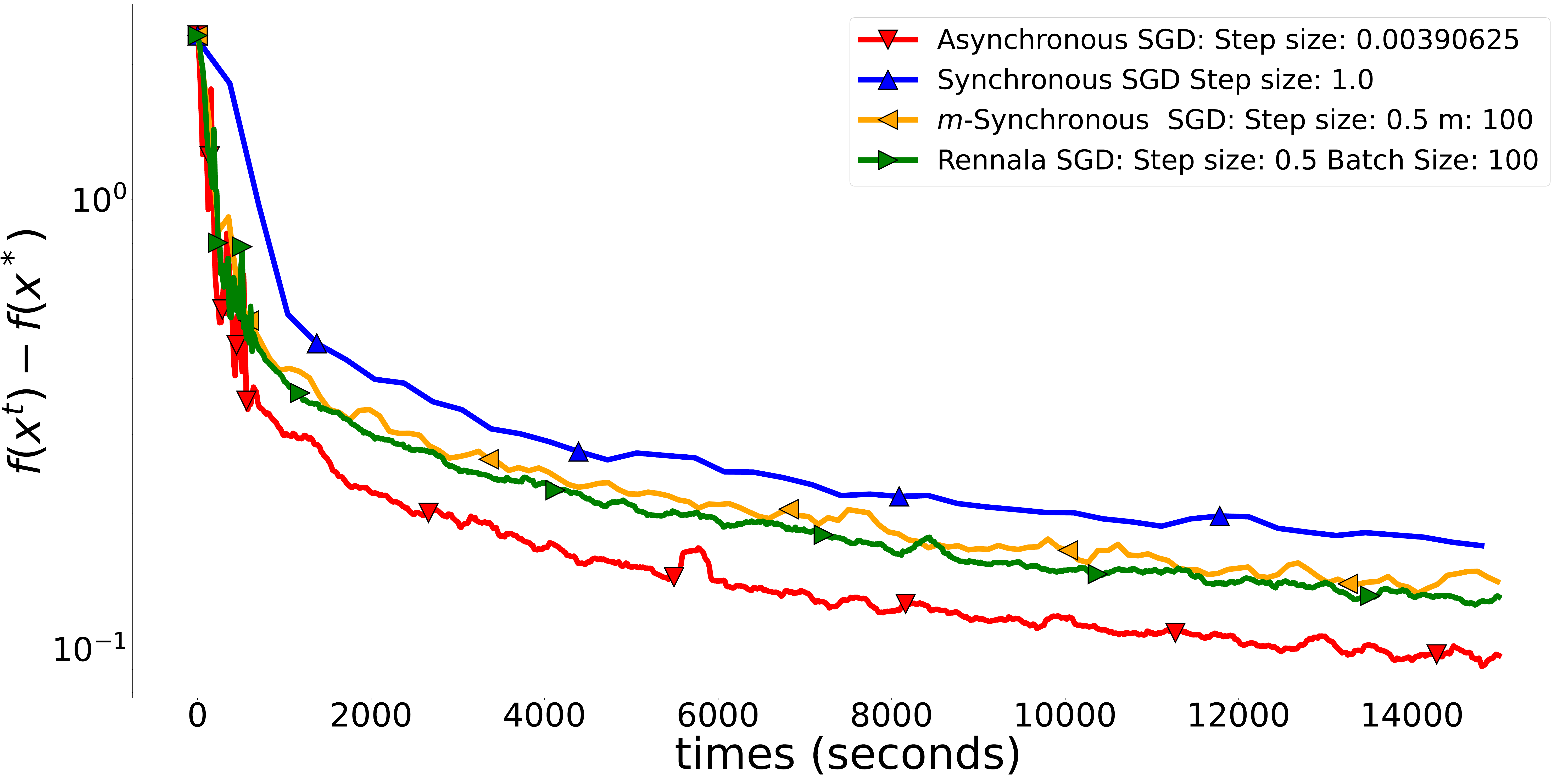}
    \caption*{$\bar{\tau}_i \sim \cN_{\geq 0}(\sqrt{i}, 1^2)$}
\end{subfigure}
\begin{subfigure}[b]{0.4\textwidth}
    \includegraphics[width=\linewidth]{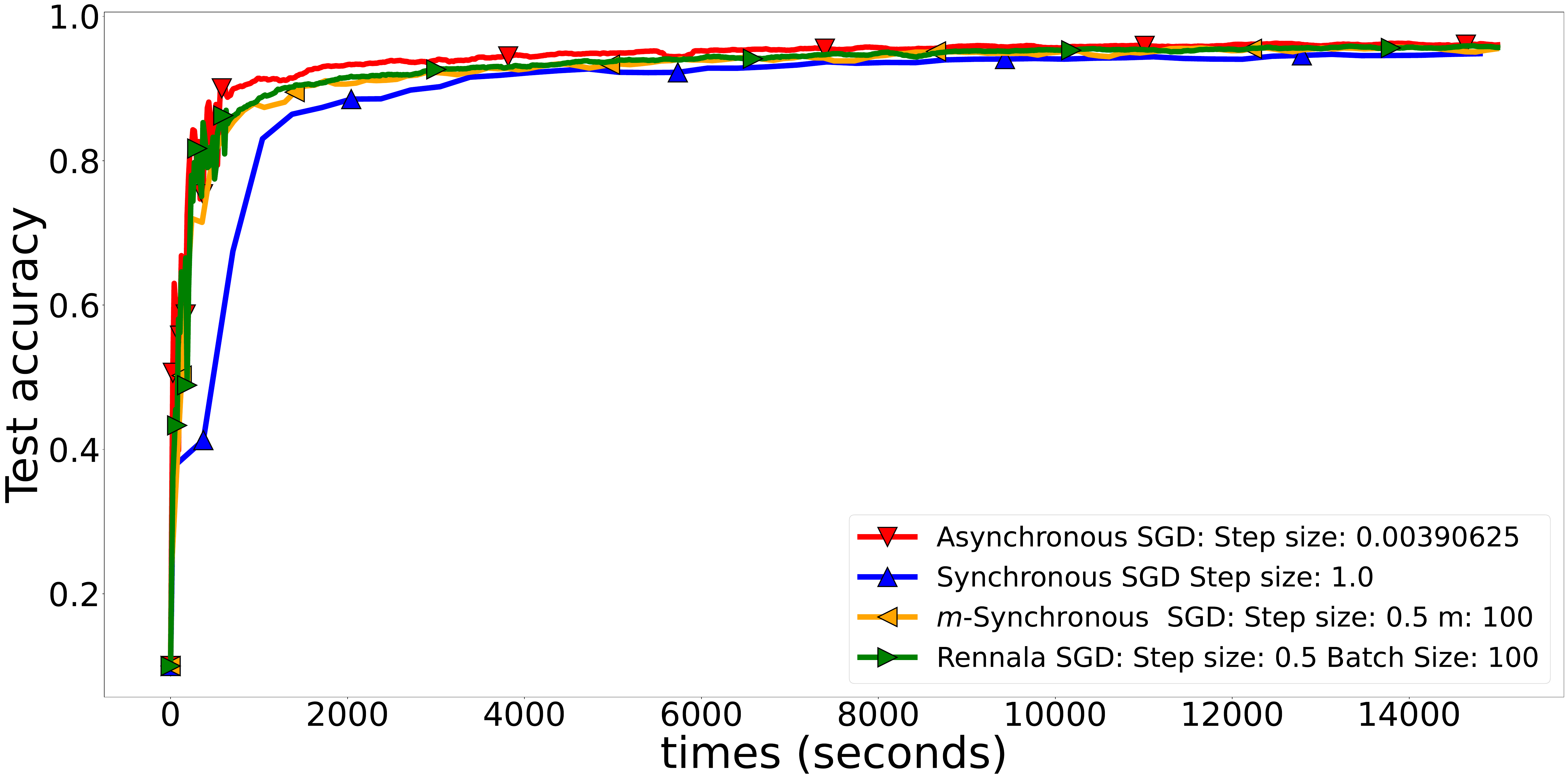}
    \caption*{$\bar{\tau}_i \sim \cN_{\geq 0}(\sqrt{i}, 1^2)$}
\end{subfigure}

\caption{Experiments with the MNIST dataset ($n = 1000$). Each row corresponds to a fixed noise level.
The first column corresponds to function values, while the second column corresponds to accuracies.}
\label{fig:mnist_comparison}
\end{figure}

\clearpage
\newpage
\subsection{NanoGPT experiments}
\label{sec:nanogpt_another}
In this section, we compare Synchronous SGD (Algorithm~\ref{alg:alg_orig}) and Asynchronous SGD (Algorithm~\ref{alg:asgd}) in the NanoGPT library by \citet{nanogpt}. We implement custom versions of the algorithms using \texttt{torch.distributed} from \citep{paszke2019pytorch}. We use the same parameters as in the \texttt{train\_shakespeare\_char.py} configuration of the library. We run the algorithms with $4$ workers operating in parallel. 
This is the setup with times having homogeneous means; however, the times themselves are random due to naturally unstable computation times. This setup is captured by Assumption~\ref{ass:general_heter_random_computation_model} with equal means.
In Figure~\ref{fig:nanogpt_sync_vs_async}, we present the training and validation losses versus the real measured wall-clock time, and observe that the convergence rates of the algorithms are similar. For a fair comparison, we implemented custom code without the support of code-optimized libraries such as \texttt{Horovod} \citep{sergeev2018horovod} or \texttt{torch.nn.parallel.DistributedDataParallel}, which can potentially improve the performance of Synchronous SGD.
\begin{figure}[h]
    \centering
    \includegraphics[width=0.8\linewidth]{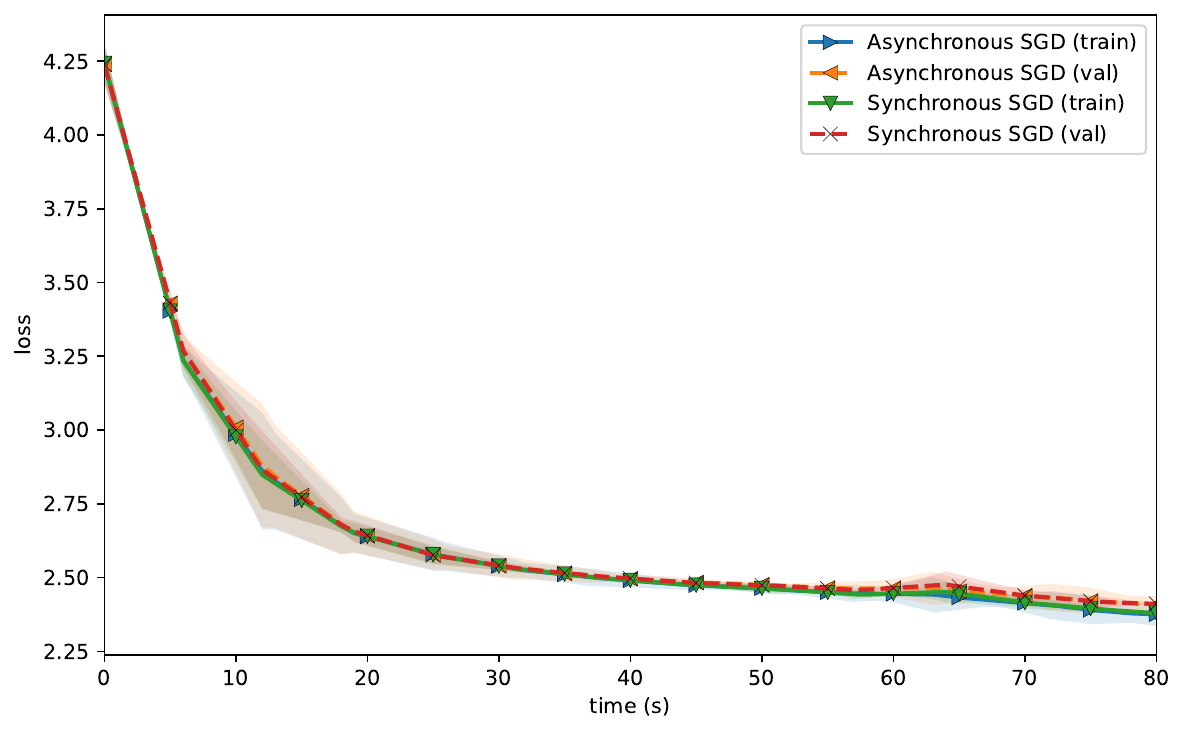}
    \caption{Training and validation losses versus measured wall-clock time for Synchronous SGD (Algorithm~\ref{alg:alg_orig}) and Asynchronous SGD (Algorithm~\ref{alg:asgd}) in NanoGPT, using $4$ workers. We run the experiments with $3$ seeds and plot the error bars.}
    \label{fig:nanogpt_sync_vs_async}
\end{figure}

\newpage
\input{checklist.tex}

\end{document}

%% file: macros.tex
\usepackage{graphicx}
\usepackage{apptools}
\usepackage[flushleft]{threeparttable}
\usepackage{array,booktabs,makecell}
\usepackage{multirow}
\usepackage{nicefrac}
\usepackage{amsthm}
\usepackage{amsmath,amsfonts,bm,amssymb}
\usepackage{wrapfig}
\usepackage{caption}
\usepackage{siunitx}
\usepackage{nccmath}
\usepackage{empheq}
\usepackage{bbm}
\usepackage{tabularx}

\usepackage{algorithmic}
\usepackage{algorithm}
\usepackage{filecontents}

\usepackage[capitalize,noabbrev]{cleveref}

\usepackage{color}
\usepackage{colortbl}
\definecolor{bgcolor}{rgb}{0.76,0.88,0.50}
\definecolor{bgcolor0}{rgb}{0.93,0.99,1}
\definecolor{bgcolor1}{rgb}{0.8,1,1}
\definecolor{bgcolor2}{rgb}{0.8,1,0.8}
\definecolor{bgcolor3}{rgb}{0.50,0.90,0.50}
\usepackage{tcolorbox}
\usepackage{pifont}

\definecolor{mydarkgreen}{RGB}{39,130,67}
\definecolor{mydarkorange}{RGB}{236,147,14}
\definecolor{mydarkred}{RGB}{192,47,25}
\definecolor{ruby}{RGB}{155,17,30}
\definecolor{chili}{RGB}{191,0,0}
\definecolor{sangria}{RGB}{146,0,10}
\definecolor{burgundy}{RGB}{128,0,32} 
\definecolor{darkred}{RGB}{132,0,0} 
\definecolor{cherry}{RGB}{192,0,0} 
\definecolor{grey}{RGB}{80,80,80} 

\definecolor{blue}{RGB}{0,0,255}

\usepackage[textsize=tiny]{todonotes}

\usepackage{xspace}

\newcommand{\norm}[1]{\left\| #1 \right\|}
\newcommand{\sqnorm}[1]{\left\| #1 \right\|^2}
\newcommand{\abs}[1]{\left| #1 \right|}

\newcommand{\R}{\mathbb{R}} 
\newcommand{\N}{\mathbb{N}} 

\newcommand{\Exp}[1]{{\mathbb{E}}\left[#1\right]}
\newcommand{\ExpSub}[2]{{\mathbb{E}}_{#1}\left[#2\right]}
\newcommand{\ExpCond}[2]{{\mathbb{E}}\left[\left.#1\right\vert#2\right]}

\newcommand{\ceil}[1]{\left\lceil#1\right\rceil}

\newcommand{\Prob}[1]{\mathbb{P}\left(#1\right)} 

\newcommand{\cN}{\mathcal{N}}
\newcommand{\cO}{\mathcal{O}}

\newcommand{\cZ}{\mathcal{Z}}

\newcommand{\eqdef}{:=}

\makeatletter
\newcommand{\vast}{\bBigg@{4}}

\DeclareMathOperator*{\argmin}{arg\,min}

\def\<{\left\langle}
\def\>{\right\rangle}
\def\[{\left[}
\def\]{\right]}
\def\({\left(}
\def\){\right)}

\newcommand{\flr}[1]{\left\lfloor #1\right\rfloor} 

\theoremstyle{plain}
\newtheorem{theorem}{Theorem}[section]
\newtheorem{proposition}[theorem]{Proposition}

\newtheorem{corollary}[theorem]{Corollary}

\theoremstyle{definition}
\newtheorem{definition}[theorem]{Definition}
\newtheorem{assumption}[theorem]{Assumption}
\theoremstyle{remark}

\newcommand*{\sketchproofname}{Sketch of Proof}

\usepackage{longtable}

%% file: checklist.tex
\section*{NeurIPS Paper Checklist}

\begin{enumerate}

\item {\bf Claims}
    \item[] Question: Do the main claims made in the abstract and introduction accurately reflect the paper's contributions and scope?
    \item[] Answer: \answerYes{} 
    \item[] Justification:

\item {\bf Limitations}
    \item[] Question: Does the paper discuss the limitations of the work performed by the authors?
    \item[] Answer: \answerYes{} 
    \item[] Justification: Section~\ref{sec:where_needed}

\item {\bf Theory assumptions and proofs}
    \item[] Question: For each theoretical result, does the paper provide the full set of assumptions and a complete (and correct) proof?
    \item[] Answer: \answerYes{} 
    \item[] Justification: See the main part and appendix.

    \item {\bf Experimental result reproducibility}
    \item[] Question: Does the paper fully disclose all the information needed to reproduce the main experimental results of the paper to the extent that it affects the main claims and/or conclusions of the paper (regardless of whether the code and data are provided or not)?
    \item[] Answer: \answerYes{} 
    \item[] Justification: Section~\ref{sec:exp_appendix}

\item {\bf Open access to data and code}
    \item[] Question: Does the paper provide open access to the data and code, with sufficient instructions to faithfully reproduce the main experimental results, as described in supplemental material?
    \item[] Answer: \answerYes{} 
    \item[] Justification: In the supplementary materials.

\item {\bf Experimental setting/details}
    \item[] Question: Does the paper specify all the training and test details (e.g., data splits, hyperparameters, how they were chosen, type of optimizer) necessary to understand the results?
    \item[] Answer: \answerYes{} 
    \item[] Justification: Section~\ref{sec:exp_appendix}

\item {\bf Experiment statistical significance}
    \item[] Question: Does the paper report error bars suitably and correctly defined or other appropriate information about the statistical significance of the experiments?
    \item[] Answer: \answerNo{} 
    \item[] Justification: 

\item {\bf Experiments compute resources}
    \item[] Question: For each experiment, does the paper provide sufficient information on the computer resources (type of compute workers, memory, time of execution) needed to reproduce the experiments?
    \item[] Answer: \answerYes{} 
    \item[] Justification: Section~\ref{sec:exp_appendix}
    
\item {\bf Code of ethics}
    \item[] Question: Does the research conducted in the paper conform, in every respect, with the NeurIPS Code of Ethics \url{https://neurips.cc/public/EthicsGuidelines}?
    \item[] Answer: \answerYes{} 
    \item[] Justification: 

\item {\bf Broader impacts}
    \item[] Question: Does the paper discuss both potential positive societal impacts and negative societal impacts of the work performed?
    \item[] Answer: \answerYes{} 
    \item[] Justification: This is a theoretical work with no direct societal impacts, and any indirect effects are beyond the scope of this paper.
    
\item {\bf Safeguards}
    \item[] Question: Does the paper describe safeguards that have been put in place for responsible release of data or models that have a high risk for misuse (e.g., pre-trained language models, image generators, or scraped datasets)?
    \item[] Answer: \answerNA{} 
    \item[] Justification: 

\item {\bf Licenses for existing assets}
    \item[] Question: Are the creators or original owners of assets (e.g., code, data, models), used in the paper, properly credited and are the license and terms of use explicitly mentioned and properly respected?
    \item[] Answer: \answerYes{} 
    \item[] Justification: 

\item {\bf New assets}
    \item[] Question: Are new assets introduced in the paper well documented and is the documentation provided alongside the assets?
    \item[] Answer: \answerNA{} 
    \item[] Justification: 

\item {\bf Crowdsourcing and research with human subjects}
    \item[] Question: For crowdsourcing experiments and research with human subjects, does the paper include the full text of instructions given to participants and screenshots, if applicable, as well as details about compensation (if any)? 
    \item[] Answer: \answerNA{} 
    \item[] Justification: 

\item {\bf Institutional review board (IRB) approvals or equivalent for research with human subjects}
    \item[] Question: Does the paper describe potential risks incurred by study participants, whether such risks were disclosed to the subjects, and whether Institutional Review Board (IRB) approvals (or an equivalent approval/review based on the requirements of your country or institution) were obtained?
    \item[] Answer: \answerNA{} 
    \item[] Justification: 

\item {\bf Declaration of LLM usage}
    \item[] Question: Does the paper describe the usage of LLMs if it is an important, original, or non-standard component of the core methods in this research? Note that if the LLM is used only for writing, editing, or formatting purposes and does \emph{not} impact the core methodology, scientific rigor, or originality of the research, declaration is not required.
    \item[] Answer: \answerNA{} 
    \item[] Justification: 

\end{enumerate}